\documentclass[a4paper,11pt]{amsart}

\usepackage{amsmath}
\usepackage{amsfonts}  
\usepackage{amssymb}
\usepackage{amsthm}
\usepackage{mathtools}

\newtheorem{lemma}{Lemma}
\newtheorem{prop}{Proposition}

\theoremstyle{definition}
\newtheorem{definition}{Definition}

\theoremstyle{remark}
\newtheorem{stm}{Statement}

\newcommand{\R}{\mathbb{R}}
\newcommand{\Z}{\mathbb{Z}}
\newcommand{\N}{\mathbb{N}}

\newcommand{\C}{\mathbb{C}}

\newcommand{\CW}{\mathcal{CW}}
\newcommand{\WF}{\mathcal{WF}}
\newcommand{\bS}{\mathbb S}
\newcommand{\pr}{\mathrm{pr}}
\newcommand{\ind}{\mathrm{ind}}
\newcommand{\otheta}{\overline{\theta}}
\newcommand{\id}{\mathrm{id}}
\newcommand{\oL}{\overline{L}}
\newcommand{\oH}{\overline{H}}
\newcommand{\skel}{\mathrm{Skel}}
\newcommand{\cN}{\mathcal N}

\newcommand{\Ls}{\mathrm{L}}

\usepackage{subcaption}
\usepackage{tikz}
\usetikzlibrary{arrows.meta,bending,positioning,shapes.geometric}
\usetikzlibrary{quotes,angles}

\usepackage{algorithm}
\usepackage{algpseudocode}

\title{TILT: topological interface recovery in limited-angle tomography }

\author{Elli Karvonen}
\address{Department of Mathematics and Statistics, P.O. Box 68 (Pietari Kalmin katu 5), FI-00014 University of Helsinki, Finland}
\email{elli.karvonen@helsinki.fi}

\author{Matti Lassas}
\address{Department of Mathematics and Statistics, P.O. Box 68 (Pietari Kalmin katu 5), FI-00014 University of Helsinki, Finland}
\email{matti.lassas@helsinki.fi}

\author{Pekka Pankka}
\address{Department of Mathematics and Statistics, P.O. Box 68 (Pietari Kalmin katu 5), FI-00014 University of Helsinki, Finland}
\email{pekka.pankka@helsinki.fi}

\author{Samuli Siltanen}
\address{Department of Mathematics and Statistics, P.O. Box 68 (Pietari Kalmin katu 5), FI-00014 University of Helsinki, Finland}
\email{samuli.siltanen@helsinki.fi}

\begin{document}

\maketitle
\begin{abstract}
A novel reconstruction method is introduced for the severely ill-posed inverse problem of limited-angle tomography. It is well known that, depending on the available measurement, angles specify a subset of the wavefront set of the unknown target, while some oriented singularities remain invisible in the data. Topological Interface recovery for Limited-angle Tomography, or TILT, is based on lifting the visible part of the wavefront set under a universal covering map. In the space provided, it is possible to connect the appropriate pieces of the lifted wavefront set correctly using dual-tree complex wavelets, a dedicated metric, and persistent homology. The result is not only a suggested invisible boundary but also a computational representation for all interfaces in the target.
\end{abstract}

\tableofcontents

\section{Introduction}

\noindent
We study two-dimensional tomography where the aim is to recover a compactly supported, non-negative  function $f:\R^2\rightarrow \R$ from the knowledge of its line integrals. Our focus is on the highly ill-posed case called {\it limited-angle tomography}, where the collection of lines is restricted to a subset of directions.

We introduce a new algorithm called TILT, for Topological Interface recovery for Limited-angle Tomography. Our focus is on imaging targets with different homogeneous regions separated by crisp interfaces. In limited-angle tomography, parts of those interfaces are not encoded in data in a robust way, and we use topological data analysis with a dedicated non-Euclidean metric to connect the known parts of interfaces to the unknown parts.

We found that, under certain assumptions, persistent homology can be used to connect the visible object boundaries by filling in the invisible boundaries. Theoretically, a suitable set of prior knowledge is this: the target consists of a small number of inclusions, which 
\begin{itemize}
    \item are well separated,
    \item have smooth boundaries with bounded curvature,
    \item are either simply connected, or at least their boundary components are not close to each other.
\end{itemize}
Under these assumptions, and if the missing angle is small, our analysis shows that TILT can approximately recover the full boundaries of inclusions {\it and provide computational representations for those close curves.} Furthermore, our computational experiments suggest that TILT can be useful even in severly limited cases such as 60 degree view angle.

The idea of TILT is as follows. We assume that the above suitable priors hold for our target. Microlocal analysis tells us which parts of boundaries (or more precisely, which parts of the wavefront set of the unknown body $f$) can be reconstruct stably. We will compute these boundary points, i.e. microlocal singularities (in our case, interfaces between domains of constant material), using complex wavelets from an initial reconstruction. Moreover, complex wavelets give us information on the direction of these singularities. The idea is to estimate the neighborhoods of invisible singularities utilizing visible singularities and a dedicated non-Euclidean distance. Notably, we know that each boundary curve's wavefront set lifted to the space $\R^2\times \R$ forms a connected path. Here the first coordinate is the planar location, and the second coordinate  specifies the direction of a singularity.  However, only part of this path can be detected in a limited-angle problem. With the help of persistent homology, we can find the minimum size of invisible singularities' neighborhoods, so that these neighborhoods together with known parts of the path contain the full path, i.e., the boundary.

Limited-angle X-ray tomography appears in  medical applications, such as digital breast tomosynthesis \cite{dobbins2003digital,wu2004comparison,niklason1997digital,vedantham2015digital,rantala2006wavelet,piccolomini2016fast, landi2017limited, landi2019nonlinear}, intraoral dental imaging \cite{webber1997tuned,kolehmainen2003statistical,mauriello2020role}, and improving the temporal resolution of dynamic CT scans. Nondestructive testing in industry sometimes forces angle limitation, for example when the object is too large \cite{de2014industrial,kurfiss20123} or otherwise not accessible from all sides, as in weld inspection of pipelines \cite{haith2017defect,silva2021x} or underwater pipeline inspection \cite{riis2018limited}. As in medicine, one can attempt to increase the temporal resolution of dynamic scans by restricting the angle \cite[Sec. 2.1.3]{goethals2022dynamic}. Limited-angle tomography also appears in applications that do not use X-rays, such as 3D transmission electron microscopy \cite{engelhardt2000electron,turk2020promise}, neutron tomography \cite{osterloh2011limited}, satellite-based ionospheric tomography \cite{norberg2015ionospheric}, or visible light applications, including phase microscopy \cite{guo2021limited} and adaptive optics \cite{gerth2015method,helin2018atmospheric}. 

Since the limited-angle tomography problem is extremely ill-posed \cite{natterer2001mathematics,davison1983ill,louis1986incomplete}, the reconstruction process needs to be augmented with prior knowledge about the unknown to make it robust. However, the targets of imaging in the above applications are wildly different, and therefore it is necessary to have a variety of tomographic algorithms available to promote diverse kinds of prior information. 

Reconstruction from limited-angle tomography data has been approached in various ways in the literature, for example using algebraic methods \cite{andersen1989algebraic}, variational regularization \cite{delaney1998globally,sidky2006accurate,ritschl2011improved,frikel2013sparse,chen2013limited,huang2016new,huang2018scale} and Bayesian inversion \cite{hanson1983bayesian,kolehmainen2003statistical,rantala2006wavelet}. In recent years, machine learning has offered impressive new possibilities for reconstructions \cite{hammernik2017deep,hauptmann2018model,bubba2019learning,huang2019traditional,barutcu2021limited,huang2020limited,Bubba2021}. All of the methods are made robust against modelling errors and noise by incorporating prior information about the target. TILT is fundamentally different from all the previous work. 

Persistent homology has become a popular topic in the past two decades. Most of its applications are related to data analysis, see e.g. \cite{pershomMedicalBrodzki, pershomMedicalOyama, pershomPhysicsHamilton}. So far, it has not been used much for inverse problems, apart from the work of  Bubenik {\it et. al.} in solving constant curvature from uniformly sampled points on disks using  Vietoris–Rips complexes \cite{Bubenik_2020}. TILT uses  persistent homology quite differently. We propose a unique way to solve an inverse problem using cubical complexes formed by using a non-euclidean metric without analysing persistent diagrams, barcodes, or persistence landscapes. 
We believe that TILT may also find use in applications outside the context of tomography, for example in image inpainting. 

One of the novelties of our paper is that in the computational topological methods,
we use a non-Euclidean metric. Indeed,
to analyze the curves $\alpha(t)$ in $\R^2$, that are parametrized along the arc length, we first lift the curves to 
curves  $(\alpha(t),\partial_t\alpha(t))$ in  the unit sphere bundle
$S\R^2$. Then, to define a neighborhood of a lifted curve 
$\{(\alpha(t),\partial_t\alpha(t)):\ t\in [0,L]\}$, we define the topological neighborhoods in  $S\R^2=\{(x,v)\in \R^2\times \R^2:\ |v|=1\}$ using a distance
function that penalize curves that have a large curvature. To do this, we use
the elastic energy, originally defined by the Bernoullis and Euler \cite{elastica_Euler,elasticaHistory_Levien2008}, and show that such energy 
 defines a non-Euclidean metric in $S\R^2$. This metric
can be used to find moderately curved continuations for a given segment of a curve,
and using persistent homology we can consider continuations of the curve segment that are closed curves. Microlocal methods make it possible to determine a part $\gamma$  from tomographic data, and by using computational topology we find continuations of $\gamma$ which are moderately curved
closed curves. The union of such curves give us an estimate of the set where the true boundary possibly lies.

This paper is organized as follows. In Section \ref{sec:theoretical_background}, we go through the needed mathematical background, namely limited-angle tomography, wavefront set and its lift to $\R^2\times \R$, and the non-Euclidean distance called {\it candywrap distance}. Moreover, we go through all the needed computational tools, including complex wavelets, morphological operations, and persistent homology. In Section \ref{sec:method}, we describe the method, the process of estimating neighborhoods of invisible singularities, in detail. Next, we show our experimental results for chosen phantoms in Section \ref{sec:results}. Finally in Section \ref{sec:discussion}, we discuss the results and possible follow-up studies.

\section{Theoretical background} \label{sec:theoretical_background}

In this section, we discuss briefly the theoretical path from measurements to recognition of boundary components.

\subsection{Limited-angle tomography, continuous formulation}\label{sec:continuousRadon}

In 2-D tomography, we aim to recover a density function  $f\colon \R^2 \to \R$.  Typically, $f$ is a non-negative, compactly supported function modelling X-ray attenuation inside a physical body. Mathematically, tomographic data is determined by the Radon transform, which we will recall next. Let $\theta \in [0,2\pi]$ and let $p\in\R$. Denote $\otheta=(\sin\theta,\cos\theta)$ and $\theta^\perp=(-\sin\theta, \cos\theta)$. Now the line with normal vector $\theta^\perp$ and directed distance $p$ from the origin is denoted by
\begin{align}\label{Def:L_line}
    L(\otheta,p)=\{x \in \R^2 \mid x \cdot \otheta=p \}.
\end{align} 

The continuous two-dimensional Radon transform $\mathcal{R}f$ of a function $f \in L^1(\R^2)$ is 
 \begin{align*}
     \mathcal{R}f(\otheta,p)=\int_{ L(\otheta,p)}f(x)ds = \int_{-\infty}^\infty f(p\otheta+t\theta^\perp)dt,
 \end{align*}
where $ds$ is the measure on $L(\otheta,p)$ induced from Lebesgue measure on $\R^2$, see Figure \ref{fig:radon}. In a limited-angle tomography problem, tomographic data is given on some proper subset of $\theta \in [0,\pi)$; we focus on the case where $\theta$ is restricted between two extremes, $\theta_0$ and $\theta_1$:
\begin{equation}\label{thetarestr}
    0<\theta_0<\theta<\theta_1<\pi.
\end{equation}

Reconstructing $f$ from its Radon transform is an ill-posed problem, especially in the limited-angle case \cite{natterer2001mathematics,davison1983ill}.

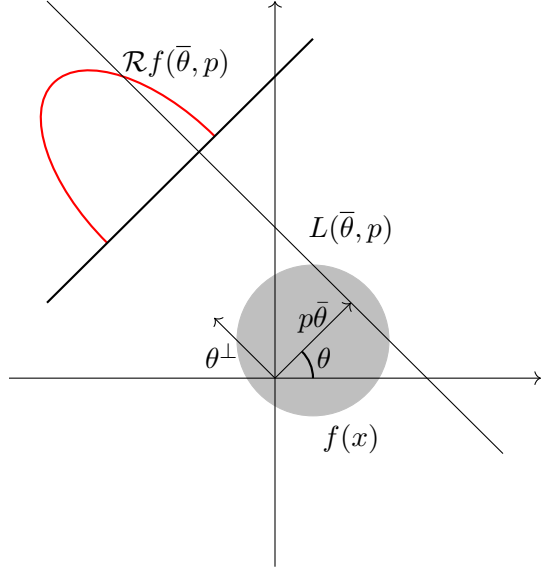
\begin{figure}
\centering
\begin{tikzpicture}
\draw[rotate around={45:(-1.5,2.5)},red,thick] (-1.5,2.5) ellipse (1 and 2); 
\draw[rotate around={45:(0,0)},fill,white] (-0.5,-0.5) rectangle (2,2.82);  
\draw[fill,color=lightgray] (0.5,0.5) circle [radius=1];  
\draw[->](0,-2.5) -- (0,5); 
\draw[->](-3.5,0) -- (3.5,0); 
\draw[-](3,-1) -- (-3,5); 
\draw[-,thick](-3,1) -- (0.5,4.5); 
\draw[->](0,0)-- (1,1); 
\draw[->](0,0) -- (-0.8,0.8); 
\draw
    (3,0) coordinate (a) 
    (0,0) coordinate (b)
    (2,2) coordinate (c)
    pic["$\theta$", draw=black, thick, -, angle eccentricity=1.4, angle radius=0.5cm]
    {angle=a--b--c};
\node[] (p) at (0.5,0.8) {$p\Bar{\theta}$};  
\node[] (t) at (-0.7,0.3) {$\theta^\perp$};  
\node[] (L) at (1,2) {$L(\otheta,p)$};
\node[] (f) at (1,-0.8) {$f(x)$};
\node[] (r) at (-1.3,4.2) {$\mathcal{R}f(\otheta,p)$};
\end{tikzpicture}
\caption{Illustration of a geometrical setup of the Radon transform of $f(x)$.}
\label{fig:radon}
\end{figure}

\subsection{Limited-angle tomography, discrete formulation}

While the continuous formulation of Section \ref{sec:continuousRadon} is convenient for theoretical purposes, in practical computations we prefer a discrete setup. This is especially crucial with incomplete tomographic data, such as the limited-angle case.

We write the reconstruction problem in matrix form:
\begin{align}\label{basicdiscreteinverseproblem}
m = Af+\epsilon,
\end{align}
where $m$ is data, $A$ is a known discretized linear forward operator, $f$ is a density function we aim to solve, and  $\epsilon$ models random noise. Typically, $f$ is represented as a pixel image with a constant attenuation value in each pixel, and the matrix $A$ encodes the lengths of X-ray paths travelling inside each pixel. It is also possible to model some finite width for the ray in this so-called {\it pencil beam model}. See e.g. \cite[Section 2.3.4]{mueller2012linear}.

The ill-posedness of the continuous model in Section \ref{sec:continuousRadon} shows up in the discrete model in the form of a ridiculously large condition number of the matrix $A$. This hampers any attempt to solve (\ref{basicdiscreteinverseproblem}) using basic methods, such as least squares, as the measurement noise is amplified uncontrollably \cite[Figure 2.19]{mueller2012linear}.

There exist multiple different strategies to recover the function $f$ using model (\ref{basicdiscreteinverseproblem}). One well-known and widely used strategy is total variation (TV) regularization, first introduced in \cite{rudin1992nonlinear} for image noise removal. The TV regularized problem takes the form 
 \begin{align*}
     \min_f \{ \| Af-m \|_2^2 + \alpha \text{TV}(f)\},
 \end{align*}
 where $\alpha$ is a regularization parameter and TV$(x)$ is a regularization term defined informally as $\int_\Omega |\nabla f(x)|dx$. The result in TV regularization is typically piecewise constant with sharp edges. In this work we use an adaptation of the algorithm in \cite{bredies2014recovering} to tomography.

\subsection{Wavefront set} \label{sec:wavefrontset}

In this section, we define the wavefront set $\WF(D)$ of a bounded subset $D\subset \R^n$ as the wavefront set of the distribution induced by the characteristic function $\chi_D$ of $D$. For the definition, we first give related definitions. Recall that a \emph{distribution $g \in \mathcal{D'}(\R^n)$} is a continuous linear functional $g\colon C^{\infty}_0(\R^n) \to \R$.

A distribution $g\in \mathcal{D'}(\R^n)$ has a \emph{compact support} if there exists a compact set $K\subset\R^n$ having the property that $g(\Phi)=0$ for each function $\Phi\in C^{\infty}_0(\R^n\setminus K)$. Also, a distribution's $g\in \mathcal{D'}(\R^n)$ Fourier transformation \emph{decreases rapidly} if, for each $N\in\N$, there exists a constant $C_N$ for which  $|\mathcal{F}g(\xi)| \leq C_N(1+|\xi|)^{-N}$ for all $\xi\in\R^n$; here $\mathcal{F}g$ is the Fourier transform of $g$. Recall that a distribution $g$ with compact support is equal to a $C^{\infty}$ function almost everywhere if and only if its Fourier transform $\mathcal{F}$ decreases rapidly.

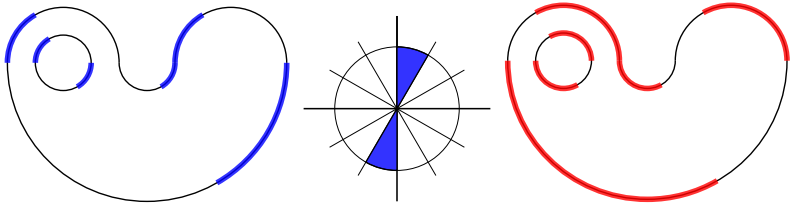
\begin{figure}

    \begin{subfigure}[b]{0.3\textwidth}
    \resizebox{\textwidth}{!}{%
    \begin{tikzpicture}
    \draw[black, line width=0.5mm, opacity=1] (-4,0) arc[start angle=0, end angle=180, radius=2]; 
    \draw[black, line width=0.5mm, opacity=1] (-2,0) arc[start angle=0, end angle=-180, radius=1]; 
    \draw[black, line width=0.5mm, opacity=1] (-6,0) circle [radius=1]; 
    \draw[black, line width=0.5mm, opacity=1] (2,0) arc[start angle=0, end angle=180, radius=2]; 
    \draw[black, line width=0.5mm, opacity=1] (2,0) arc[start angle=0, end angle=-180, radius=5]; 
    \draw[blue, line width=2mm, opacity=0.8] (-8,0) arc[start angle=180, end angle=120, radius=2]; 
    \draw[blue, line width=2mm, opacity=0.8] (-2,0) arc[start angle=180, end angle=120, radius=2]; 
    \draw[blue, line width=2mm, opacity=0.8] (-2,0) arc[start angle=0, end angle=-60, radius=1]; 
    \draw[blue, line width=2mm, opacity=0.8] (2,0) arc[start angle=0, end angle=-60, radius=5]; 
    \draw[blue, line width=2mm, opacity=0.8] (-7,0) arc[start angle=180, end angle=120, radius=1]; 
    \draw[blue, line width=2mm, opacity=0.8] (-5,0) arc[start angle=0, end angle=-60, radius=1]; 
    \end{tikzpicture}
    }%
    \end{subfigure}
      \begin{subfigure}[b]{0.2\textwidth}
    \resizebox{\textwidth}{!}{%
    \begin{tikzpicture}
    \draw[black, line width=0.5mm, opacity=1] (-3,0)--(3,0); 
    \draw[black, line width=0.5mm, opacity=1] (0,-3)--(0,3); 
    \draw[black, line width=0.5mm, opacity=1] (0,0)--(0,3); 
    \draw[] (0,0) -- (30:2.5);
    \draw[] (0,0) -- (60:2.5);
    \draw[] (0,0) -- (120:2.5);
    \draw[] (0,0) -- (150:2.5);
    \draw[] (0,0) -- (210:2.5);
    \draw[] (0,0) -- (240:2.5);
    \draw[] (0,0) -- (300:2.5);
    \draw[] (0,0) -- (330:2.5);
    \draw[very thick,fill=blue!80] (0,0) --  (90:2) arc(90:60:2) -- cycle;
    \draw[very thick,fill=blue!80] (0,0) --  (270:2) arc(270:240:2) -- cycle;
    \draw[] (0,0) circle [radius=2];    
    \end{tikzpicture}
    }%
    \end{subfigure}
    \begin{subfigure}[b]{0.3\textwidth}
    \resizebox{\textwidth}{!}{%
    \begin{tikzpicture}
    \draw[black, line width=0.5mm, opacity=1] (-4,0) arc[start angle=0, end angle=180, radius=2]; 
    \draw[black, line width=0.5mm, opacity=1] (-2,0) arc[start angle=0, end angle=-180, radius=1]; 
    \draw[black, line width=0.5mm, opacity=1] (-6,0) circle [radius=1]; 
    \draw[black, line width=0.5mm, opacity=1] (2,0) arc[start angle=0, end angle=180, radius=2]; 
    \draw[black, line width=0.5mm, opacity=1] (2,0) arc[start angle=0, end angle=-180, radius=5]; 
    \draw[red, line width=2mm, opacity=0.8] (-4,0) arc[start angle=0, end angle=120, radius=2]; 
    \draw[red, line width=2mm, opacity=0.8] (2,0) arc[start angle=0, end angle=120, radius=2]; 
    \draw[red, line width=2mm, opacity=0.8] (-4,0) arc[start angle=180, end angle=300, radius=1]; 
    \draw[red, line width=2mm, opacity=0.8] (-8,0) arc[start angle=180, end angle=300, radius=5]; 
    \draw[red, line width=2mm, opacity=0.8] (-5,0) arc[start angle=0, end angle=120, radius=1]; 
    \draw[red, line width=2mm, opacity=0.8] (-7,0) arc[start angle=180, end angle=300, radius=1]; 
    \end{tikzpicture}
    }%
    \end{subfigure}
    \caption{Visible singular, bolded in blue (illustrated on the left-hand side), and invisible singularities, bolded in red (illustrated on the right-hand side), of an object whose boundaries are drawn black when X-rays span from 60  to 90 degrees (illustrated in middle).}
    \label{fig:VisVsInvis}
\end{figure}

We are now ready to define the wavefront set of $D$ as follows. 

\begin{definition}
The \emph{(unit) wavefront set $\WF_1(D)\subset \R^n\times \bS^{n-1}$} of a bounded set $D\subset \R^2$ consists of the pairs $(x,\xi)$ that have the following property: for each function $\Phi \in C^\infty_0(\R^n)$ satisfying $\Phi(x)\ne 0$, the Fourier transform $\mathcal F(\Phi g)$ does not decrease rapidly in any open cone neighborhood of $\{ x + t \xi \in \R^n \mid t>0\}$.
\end{definition}

The wavefront set plays an important role in X-ray tomography: we know precisely which singularities of a density function $f$ can be detected and stably recovered \cite{greenleaf1989non,Quinto1993,frikel2013sparse,FrikelQuinto}.  Let $D\subset \R^n$  be a bounded set. Consider that we have X-ray data for angles in a set $(\theta_0,\theta_1)$ as in (\ref{thetarestr}). Now we know that $(x_0,\xi_0) \in \WF_1(D)$ can be detected if there exist $\theta\in(\theta_0,\theta_1)$ for which the line $L(\otheta,|x_0|)$ is perpendicular to $\xi_0$. Notably, we have the X-ray data for angles in a set $(\theta-\varepsilon,\theta+\varepsilon)\subset(\theta_0,\theta_1)$, which is essential to be able to detect singularity. We say that a singularity is \emph{visible} if it can be detected from X-ray data. Otherwise, it is \emph{invisible}. Since in limited-angle data we do not have all the directions, some singularities might be invisible (Figure \ref{fig:VisVsInvis}).

\subsection{Lift of a wavefront set under universal covering map}

In this section, we prove a useful lifting property for the wavefront set. For the statement, let $\widehat \pi \colon \R\to \bS^1$ be the standard covering map $t\mapsto e^{i t}$ and denote $\widetilde \pi = \id_{\R^2}\times \widehat \pi \colon \R^2\times \R\to \R^2\times \bS^1$ the covering map $(x,t) \mapsto (x,e^{i t})$.

\begin{prop}
\label{prop:lifting-proposition}
Let $\Gamma \subset \R^2$ be a smooth loop, $\beta \colon \bS^1 \to \Gamma$ a diffeomorphism, and $\alpha = \beta \circ \hat\pi \colon \R\to \R^2$ be a $C^\infty$-smooth parametrization of $\Gamma$. Then 
\begin{enumerate}
\item $\sigma \colon \R\to \R^2\times \bS^1$, 
\[
t\mapsto \left(\alpha(t), \left( \frac{\dot \alpha(t)}{|\dot \alpha(t)|}\right)^\perp\right),
\]
is a path into $\cN_1(\Gamma)$, and
\item each lift $\tilde \sigma \colon \R \to \R^2\times \R$ under $\tilde \pi$ meets all levels of $\R^2\times \{t\}$ for $t\in \R$.
\end{enumerate}
\end{prop}

The heuristic interpretation of this proposition is that we may detect boundary components $\Omega$ from lifts of components $\Gamma$ of the wavefront sets $\WF_1(\Omega)$.

The main tool in the proof of Proposition \ref{prop:lifting-proposition} is the notion of a unit normal bundle of a smooth loop $\Gamma \subset \R^2$. By assumption, $\Gamma$ is diffeomorphic to $\bS^1$. For each $x\in \Gamma$, the tangent space of $\Gamma$ is a line $T_x(\Gamma)$ in $T_x \R^2$ and we denote $N_x \subset T_x \R^2$ the line orthogonal to $T_x(\Gamma)$. Thus $N_x$ is the \emph{normal of $\Gamma$ at $x$}. As usual, we consider $T_x \R^2$ as the space of directions at $x$ and write $T_x \R^2 = \{x\}\times \R^2$. We denote
\[
\cN_1(\Gamma) = \{ (x,\nu) \in \R^2\times \bS^1 \colon x\in \Gamma,\ \nu\in N_x,\ |\nu|=1\}
\]
the \emph{unit normal bundle of $\Gamma$}.

The unit normal bundle $\cN_1(\Gamma)$ has two components, which we denote -- in what follows -- as $\Gamma_+$ and $\Gamma_-$. The symbols $+$ and $-$ here refer to the heuristics that, if $\Omega$ is a domain bounded by $\Gamma$, then we may take $\Gamma_+$ to correspond the outward pointing normals of $\Gamma = \partial \Omega$ and $\Gamma_-$ the inward pointing normals. The projection $\pr_\Gamma \colon \cN_1(\Gamma) \to \Gamma$, $(x,\xi) \mapsto x$, restricts to an isomorphism on both components $\Gamma_+$ and $\Gamma_-$. 

 Since the components of the wavefront set $\WF_1(\Omega)$ are exactly the components $\Gamma^+$ and $\Gamma^-$ of $\cN_1(\Gamma)$ for boundary components $\Gamma \subset \partial \Omega$ and 
 \[
 \left( \frac{\dot \alpha(t)}{|\dot \alpha(t)|}\right)^\perp \in \WF_1(\Omega),
 \]
Proposition \ref{prop:lifting-proposition} follows immediately from the following proposition.

\begin{prop}
\label{prop:lifting}
Let $\Gamma \subset \R^2$ be a $C^\infty$-smooth loop in $\R^2$ and $(x_0,\xi_0)\in \cN_1(\Gamma)$. Then there exists a path $\gamma = (\gamma_1,\gamma_2) \colon \R\to \R^2\times \R$ having the following properties:
\begin{enumerate}
\item $\gamma_1(\R) = \Gamma$ and $\gamma_2(t+2\pi) = \gamma_2(t) + 2\pi$ for each $t\in \R$, \label{item:lifting-1}
\item  $(x_0,\xi_0) \in (\widetilde \pi \circ \gamma)(\R)$,\label{item:lifting-2}
\item $(\widetilde \pi \circ \gamma)(t) \in  \cN_1(\Gamma)$ and $(\widetilde \pi \circ \gamma)(t+2\pi) = (\widetilde \pi \circ \gamma)(t)$ for each $t\in \R$.\label{item:lifting-3}
\end{enumerate}
Moreover, the path $\gamma$ is essentially unique in the sense that, if $\gamma' \colon \R \to \R^2\times \R$ is another curve satisfying the same conditions, then there exists $j \in \N$ for which $\gamma'(t) = \gamma(t + j2\pi)$ for every $t\in \R$. Also, if $\gamma' = (\gamma'_1, \gamma'_2) \colon \R \to \R^2\times \R$ is a curve satisfying the same conditions, but for $(x_0,-\xi_0)\in \cN_1(\Gamma)$, then there exists $j \in \N$ for which $\gamma'_1(t) = \gamma'_1(t)$ and $\gamma_2(t) = \gamma'_2(t + \pi + j2\pi)$ for every $t\in \R$.
\end{prop}

Before the proof, we introduce the following notation. Let $\pi_{\bS^1} \colon \R^2\times \bS^1\to \bS^1$ be the projection to the second coordinate and denote 
$\iota^+ \colon \Gamma \to \bS^1$ and $\iota^- \colon \Gamma \to \bS^1$ the maps $\iota^+ = \pi_{\bS^1} \circ (\pi_\Gamma|_{\Gamma_+})^{-1}$ and $\iota^- = \pi_{\bS^1}\circ (\pi_\Gamma|_{\Gamma_-})^{-1}$. 

The proof is based on the fact that the mappings $\iota^\pm$ are so-called {\it Gauss maps} of $\Gamma$ and they have the following property (see e.g.~do Carmo \cite[Theorem 5-7.2]{DoCarmo}): 
\begin{quote}
The mapping $\iota^\pm \colon \Gamma \to \bS^1$ has degree $\pm 1$, where the sign depends on the orientation of $\Gamma$. In particular, the mappings $\iota^\pm$ are not null homotopic.
\end{quote}

\begin{proof}[Proof of Proposition \ref{prop:lifting}]
Let $\sigma \colon \bS^1 \to \R^2$ be an embedding onto $\Gamma$, that is, $\sigma(\bS^1) = \Gamma$ and the restriction of $\sigma$ to its image is a homeomorphism. We denote $\alpha^\pm = (\sigma, \iota^\pm \circ \sigma) \colon \bS^1 \to \R^2\times \bS^1$. Note that $\alpha^\pm(\theta) = (\sigma(\theta), \iota^\pm(\sigma(\theta)))\in \cN_1(\Gamma)$ for each $\theta\in \bS^1$. Moreover, $\alpha^+(\bS^1) \cap \alpha^-(\bS^1) = \emptyset$. Let now $\widetilde \alpha^\pm = (\widetilde \alpha^\pm_1, \widetilde \alpha^\pm_2) \colon \R\to \R^2 \times \R$ be a lift of $\alpha \circ \widehat \pi$ under $\widetilde \pi$. Condition \eqref{item:lifting-3} now holds for both paths $\widetilde \alpha^\pm$.

To satisfy condition \eqref{item:lifting-2}, we now choose $\gamma=(\gamma_1,\gamma_2)$ to be the path $\widetilde \alpha^\pm$ for which $(x_0,\xi_0)$ belongs to the image of $(\sigma,\nu^\pm)$. We may assume that $\gamma = \widetilde \alpha^+$.

Since $\widetilde \pi \circ \widetilde \alpha^\pm = \alpha^\pm = (\sigma, \nu^\pm)$, we have that $\gamma_1(\R) = \Gamma$. To satisfy the second part of \eqref{item:lifting-1}, we observe that the difference $\gamma_2(t+2\pi) - \gamma_2(t)$ is determined by the degree of the map $\iota^+ \colon \Gamma \to \bS^1$. Thus $|\gamma_2(t+2\pi) - \gamma_2(t)|=2\pi$. Hence, by changing the direction of $\gamma$ if necessary, we have that $\gamma_2(t+2\pi) = \gamma_2(t) + 2\pi$. Condition \eqref{item:lifting-1} now holds for $\gamma$.

If $\gamma' \colon \R \to \R^2\times \R$ is another curve satisfying the same conditions, then it is a lift of $\widetilde \pi \circ \gamma$ under covering map $\widetilde \pi$. Thus the claim follows by uniqueness of lifts.

Finally, suppose that $\gamma' = (\gamma'_1, \gamma'_2) \colon \R \to \R^2\times \R$ is a curve satisfying the same conditions for $(x_0,-\xi_0)$. Since $\widehat \pi(t+\pi) = - \widehat \pi(t)$ for each $t\in \R$, we have that $|\gamma_2(t)-\gamma'_2(t)|=\pi$ for each $t\in \R$. The claim follows again from the uniqueness of lifts. 
\end{proof}

Figure \ref{fig:lift} illustrates a lift of a wavefront set under universal covering map.
We finish this section with a formulation of Proposition \ref{prop:lifting} in terms of tangent vectors of parametrizations, which will be used in the forthcoming sections.

\begin{figure}
    \centering
     \begin{subfigure}[t]{0.45\textwidth}
        \includegraphics[width=\textwidth]{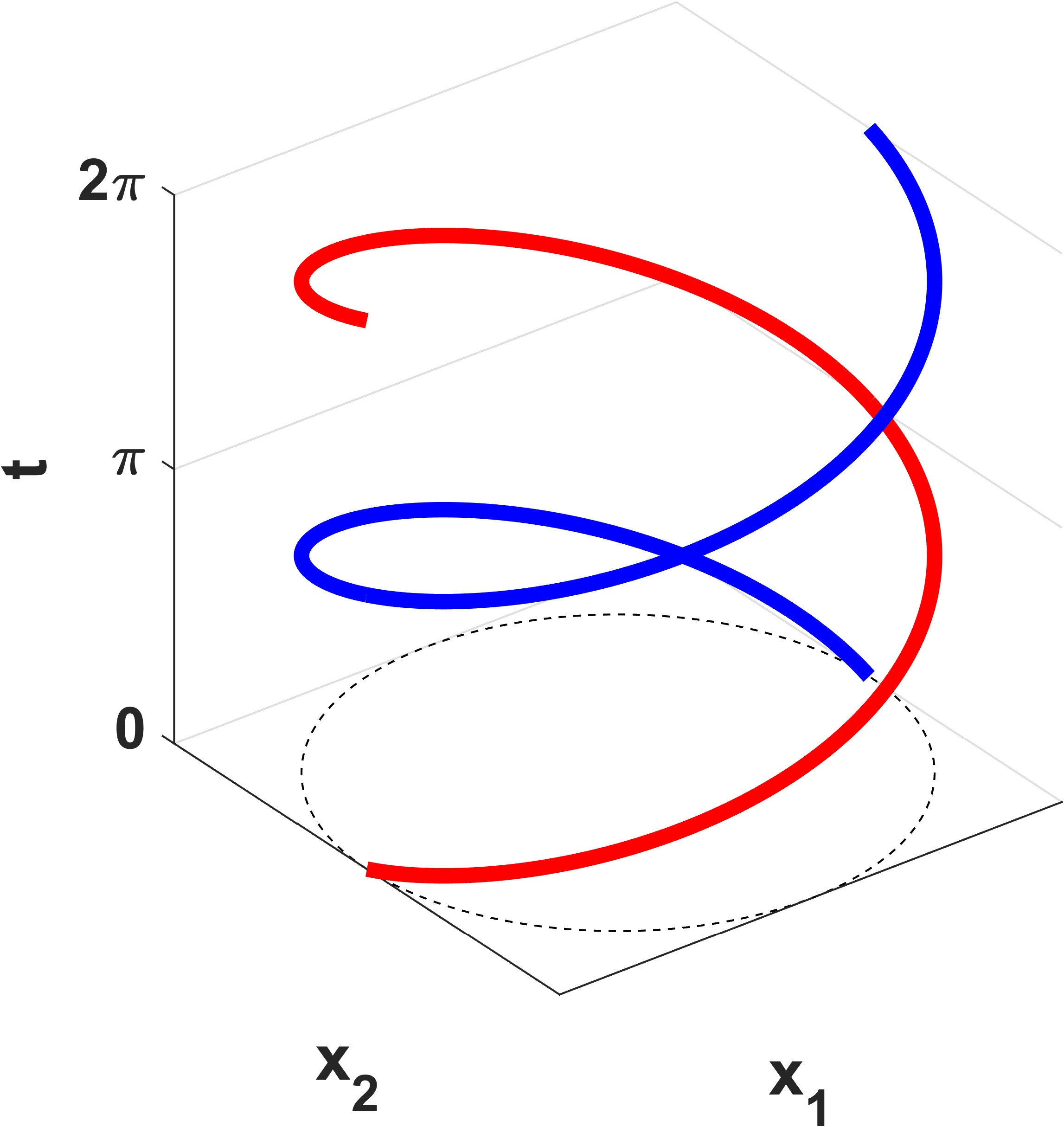}
    \caption{Here $\Gamma$ is a boundary of unit disk. The blue path illustrates the path $\gamma_+$, i.e., the lift of outwards pointing singularities and the red path illustrates the path $\gamma_-$.}
    \end{subfigure}
        \hspace{5mm}
    \begin{subfigure}[t]{0.45\textwidth}
        \includegraphics[width=\textwidth]{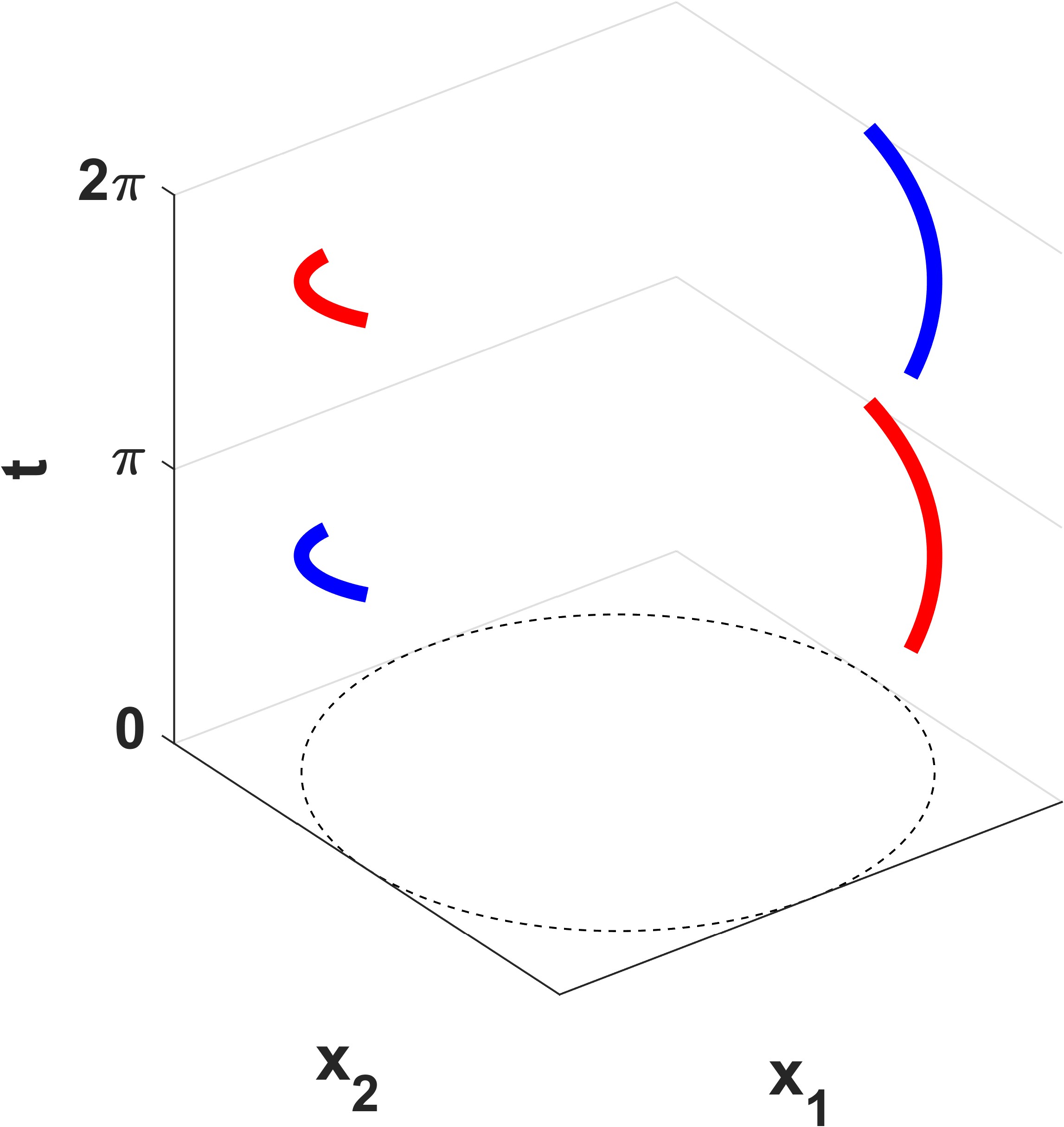}
    \caption{In a limited-angle problem, when measured angles are in a set $(\theta_0,\theta_1)$, we can recover two disjoint parts of the path $\gamma_\pm$.   }
    \end{subfigure}
    \begin{subfigure}[t]{0.45\textwidth}
        \includegraphics[width=\textwidth]{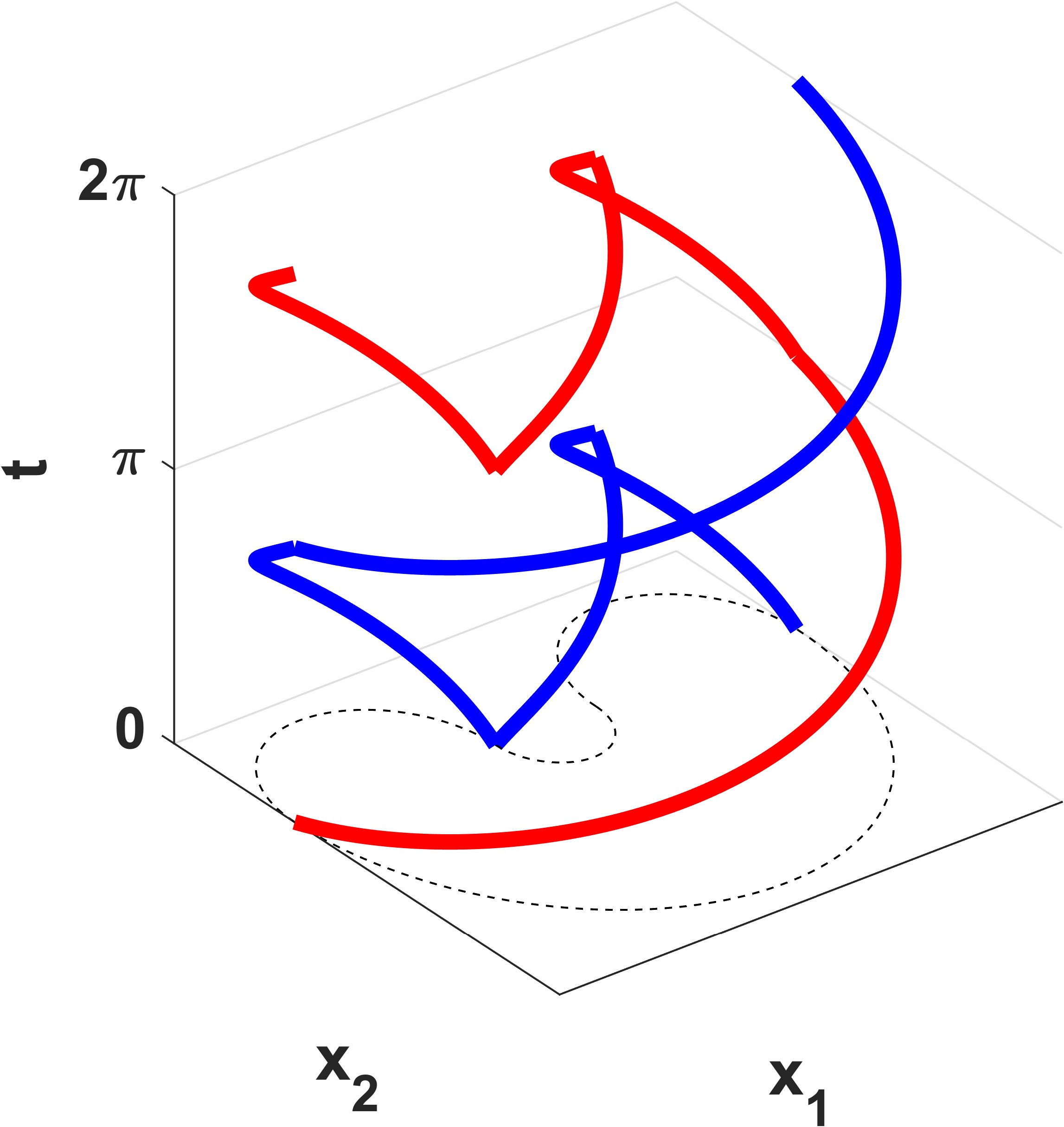}
    \caption{Here $\Gamma$ is the boundary of the non-convex set with smooth boundary. The blue path illustrates the path $\gamma_+$, and the red path illustrates the path $\gamma_-$.}

    \end{subfigure}
    \hspace{5mm}
    \begin{subfigure}[t]{0.45\textwidth}
        \includegraphics[width=\textwidth]{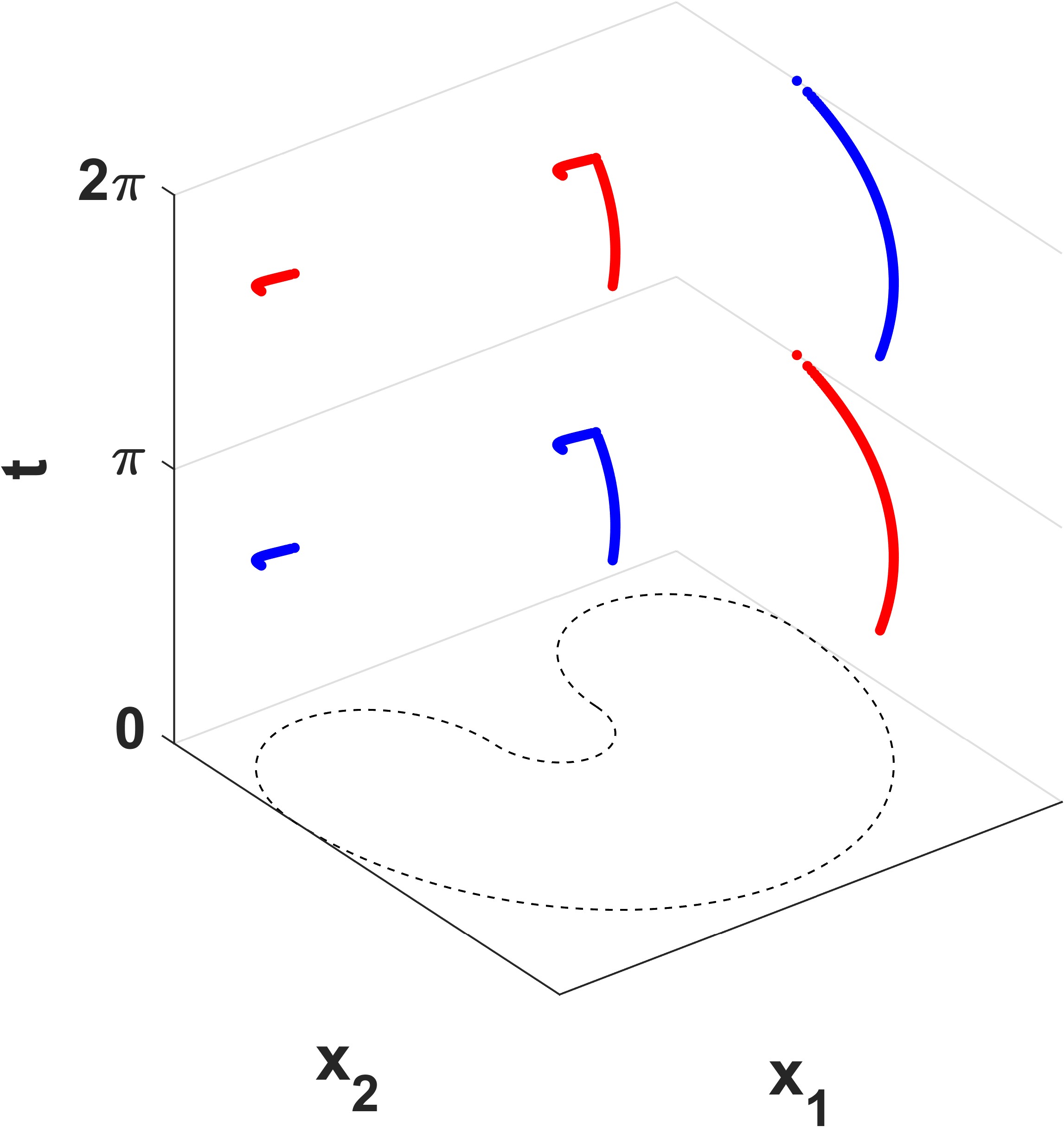}
    \caption{In a limited-angle problem, when measured angles are in a set $(\theta_0,\theta_1)$, we can recover multiple disjoint parts of the path $\gamma_\pm$.}
    \end{subfigure}
    \caption{
        The wavefront set of a smooth boundary component $\Gamma$ has two components, $\Gamma^+$ and $\Gamma^-$, i.e., the set of outward pointing singularities and the set of inward pointing singularities. These sets lifted to the space $\R^2\times \R$ under universal covering map form paths  $\gamma_+:\R \to \R^2 \times \R$ and  $\gamma_-:\R \to \R^2 \times \R$. Only parts of these paths can be recovered stably from limited-angle data. The inverse problem that we consider is to recover the boundary component from parts of the paths.}
       \label{fig:lift}
\end{figure}

\subsection{Candywrap distance} \label{sec:cw_dist}

In this section, we will introduce a tailored non-Euclidean metric on the unit sphere bundle $S\mathbb R^2$ motivated by elasticas. Elasticas are plane curves which minimizes the bending $p$-energy $E_p[\alpha]$, defined as 
$$
E_p[\alpha]=\int_0^L |K_\alpha(t)|^p|\partial_t\alpha(t)|dt,
$$ 
where $K_\alpha(t)$ is the curvature of $\alpha(t)$ and 

$p\ge 1$.
The problem of minimizing $E_p[\alpha]$ have been studied already in the 18th century by the Bernoullis and Euler \cite{elastica_Euler,elasticaHistory_Levien2008}. Minimizing
the energy $E_p[\alpha]$ is non-convex problem  \cite{elastica_Miura2020}. In the case $p=2$,  Mumford \cite{Elastica_Mumford1994} 
showed that the minimizers $\alpha$ of $E_p[\alpha]$ can be represented as $\vartheta$-functions
associated to periodic lattices. However, this gives only a necessary but not a sufficient condition for the minimizers of  $E_p[\alpha]$. Energy functions similar to  $E_p[\alpha]$ 
have found application in various fields, for example, in shape processing and image inpainting, see results of Chambolle and Pock \cite{elastica_Antonin_Pock2019}. In our case, elasticas have inspired us to define a metric that meaningfully estimates the neighborhoods of invisible singularities. In particular, we show that the infimum of $(E_2[\alpha]+|\alpha|^2)^{1/2}$ over the curves connecting points $(x_1,v_1)$ and $(x_2,v_2)$ defines a metric  in  $S\mathbb R^2$. Later in this section, we introduce an approximate version of this  distance function which 
leads to convenient numerical algorithms.

For building intuition behind the upcoming metric, let us consider a line segment. We know that from endpoints the boundary continues smoothly with normal directions between interval $[0,\pi/2]$. For the moment, we do continuations using arcs of circles. However, we cannot know the correct radius for a continuation. Thus we take all radii between a chosen interval $[a,b]\subset \R$. These continuations form a shape that looks like a twisted candy wrapper, see Figures \ref{fig:candywrap_shape_A}, \ref{fig:candywrap_shape_B}. When known parts of the boundary are arcs with the same center point and the same but unknown radius, then arcs with the same radius would give perfect continuations, see Figure \ref{fig:candywrap_circle_fill}. These candywrap-like shapes were the original inspiration for introducing a metric which defines neighborhoods containing various, not only arc-like, smooth curves. The approximation of this metric is honored with being called the name candywrap distance in accordance with the original observation.

\begin{figure}[!ht]
\begin{subfigure}[t]{0.45\textwidth}
    \centering
   \includegraphics[width=\textwidth]{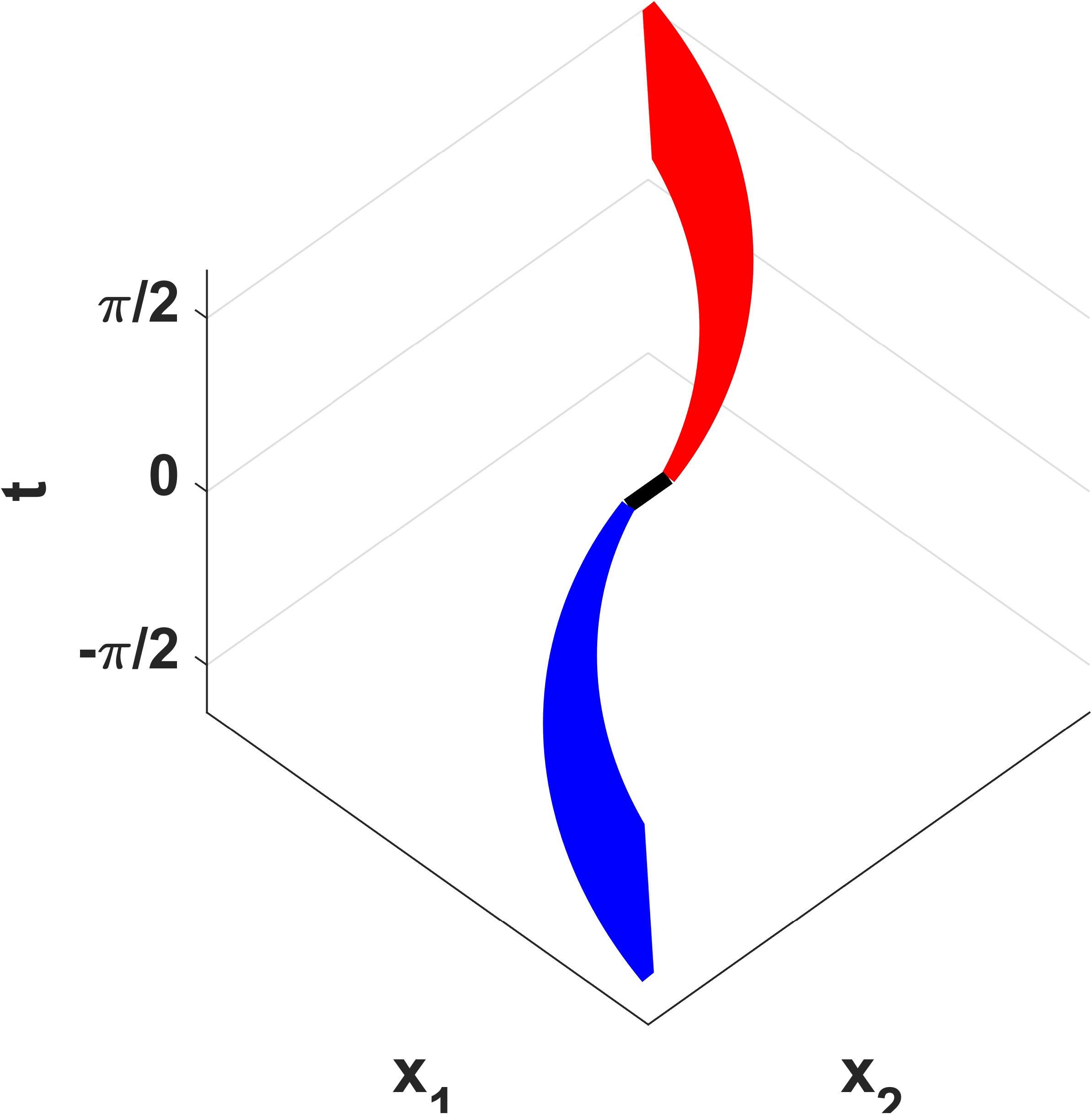}
   \caption{The boundary continuations from a line segment with specific normal directions and arc-like curves form a candywrap-like shape. Here are all arcs which have radius in the interval $[5,10]$. }
    \label{fig:candywrap_shape_A}
\end{subfigure}
\hspace{5mm}
\begin{subfigure}[t]{0.3\textwidth}
    \centering
    \includegraphics[width=\textwidth]{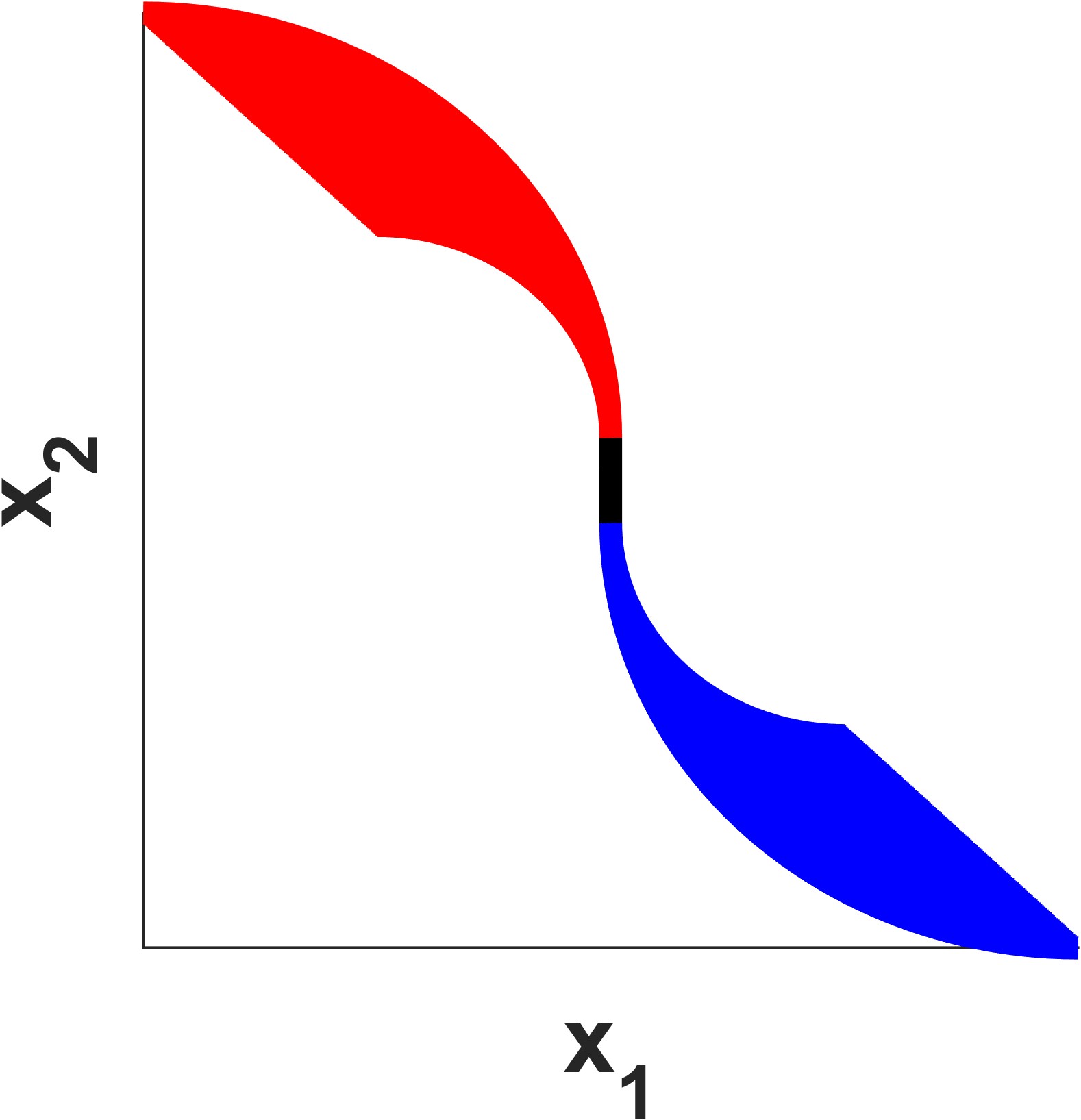}
    \caption{A candywrap-like shape viewed from a different direction.} 
        \label{fig:candywrap_shape_B}
\end{subfigure}
\begin{subfigure}[]{0.5\textwidth}
    \centering
    \includegraphics[width=\textwidth]{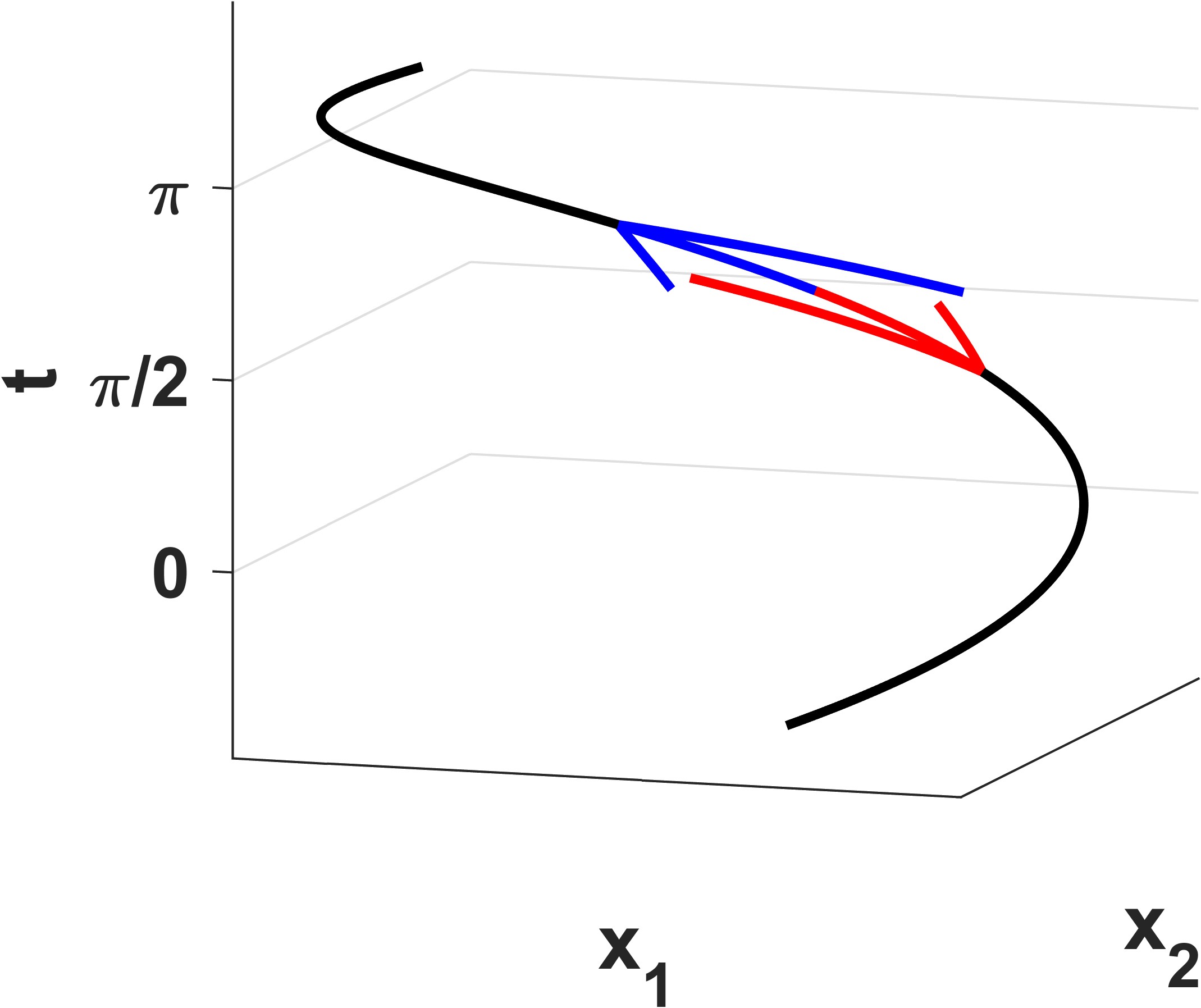}
    \caption{Two arcs with the same center point and radius, here illustrated in black, can be connected perfectly with continuation from endpoints using arcs with the same radius. Continuations using arcs with different radii do not meet each other. Here continuations are illustrated in blue and red.} 
        \label{fig:candywrap_circle_fill}
\end{subfigure}
    \caption{Inspiration behind the metric in Definition \ref{def:theor_cw_dist} and the approximation of this metric are continuations done by arcs with different radii.}
\end{figure}

We start by recalling that a curvature of a curve $\alpha\colon [t_0,t_1]\to \R^2$ is 
\begin{align*}
    K_\alpha(t)=\frac{\det(\alpha'(t),\alpha''(t))}{\mid\alpha'(t)\mid^3}.
\end{align*} 
Recall also that a path $\alpha \in C^1([a,b],\R^2)$ is \emph{piecewise $C^2$-smooth} if there exists a partition $a=t_0 <t_1 < \cdots <t_k-1 < t_k$=b of the interval $[t_0,t_k]$ for which $\alpha|_{[t_{j-1},t_j]}\in C^2([t_{j-1}, t_j])$ for each $j=1,\ldots, k$. For $(p_0,v_0),(p_1,v_1)\in S\R^2$, we denote $X[(p_0,v_0),(x_1,v_1)]$ the family of all piecewise $C^2$-smooth paths $\alpha \in C^1([a,b],\R^2)$ for which $\alpha(a)=p_0$, $\alpha(b)=p_1$, $\alpha'(a)=v_0$, and $\alpha'(b)=v_1$.
\begin{definition}
\label{def:theor_cw_dist}
We define a distance on the sphere bundle $S\R^2$ as follows: 

\begin{align*}
    D\big((p_0,v_0),(p_1,v_1)\big)=\inf_{\alpha \in X[(p_0,v_0),(p_1,v_1)]} \sum_{j=1}^{k}\int_{t_{j-1}}^{t_{j}} (|K_\alpha(t)|^2+1)|\alpha'(t)| \,dt,
\end{align*}
when $(p_0,v_0)\neq (p_1,v_1)$, and
\begin{align*}
     D\big((p_0,v_0),(p_1,v_1)\big) = 0,
\end{align*}
when $(p_0,v_0) = (p_1,v_1)$.
\end{definition}

The minimizers $\alpha$ in the definition of the function
$D\big((p_0,v_0),(p_1,v_1)\big)$ can be considered as generalized Eulerian elasticas
as the minimized cost function is the sum of the Eulerian $2$-energy $E_2[\alpha]$ and  the sum of the squared of the length of the curve.
Next we show that the function $D\big((p_0,v_0),(p_1,v_1)\big)$ defines a metric on the unit sphere bundle of $\mathbb R^2$.

\begin{prop}
    The distance $D:S\R^2 \to \R$, defined in Definition \ref{def:theor_cw_dist}, is a metric in $S\R^2$.
\end{prop}
\begin{proof}

Clearly $D((p_0,v_0),(p_1,v_1)) \geq 0$. 
Suppose that $\alpha \in X[(p_0,v_0),(p_1,v_1)]$.
Note that the distance $D$ does not depend on parametrization of a curve $\alpha$. Thus without loss of generality, we may assume that $|\alpha'(t)|=1$. That implies $\alpha'\perp \alpha''$. Furthermore $K_\alpha(t)=|\alpha''(t)|$.

Consider a path $\alpha\in C^1([0,L])$, parametrized along the path length,
such that $\alpha|_{[t_{j-1},t_j]}\in C^2([t_{j-1}, t_j])$, $t_0=0$, $t_k=L$, $q_j=\alpha(t_j)$, $j=0,1,\dots,k$, $\alpha(t_0)=p_0$, and $\alpha(t_k)=p_1$. Then 
\begin{align*}
\sum_{j=1}^k|q_{j}-q_{j-1}| &= |q_{k}-q_{k-1}| + |q_{k-1}-q_{k-2}| +\cdots + |q_{1}-q_{0}| \\
 &\geq |q_k-q_0|.
\end{align*}
Thus \begin{align*}
    \sum_{j=1}^{k}\int_{t_{j-1}}^{t_j} 1 \,dt &=    \sum_{j=1}^{k}\int_{t_{j-1}}^{t_j} |\alpha'(t)| \,dt 
     \geq \sum_{j=1}^{k}|\int_{t_{j-1}}^{t_j} \alpha'(t) \,dt|\\ &
     = \sum_{j=1}^{k}|\alpha(t_j)-\alpha(t_{j-1})|=\sum_{j=1}^{k}|q_j-q_{j-1}| \\
     &\geq |q_k-q_0| = |p_1-p_0|.
\end{align*}
Denote $w_j=\alpha'(t_j)$, so that $\alpha'(t_0)=v_0$ and $\alpha'(t_k)=v_1$.
Now,
\begin{align*}
|w_{j}-w_{j-1}|&=|\alpha'(t_j)-\alpha'(t_{j-1})|=|\int_{t_{j-1}}^{t_j}\alpha''(t)\,dt|\leq \int_{t_{j-1}}^{t_j}|\alpha''(t)|\,dt \\
&= \langle|\alpha''|,1\rangle_{\Ls^2(t_{j-1},t_j)}
\leq \|\alpha''\|_{\Ls^2(t_{j-1},t_j)}\cdot (t_j-t_{j-1})^\frac{1}{2}. 
\end{align*}
Denote that $T_j=t_j-t_{j-1}$.
Then $\frac{1}{T_j}|w_j-w_{j-1}|^2 \leq \|\alpha''(t)\|^2_{\Ls^2(t_{j-1},t_j)}$.
Now \begin{align*}
    \int_{t_{j-1}}^{t_j}|\alpha''(t)|^2\,dt+\int_{t_{j-1}}^{t_j}1\,dt \geq \frac{1}{T_j}|w_j-w_{j-1}|^2+T_j.
\end{align*}
Let $f_j(T)=\frac{1}{T}|w_j-w_{j-1}|^2+T$. Now $f_j(T)\geq \min_{t>0}f_j(t)=f_j(T_0)$, where  $T_0=2\sqrt{|w_j-w_{j-1}|^2}=2|w_j-w_{j-1}|$ satisfies $f_j'(T_0)=0$. Thus 
\begin{align*}
    \int_{t_{j-1}}^{t_j}|\alpha''(t)|\,dt+\int_{t_{j-1}}^{t_j}1\,dt \geq 2|w_j-w_{j-1}|.
\end{align*}
Furthermore \begin{align*}
      \sum_{j=1}^{k} \big( \int_{t_{j-1}}^{t_j}|\alpha''(t)|^2\,dt+\int_{t_{j-1}}^{t_j}1\,dt \big)\geq  \sum_{j=1}^{k}2|w_j-w_{j-1}| \geq 2|w_k-w_0| = 2|v_1-v_0|.
\end{align*}
To conclude, 
\begin{align*}
    D((p_0,v_0),(p_1,v_1)) \geq
    \max(
        |p_1-p_0|,
        2|v_1-v_0|).
\end{align*}
Especially, $D((p_0,v_0),(p_1,v_1))=0$ only if $(p_0,v_0)=(p_1,v_1)$. Otherwise $D((p_0,v_0),(p_1,v_1))>0$. 

Clearly it is true that $D((p_0,v_0),(p_1,v_1))=D((p_1,v_1),(p_0,v_0)$. 
Finally $D((p_0,v_0),(p_2,v_2))\leq D((p_0,v_0),(p_1,v_1)) + D((p_1,v_1),(p_2,v_2))$ due to the fact that we are taking infinite sums over the integral of piece-wise smooth paths.
\end{proof}

The metric $D$ defined in Definition \ref{def:theor_cw_dist} is difficult to calculate in practice. Thus we will approximate it with a more computable function. Recall that we may assume $|\alpha'(t)|=1$ so that $|K_\alpha(t)|=|\alpha''(t)|$. When $p_0=(0,0)$ 
$v_0=(1,0)$, and $(p_1,v_1)\in S\R^2$, then 
\begin{align*}    
D((p_0,v_0),(p_1,v_1))=\inf_{\alpha \in X[(p_0,v_0),(p_1,v_1)],|\alpha'(t)|=1} \int_0^{L(\alpha)} |\alpha''(t)|^2\,dt + \int_0^{L(\alpha)} 1 \,dt,            
\end{align*}
where $L(\alpha)$ is the length of a curve $\alpha$. 
Let us next write $p_j=(x_j,y_j)$ and denote by  $\theta_j$ the angle between the $x$-axis and $v_j$. Using this convention, we identify $(p_j,v_j)$ with $(x_j,y_j,\theta_j)$.
By translating the origin and rotating the coordinate axis, we may assume that 
$(x_0,y_0,\theta_0)=(0,0,0)$. Then, we approximate $D((p_0,v_0),(p_1,v_1))=
D((0,0,0),(x_1,y_1,\theta_1))$  by the function
\begin{align*}    
\widetilde{D}((0,0,0),(x_1,y_1,\theta_1))=\inf_{g} \int_0^{x_1} |g''(x)|^2\,dx + |(x_1,y_1)-(0,0)|,       
\end{align*}
where the infimum
is taken over functions $g\in C^2([0,x_1])$ that satisfy $g(0)=0,$ $g'(0)=0$
and $g(x_1)=y_1$, $g'(x_1)=\tan(\theta_1).$ This approximation is motivated by considering curves $\alpha_g(s)=(s,g(s))$. Note that for $\alpha_g$ we have $\alpha_g''(s)=(0,g''(s))$ but $|\alpha'_g(s)|=\sqrt{1+|g'(s)|^2}\neq 1$.

\begin{definition}
We define a modified distance function $\widetilde{D}\colon (\R^3)^2 \to \R$, called a \emph{candywrap distance}, on the space $\R^3$ as follows:

For $(x_0,y_0,\theta_0)\in \R^3$, satisfying $x_0\in (0,\infty)$
and $\cos(\theta_0)>0$,
we define
\begin{align}\label{min tilde D}
\widetilde{D}\big((x_0,y_0,\theta_0),(0,0,0)\big) = \inf_g \int_{0}^{x_0} | g''(s) | ^2 \,ds+ \sqrt{x_0^2+y_0^2},    
\end{align}
where infimum is taken over functions $g:[0,x_0]\to \R$,
$g\in C^2([0,x_0])$ satisfying  $g(0)=0$, $g(x_0)=y_0$, $g'(0)=0$ and $g'(x_0)=\tan(\theta_0)$.
Moreover, when $x_0<0$ and $\cos(\theta_0)>0$, we define 
\begin{align*}
\widetilde{D}\big((x_0,y_0,\theta_0),(0,0,0)\big) =
\widetilde{D}\big((-x_0,-y_0,\theta_0),(0,0,0)\big). 
\end{align*}
When $(x_0,y_0)\in \{0\}\times (\R\setminus \{0\})$ or 
$\cos(\theta_0)\le 0$,
we define $\widetilde{D}\big((x_0,y_0,\theta_0),(0,0,0)\big)=\infty$
and $\widetilde{D}\big(0,0,0),(0,0,0)\big)=0$.

For $(x_i,y_i,\theta_i)\in \R^3$, $i=1,2$, let
\begin{align*}
\widetilde{D}\big((x_1,y_1,\theta_1),(x_2,y_2,\theta_2)\big) = 
\widetilde{D}\big( \rm{Rot}_{\theta_2}\rm{Mov}_{(x_2,y_2)}
(x_1,y_1,\theta_1),(0,0,0)
\big)
\end{align*}
where $\rm{Mov}_{(s,t)}:\R^3 \to \R^3,$ $(x,y,\theta)\mapsto (x-s,y-t,\theta)$ is a shift operation,  $\rm{Rot}_{\theta_0}:\R^3 \to \R^3$, $(x,y,\theta) \mapsto (\cos(\theta_0)x+\sin(\theta_0)y, -\sin(\theta_0)x+\cos(\theta_0)y,\theta-\theta_0)$ is a rotation operator.
\end{definition}

As the modified distance function $\widetilde{D}\big((x_0,y_0,\theta_0),(0,0,0)\big)$ in \eqref{min tilde D} is defined using a quadratic form $Q[g,g]=\|\frac {d^2}{ds^2}g\|_{L^2}^2$,
we see that the minimum at \eqref{min tilde D} is obtained and  
the minimizer $g:[0,x_0]\to \R$ satisfies $\frac{d^4}{ds^4}g=0$ in the sense of distributions
and thus $g$
is the unique third-order polynomial satisfying   $g(0)=0$, $g(x_0)=y_0$, $g'(0)=0$ and $g'(x_0)=\tan(\theta_0)$.
We emphasize that the modified distance function $\widetilde{D}\big((x_0,y_0,\theta_0),(0,0,0)\big)$
is not symmetric, see Figure \ref{fig:candywrap_ball_sizes_and_angles}. 

\begin{lemma}
\label{lemma:cw_dist}
The candywrap distance is given by the formula \begin{align*}
    \widetilde{D}\big( (0,0,0),(x,y,\theta)\big) = \begin{cases}
        4(3a^2x^3+3abx^2+b^2x) + \sqrt{x^2+y^2} \text{, if } x>0 \\        -4(3a^2x^3+3abx^2+b^2x) + \sqrt{x^2+y^2} \text{, if } x<0,
    \end{cases}
\end{align*} 
where $a=\frac{1}{x^2}(\tan\theta-2\frac{y}{x})$ and $b=\frac{1}{x}(-\tan\theta+3\frac{y}{x})$.
\end{lemma}

\begin{proof}
Let $g(s)=as^3+bs^2+cs+d$. We know by the definition of the candywrap distance that $g(0)=0$ and $g'(0)=0$. These imply that $d=0$ and $c=0$. Hence $g(s)=as^3+bs^2$, $g'(s)=3as^2+2bs$ and $g''(s)=6as+2b$. Let us solve the coefficients $a$ and $b$. We get a system of equations from $g(x)=y$ and $g'(x)=\tan\theta$:
\begin{align*}
    \begin{cases}
        ax^3+bx^2 = y \\
        3ax^2 + 2bx = \tan\theta
    \end{cases}
    \Leftrightarrow
    \begin{cases}
        2ax^3+2bx^2 = 2y \\
        3ax^3 + 2bx^x = x\tan\theta
    \end{cases}
    \Leftrightarrow
    \begin{cases}
        2ax^3+2bx^2 = 2y \\
        ax^3 = x\tan\theta-2y.
    \end{cases}
\end{align*}
We compute that $a=\frac{1}{x^2}(\tan\theta-2\frac{y}{x})$. Furthermore we compute that $b=\frac{1}{x}(-\tan\theta + 3\frac{y}{x} )$. 

We note that \begin{align*}
  | g''(s) |^2  = 4(9a^2s^2+6abs+b^2). 
\end{align*}
Now if $x>0$, we get that
\begin{align*}
    \int_0^x  |g''(s)|^2 \,ds = \int_0^x 4(9a^2s^2+6abs+b^2) \,ds =  4(3a^2x^3+3abx^2+b^2x.)
\end{align*}
The statement follows by adding $\sqrt{x^2+y^2}$ to $\int_x^0   | g''(s)|^2 \,ds$.  
\end{proof}

With the candywrap distance, we can approximate neighborhoods of invisible singularities. This can be done by setting intersections of balls defined with the candywrap distance using known boundary curves' endpoints as center points. The intersections can be done taking only the directions of unknown singularities. By growing the radii of these balls, we can study when known disjoint parts of boundaries with these balls meet all the levels $\R^2\times{t}$, $t\in[0,2\pi]$. If so, we may say that invisible singularities belong to these balls, i.e., we have found neighborhoods of invisible singularities. The goal is to find the minimum radius which gives the neighborhoods. This can be done using persistent homology, which will be explained later.

Note that the candywrap distance from the origo to the point near the origo (in the Euclidean sense) is sufficiently large. Thus the minimum radius might have to be unreasonably large for all levels to be met. Thus one might use all points as center points in small Euclidean balls for candywrap distance-defined balls.               

\begin{figure}
     \centering
     \includegraphics[width=\textwidth]{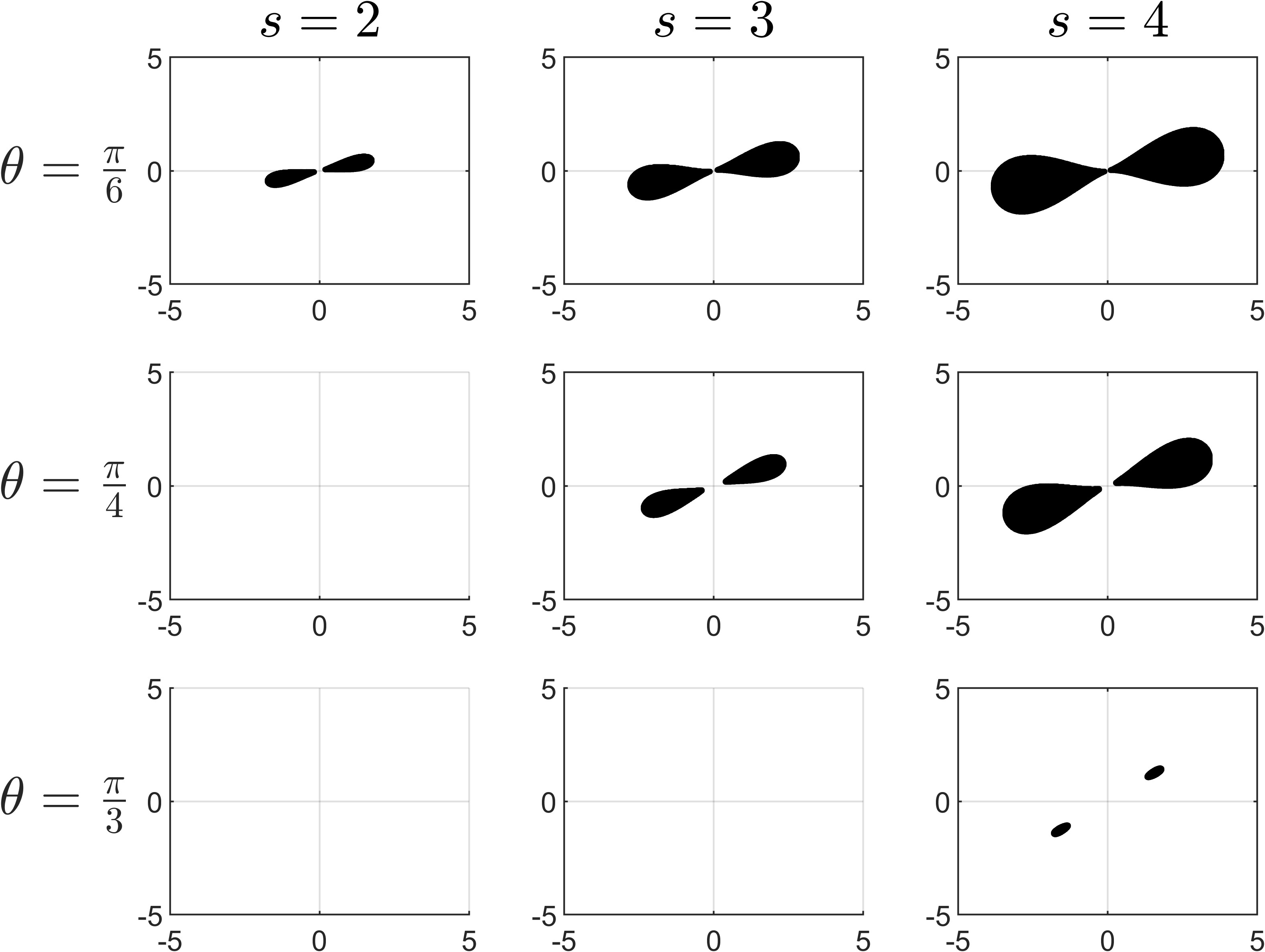}
     \caption{The illustrations of sets $\{(x,y)\in\R^2 \mid \Tilde{D}\big((0,0,0),(x,y,\theta)\big)<s\}$, where $s=2,3,4$ and $\theta=\frac{\pi}{6},\frac{\pi}{4},\frac{\pi}{3}$.} 
     \label{fig:candywrap_ball_sizes_and_angles}
 \end{figure}

\subsection{Complex wavelet transform}

We need computational tools for dealing with the (incomplete) wavefront set recovered from tomographic reconstructions. Shearlets and curvelets have exact results for characterizing the wavefront set, which was the main reason they were originally introduced \cite{shearlets,curvelet}. However, we look for a simpler, computationally efficient algorithm and explore the complex dual-tree wavelet transform \cite{DoubleDensityDualTree_Selesnick} for that purpose. 

Theoretically, complex wavelets lack the supreme approximation properties enjoyed by curvelets and shearlets. However, in practice complex wavelets are easy to adapt, computationally cheap, and by experience they track singularities and their directions accurately enough for our practical purposes. 

To be more precise, it is observed that if a detail coefficient is large it is 'more likely' that there is a singularity with the studied direction. As an example, the $LH$ wavelet can capture singularities with direction $[60^\circ,90^\circ]$. See Table \ref{table:wavelet_orientations} for which directions other wavelets capture, and Figure \ref{fig:ellipse_cw_example} as an example. We can utilize this observation when we study a wavefront set of a density function.

A complex wavelet transform ($\C$WT) is a computational tool to find the subsets of the wavefront set of given data. In our study, the given data is reconstructed density function $f$ from tomographic data. A complex wavelet transform gives more directional information about singularities in two and higher dimensions than discrete real wavelet transformations. Computationally complex wavelets' detail and approximation coefficients can be formed from two 1D real mother wavelets and two 1D real scaling functions using the dual-tree structure.  This makes $\C$WT computationally easy to adapt. In this section, we will go through the general idea of $\C$WT. For more details, the authors suggest reading \cite{Daubechies} for general wavelet theory, and \cite{DualTree} for complex wavelets. 

Let us first define a one-dimensional complex-valued wavelet, that is
 \begin{align*}
 \psi(x)=\psi_r(x)+i\psi_i(x),    
 \end{align*}
 where both $\psi_r$ and $\psi_i$ are real-valued wavelets. Additionally, $\psi_r$ is even and $\psi_i$ is odd. Moreover, $\psi_r$ and $\psi_i$ are almost Hilbert transform pairs, that is, $(\mathcal{H}\psi_i)(x) \approx \frac{1}{\pi}\int_{-\infty}^{\infty}\frac{\psi_r(\tau)}{x-\tau}\,d\tau$. We define a one-dimensional complex scaling function similarly, $\phi(x) = \phi_r(x)+i\phi_i(x)$, where $\phi_r\approx\mathcal{H}\phi_i$. 

A one-dimensional complex wavelet transform is
\begin{align*}
    f(x) &= \sum_{n=-\infty}^{\infty}a(n)\phi(x-n)+\sum_{j=0}^{\infty}\sum_{n=-\infty}^{\infty}d(j,n)2^{j/2}\psi(2^j x-n), \\
\end{align*}
where approximation coefficients $a(n)$ are defined as \begin{align*}
     a(n)=a_i(n)+ia_r(n)=\int_{-\infty}^\infty f(x)\phi_r(x-n)\,dx +i \int_{-\infty}^\infty f(x)\phi_i(x-n)\,dx  
 \end{align*}
 and detail coefficients $d(j,n)$ are defined as
 \begin{align*}
d(j,n)&=d_r(j,n)+id_i(j,n) \\
&=2^{j/2}\int_{-\infty}^\infty f(x)\psi_r(2^jx-n)\,dx+i 2^{j/2}\int_{-\infty}^\infty f(x)\psi_i(2^jx-n)\,dx.
 \end{align*} 
We call $j=0,1,\dots,\infty$ a level. 

We obtain 2-D complex wavelets by taking the tensor product of a wavelet and scaling function and switching their roles along the $x$-direction and $y$-direction.
 For example, a two-dimensional, $-45^\circ$ orientated, complex wavelet is 
 \begin{align*}
 \psi(x)\psi(y) 
 &= [\psi_r(x)+i\psi_i(x)][\psi_r(y)+i\psi_i(y)] \\
 &= \psi_r(x)\psi_r(y)-\psi_i(x)\psi_i(y) + i [\psi_r(x)\psi_i(y)+\psi_i(x)\psi_r(y)].
 \end{align*}
We get another diagonal orientated, $+45^\circ$, wavelet from taking the complex conjugate of $\psi(y)$. That is, 
 \begin{align*}
 \psi(x)\overline{\psi(y)} 
 &= [\psi_r(x)+i\psi_i(x)][\psi_r(y)-i\psi_i(y)] \\
 &= \psi_r(x)\psi_r(y)-\psi_i(x)\psi_i(y) + i [\psi_r(x)\psi_i(y)-\psi_i(x)\psi_r(y)].
 \end{align*}
 We get four other oriented ($-75^\circ, +75^\circ, -15^\circ$, and  $+15^\circ$) wavelets by considering wavelets $\psi(x)\phi(y)$, $\psi(x)\overline{\phi(y)}$, $\phi(x)\psi(y)$, and $\phi(x)\overline{\psi(y)}$ respectively. A 2-dimensional scaling function is defined as $\phi(x)\phi(y)$.  The two-dimensional complex wavelets being oriented comes from the fact that the Fourier spectrum is supported in only one quadrant of the two-dimensional frequency plane.

Computationally, we use a dual-tree structure and filters to compute complex wavelet coefficients \cite{DualTree}. In practice, we perform two real wavelet transformations which together form complex wavelet transformations. Both transformations have a set of filters, a high-pass filter and a low-pass filter. They are designed so that together transformation is approximately analytic. 

A high-pass filter $H$ is associated with the mother wavelet $\psi$ and a low-pass filter $L$ is associated with the scaling function $\phi$ to compute the approximation and the detail coefficients. In a complex wavelet transform, we have two sets of high-pass filters, $H_{r}$ and $H_{i}$, and low-pass filters, $L_{r}$ and $L_{i}$. The filters $H_{r}$ and $L_{r}$ are associated with the mother wavelet $ \psi_r$ and scaling function $\phi_r$ respectively. The filters $H_{i}$ and $L_{i}$ are associated with the mother wavelet $ \psi_i$ and scaling function $\phi_i$ respectively. 

We say that the wavelet $\psi(x)\psi(y)$ is an $HH$ wavelet since it is computationally obtained by using high-pass filters $H_i$ and $H_r$. The wavelet $\psi(x)\overline{\psi}(y)$ is an $H\oH$ wavelet since it is computed using high-pass filters like the previous wavelet $\psi(x)\psi(y)$ but we are taking the complex conjugate. The other four wavelets' relationship to filters and orientations are presented in Table \ref{table:wavelet_orientations}.  
\begin{table}[!ht]
\begin{center}
 \caption{Computationally, filters correspond to complex wavelets. Each wavelet is oriented differently and they can capture singularities in different directions.}
\begin{tabular}{|c|c|c|c|}
\hline
Wavelet & Filters & Orientation of wavelet & Direction of singularity \\
\hline
$\phi(x)\psi(y)$&$LH$& $-15^\circ$ & $[60^\circ,90^\circ]$\\
$\phi(x)\phi(y)$&$HH$& $-45^\circ$ & $[30^\circ,60^\circ]$\\
$\psi(x)\phi(y)$&$HL$&$-75^\circ$  & $[0^\circ,30^\circ]$ \\
$\psi(x)\overline{\phi(y)}$&$H\oL$& $+75^\circ$ & $[-30^\circ,0^\circ]$ \\
$\psi(x)\overline{\psi(y)}$&$H\oH$& $+45^\circ$ & $[-60^\circ,-30^\circ]$\\
$\phi(x)\overline{\psi(y)}$&$L\oH$& $+15^\circ$ & $[-90^\circ,-60^\circ]$ \\
\hline
\end{tabular}
\label{table:wavelet_orientations}
\end{center}
\end{table}

\begin{figure}
    \centering
     \begin{subfigure}[t]{0.12\textwidth}
         \centering
         \includegraphics[width=\textwidth]{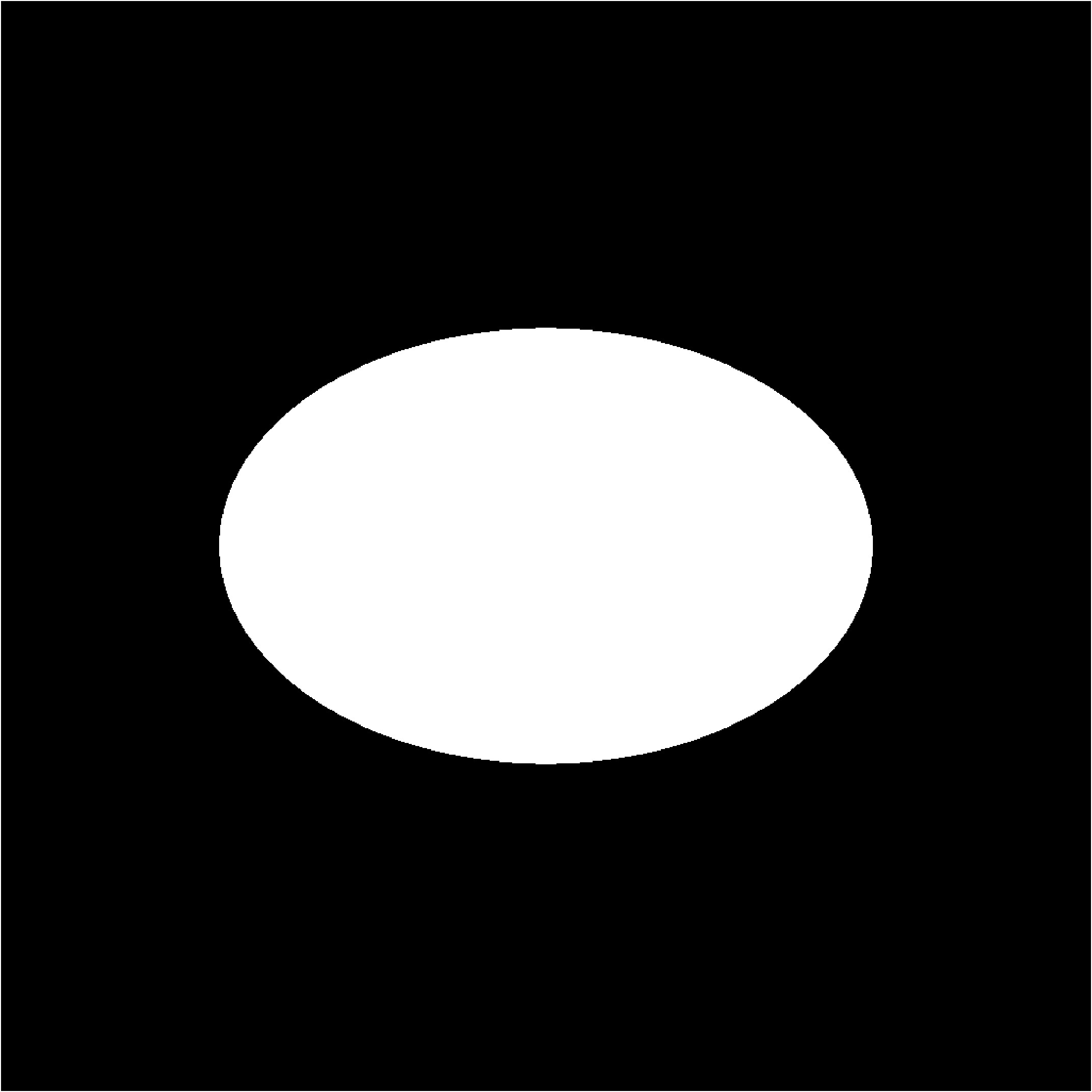}
         \caption{An ellipse}
         \label{fig:ellipse}
     \end{subfigure}
     \begin{subfigure}[t]{0.12\textwidth}
         \centering
         \includegraphics[width=\textwidth]{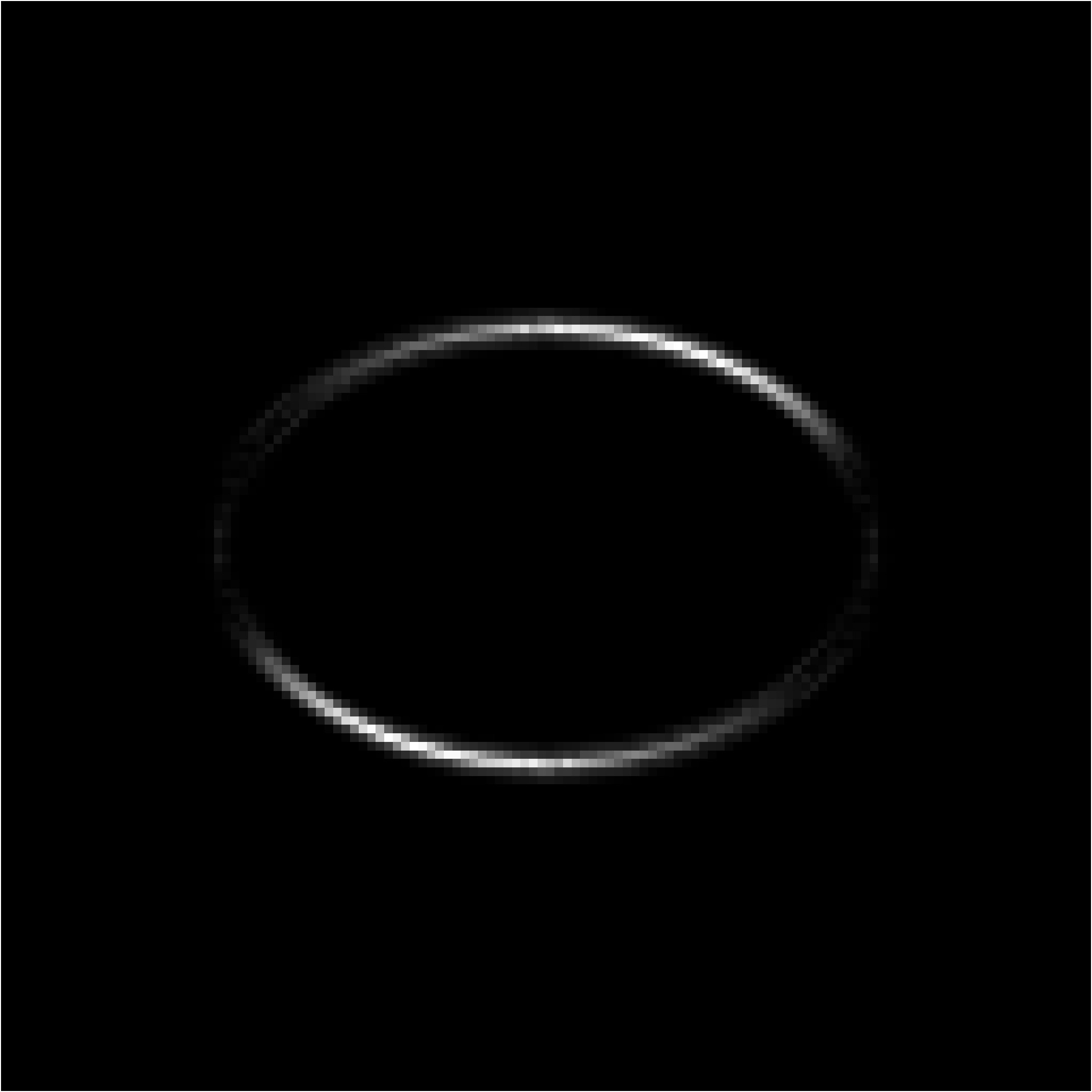}
         \caption{$LH$}
         \label{fig:coefs_abs6}
     \end{subfigure}
     \begin{subfigure}[t]{0.12\textwidth}
         \centering
         \includegraphics[width=\textwidth]{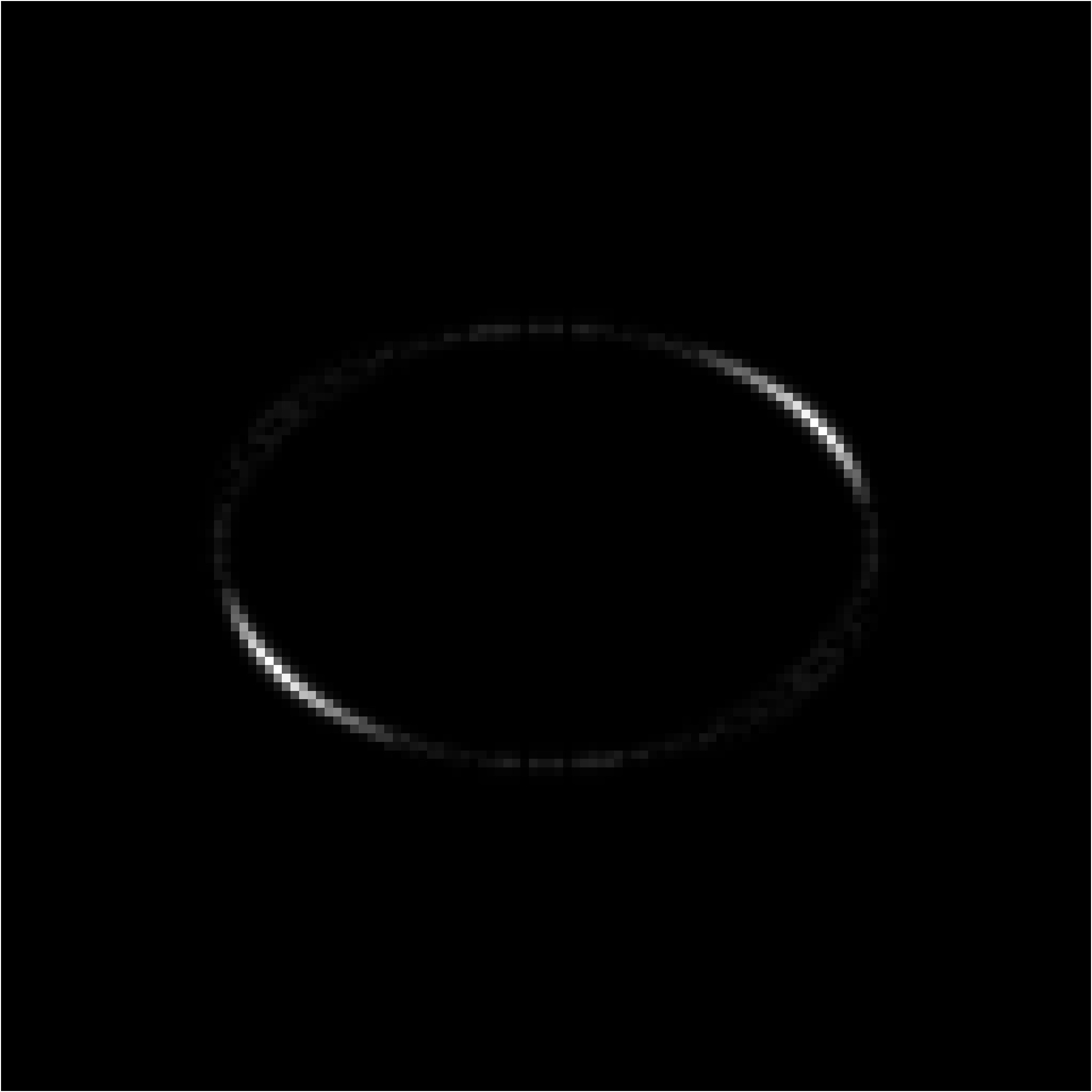}
         \caption{$HH$}
         \label{fig:coefs_abs7}
     \end{subfigure}
     \begin{subfigure}[t]{0.12\textwidth}
         \centering
         \includegraphics[width=\textwidth]{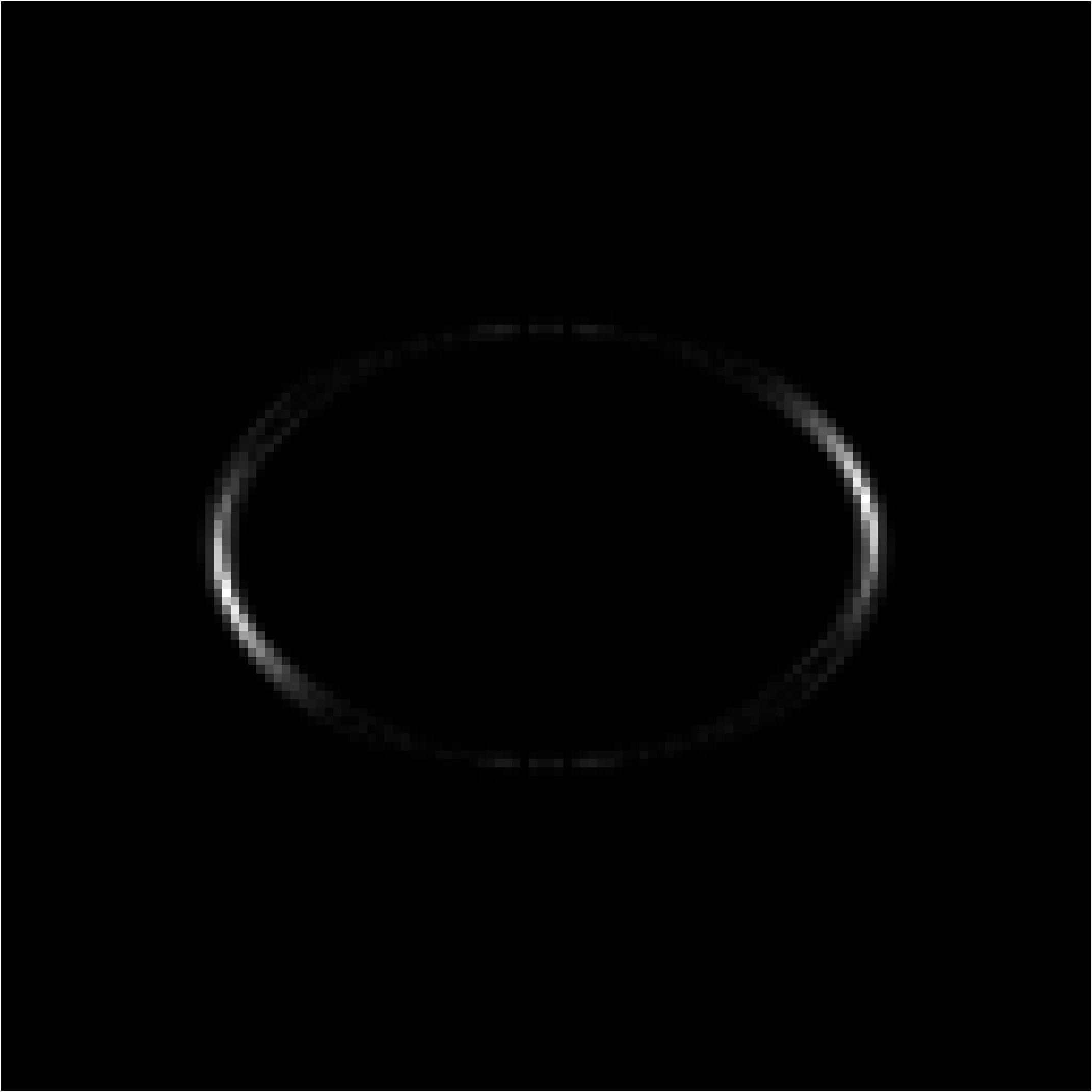}
         \caption{$HL$}
         \label{fig:coefs_abs8}
     \end{subfigure}
     \begin{subfigure}[t]{0.12\textwidth}
         \centering
         \includegraphics[width=\textwidth]{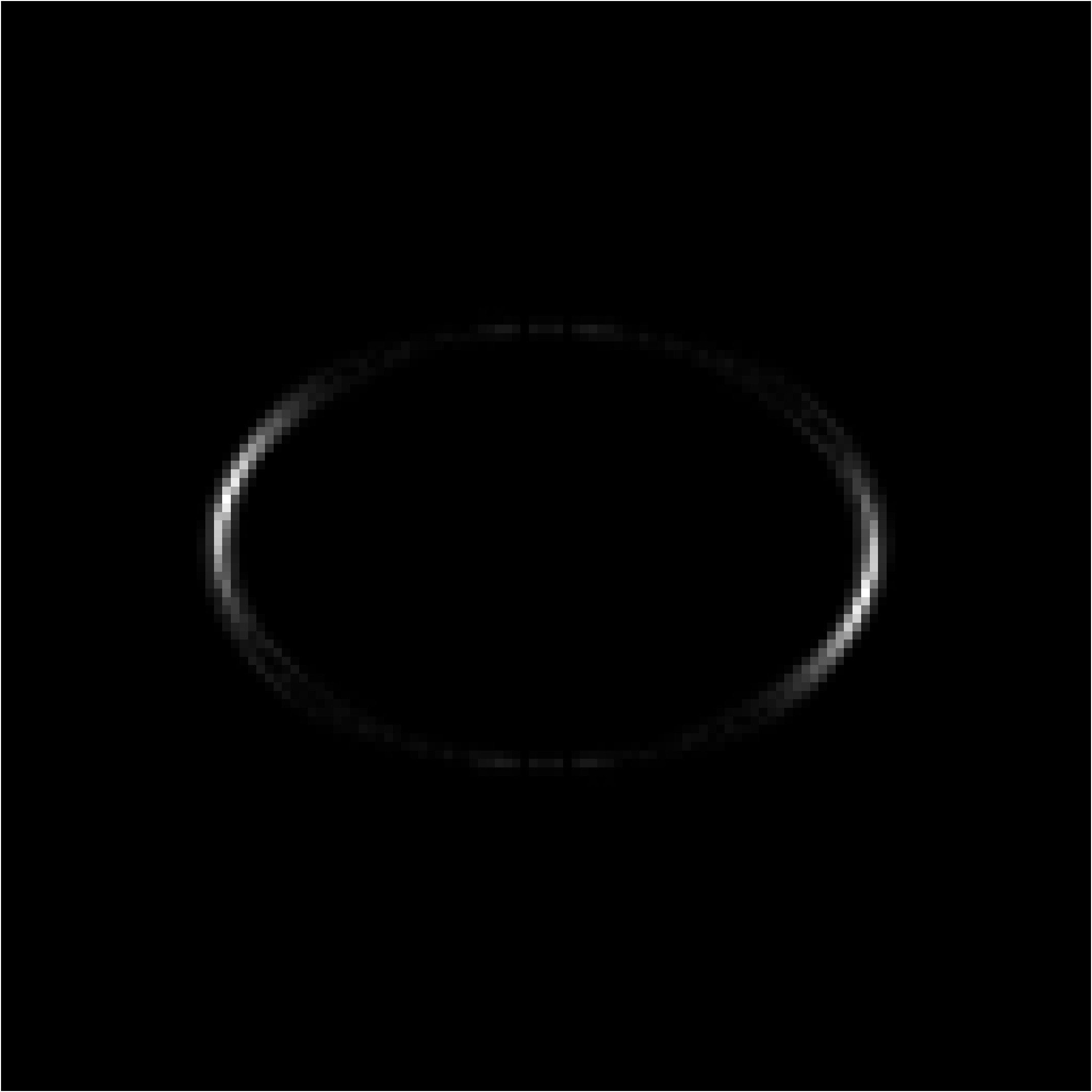}
         \caption{$H\oL$}
         \label{fig:coefs_abs9}
              \end{subfigure}
         \begin{subfigure}[t]{0.12\textwidth}
         \centering
         \includegraphics[width=\textwidth]{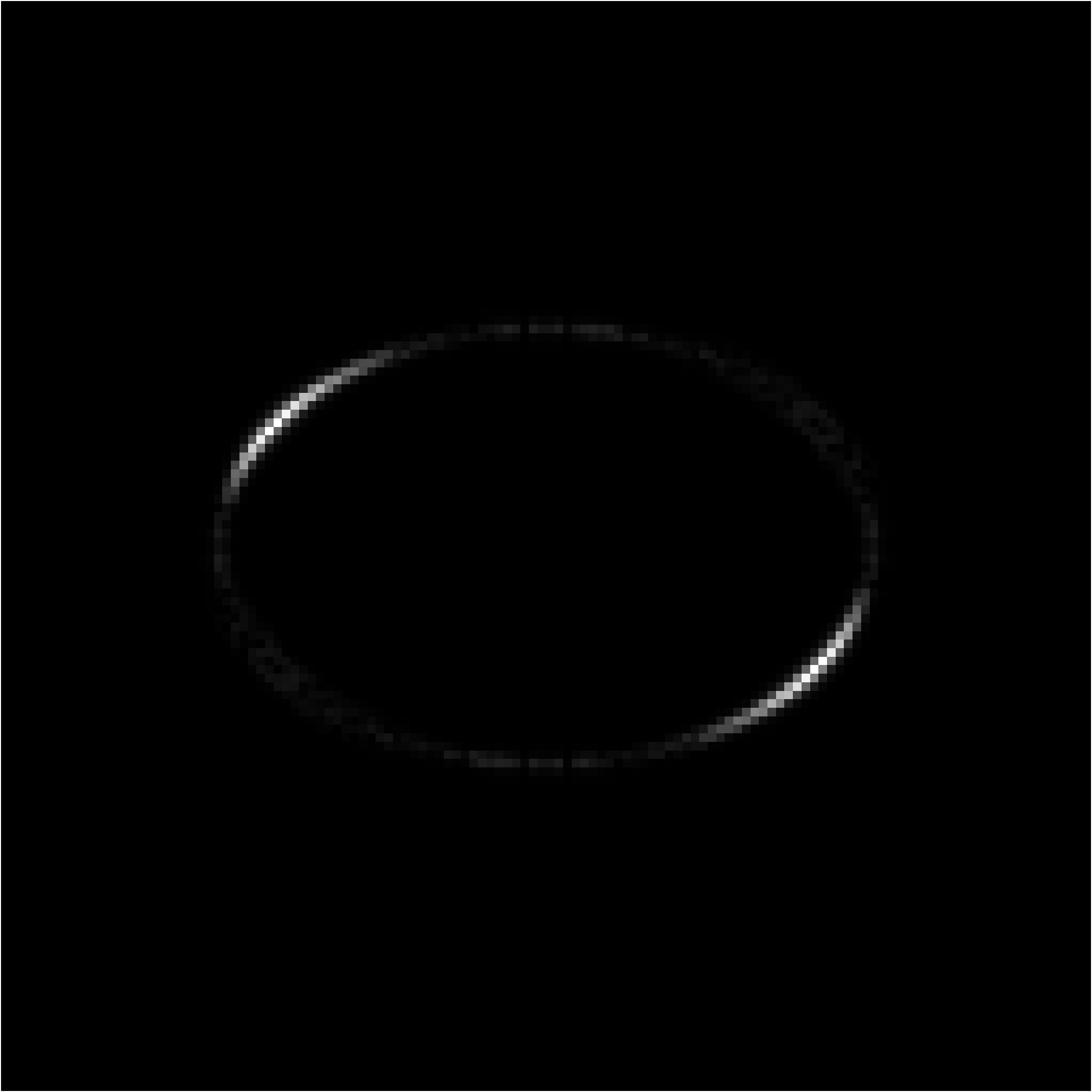}
         \caption{$H\oH$}
         \label{fig:coefs_abs10}
     \end{subfigure}
     \begin{subfigure}[t]{0.12\textwidth}
         \centering
         \includegraphics[width=\textwidth]{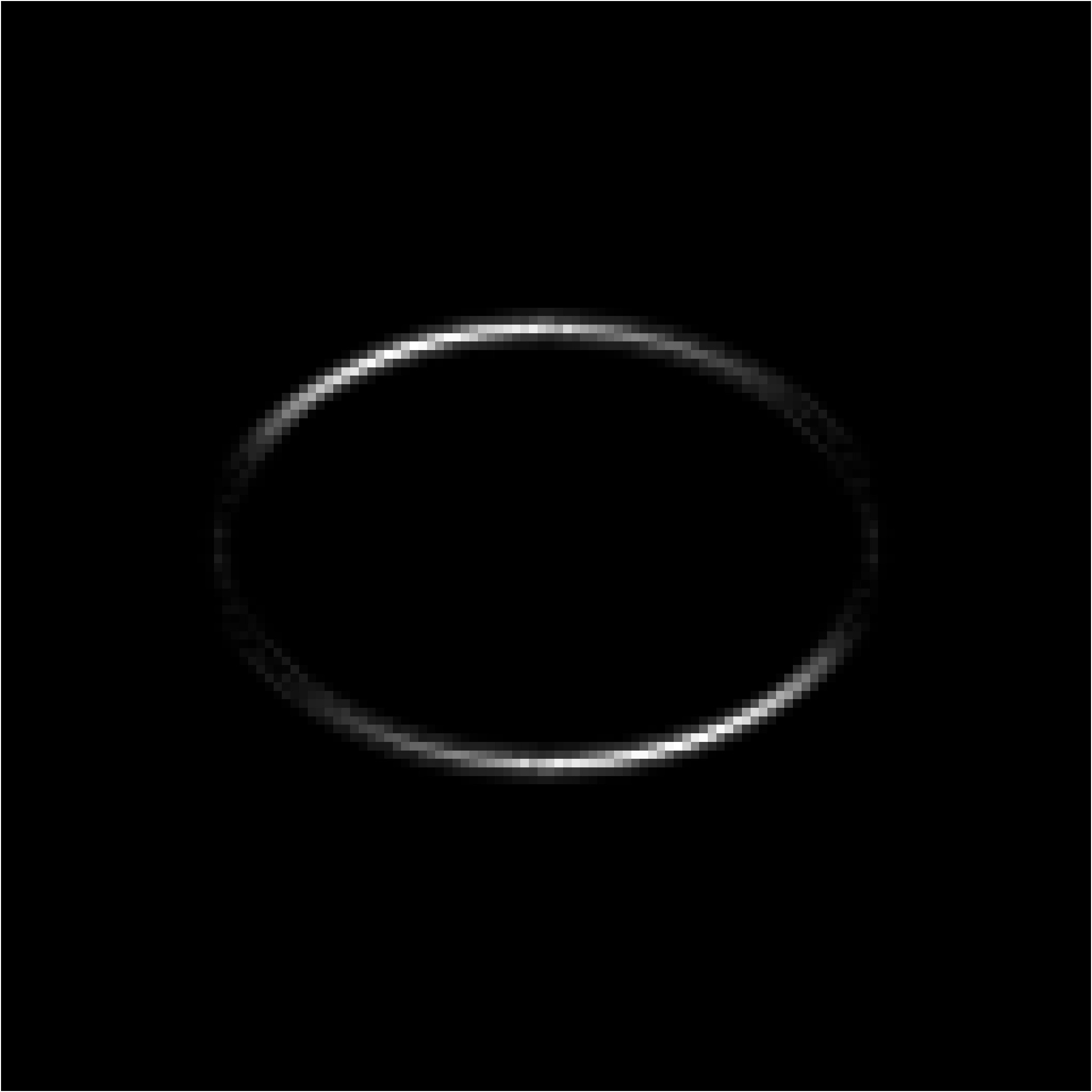}
         \caption{$L\oH$}
         \label{fig:coefs_abs11}
     \end{subfigure}
     \caption{Subfigures (b)-(g) show magnitudes of detail coefficients of an ellipse in (a) computed with dual-tree complex wavelets, when $j=7$.} 
        \label{fig:ellipse_cw_example}
\end{figure}

\subsection{Persistent homology}

In our study, we will use persistent homology to find minimum distance $s$ for invisible singularities' neighborhoods, as explained in Section \ref{sec:cw_dist}. In this section, we will go through the required background of cubical (persistent) homology.
Persistent homology is a tool to consider the homology of data across multiple scales. In this paper, we focus on cubical (persistent) homology, in which building blocks are elementary cubes. We limit the theory to elementary cubes in $\R^3$ since that is the space we are working within our study. However, the theory can be done for any dimension, see \cite{CompHomKaczynski}.

Let us recall basic terminology and notation. An \emph{elementary interval} is an interval of the form $[l,l+1]\subset \R$, or $[l]=[l,l]$ for some $l\in\Z$, i.e. a unit length interval or a single point.
A product 
\[
C=[l_1,l_1+s_1] \times [l_2,l_2+s_2]\times [l_3,l_3+s_3] \subset \R^3
\]
of elementary intervals, where $l_1,l_2,l_3\in \Z$ and $s_1,s_2,s_3\in \{0,1\}$, is a \emph{$p$-cube for $p=0,1,2,3$} if $s_1+s_2+s_3 = p$. In particular, a $p$-cube is isometric to the standard $p$-dimensional unit cube $[0,1]^p$. Note that there are $2p$ $(p-1)$-cubes contained in a $p$-cube $C$. We call these $(p-1)$-cubes the faces of $C$.

A cubical complex $Q\subset \R^3$ is a finite collection of $p-$cubes, for $p =0,1,2,3$, satisfying the following conditions: 
\begin{itemize}
\item $Q$ contains the faces of all its cubes, and 
\item intersection of cubes in $Q$ are in $Q$.
\end{itemize}

A chain group $C_p(Q)$ is a free $\Z_2$ vector space generated by the set of $p$-cubes. We call elements of $C_p(Q)$ $p-$chains.
The boundary operator $\partial_p \colon C_p(Q)\to C_{p-1}(Q)$ is the map satisfying, for each $C\in C_p(Q)$, that $\partial_p(C)$ is the sum of the faces of $C$. Since each $C_p(Q)$ is a $\Z_2$-vector space, we have that $\partial_{p-1} \circ \partial_p = 0$.

We call $p$-chain $c$ a \emph{$p-$cycle} if $\partial_p(c)=0$ and a \emph{$p-$boundary} if there exists a $(p+1)$-chain $b$ such that  $\partial_{p+1}(b)=c$. We denote $Z_p(Q)$ the set of all $p-$cycles related to a cubical complex $Q$.  
Similarly, we denote $B_p(Q)$ the set of all $p-$boundaries. Now we can define a homology group.

\begin{definition}
A $p-$th homology group $H_p(Q)$ is a quotient group of $p$-cycles mod $p$-boundaries, $H_p(Q)=Z(p)/B(p)$. 
\end{definition}

We call a nested sequence of cubical complexes
\begin{align*}
    Q_0 \subset Q_1 \subset \cdots \subset Q_k
\end{align*}
a \emph{filtration}. Then the inclusion between $Q_i$ and $Q_j$ induces a map
\begin{align*}
    \ind_{i,j}\colon H_p(Q_{i}) \to H_p(Q_{j}).
\end{align*}

Especially, we are interested in homology groups $H_0(Q)$ which encode components of $Q$. By forming a filtration we can study the evolution of components through different scales. Recall that a component of a cubical complex $Q$ is a connected collection of $p-$cubes of $Q$.

Next, we will consider how 3D matrices can be seen as cubical complexes. We will define a filtration related to 3D matrices and thus consider zeroth homology groups. First, we will define a relation between a binary 3D matrix, $M\in \{0,1\}^{n\times m\times h}$, which can be seen as stack of $\{0,1\}^{n \times m}$ matrices of height $h$, and a cubical complex $Q$.

Consider 3D binary matrix $M\in \{0,1\}^{n\times m\times h}$. We call elements of $M$ voxels. We link voxels to 3-cubes as follows. A voxel $M(x,y,z)=1$ is linked to $3-$cube $q_{x,y,z}=[x-1,x]\times[y-1,y]\times [z-1,z]$. Now we say that a cubical complex $Q$ is related to a 3D matrix $M$, denoted by $M\rightleftharpoons Q$, if the following holds:
$M(x,y,z)=1$, if and only if  $q_{x,y,z}\in Q$.

Consider that we have 3D matrices $M_0,M_1,\dots,M_k$ for which $M_i(x,y,z)=1$ only if for all $j>i$ holds $M_j(x,y,z)=1$. Let $Q_j \rightleftharpoons M_j$ for all $j=0,\dots,k$. Now the filtration of cubical complexes related to $3D$ matrices $M_0,M_1,\dots,M_k$ is
$$Q_{0} \subset Q_{1} \subset \cdots \subset Q_{k}.$$
We finalize this section with the definition of a component matrix.
\begin{definition}
    Let $M,M_j\in \{0,1\}^{n\times m \times h}, j=1,\dots I$, and $M=\sum_{i=0}^I M_i$, where summing is done elementwise. Let $Q_j \rightleftharpoons M_j$ for all $j=1,\dots,I$ and $M\rightleftharpoons Q$. We say that $M_j$, $j=1,\dots, I$ is a \emph{component matrix} of $M$ if and only if $Q_i$ is a component of $Q$.  
\end{definition}

Now finding the components from the 3D matrix $M_j$ corresponds to finding generators of a group $H_0(Q_j)$, where $Q_j \rightleftharpoons M_j$.

\subsection{Morphological operations}

Morphological operations are often used in binary image processing. Morphological dilation makes objects in images more visible. It also fills small holes. Morphological erosion removes isolated pixels and makes lines thinner. Thus only 'essential' shapes remain. Many other operations are combinations of these two operations; see \cite{DigitalImageProcessing} Since binary data can be seen as images, the morphological operations can also be used for data processing, like we do in our method.

In morphology, we consider a binary image $\mathcal{D}$ as a subset of a set $\Z^2$. A translation by a vector $s \in \Z^2$ of a set $\mathcal{D}$ is a set $\mathcal{D}_s=\{x-s \mid x\in \mathcal{D}\}$. Structuring elements are also binary images. 

An \emph{erosion of a binary image} $\mathcal{D}\subset \Z^2$ by a structuring element $\mathcal{S}\subset \Z^2$ is defined by
\begin{align*}
    \mathcal{D} \ominus \mathcal{S} := \bigcap_{s \in \mathcal{S}} \mathcal{D}_{s}.
\end{align*}

A \emph{dilation of a binary image} $\mathcal{D}$ by a structuring element $\mathcal{S}$ is defined by
\begin{align*}
    \mathcal{D} \oplus \mathcal{S} := \bigcup_{s \in \mathcal{S}} \mathcal{D}_{-s}.
\end{align*} 

A \emph{morphological opening} $\mathcal{D} \circ \mathcal{S}$ is a morphological operator which does both the erosion and the dilation. It is defined by
\begin{align*}
\mathcal{D} \circ \mathcal{S} = (\mathcal{D}\ominus \mathcal{S} )\oplus \mathcal{S}.
\end{align*}

A \emph{morphological skeletonization} is a controlled iterative erosion process. Let denote $n$ times dilation of the structuring element $\mathcal{S}$ with itself by $n\mathcal{S}$, i.e., $$n\mathcal{S}=\underset{n \text{ times}}{\underbrace{\mathcal{S}\oplus \cdots \oplus \mathcal{S}}},$$ and $0\mathcal{S}=\{(0,0)\}$.
The skeletonization operation of an image $\mathcal{D}$ by a structuring element $\mathcal{S}$ is
\begin{align*}
\mathrm{Skel}(\mathcal{D})=\bigcup_{n=0}^N(\mathcal{D}\ominus n\mathcal{S}) \setminus [(\mathcal{D}\ominus n\mathcal{S})\circ \mathcal{S}],
\end{align*}
where $N=\max\{n \mid \mathcal{D}\ominus n\mathcal{S}\neq \varnothing\}$. See the examples of the above operations in Figure \ref{fig:morph_operations}.

\begin{figure}
     \centering
     \captionsetup[subfigure]{justification=centering}
     \begin{subfigure}[t]{0.25\textwidth}
         \centering
         \includegraphics[width=\textwidth]{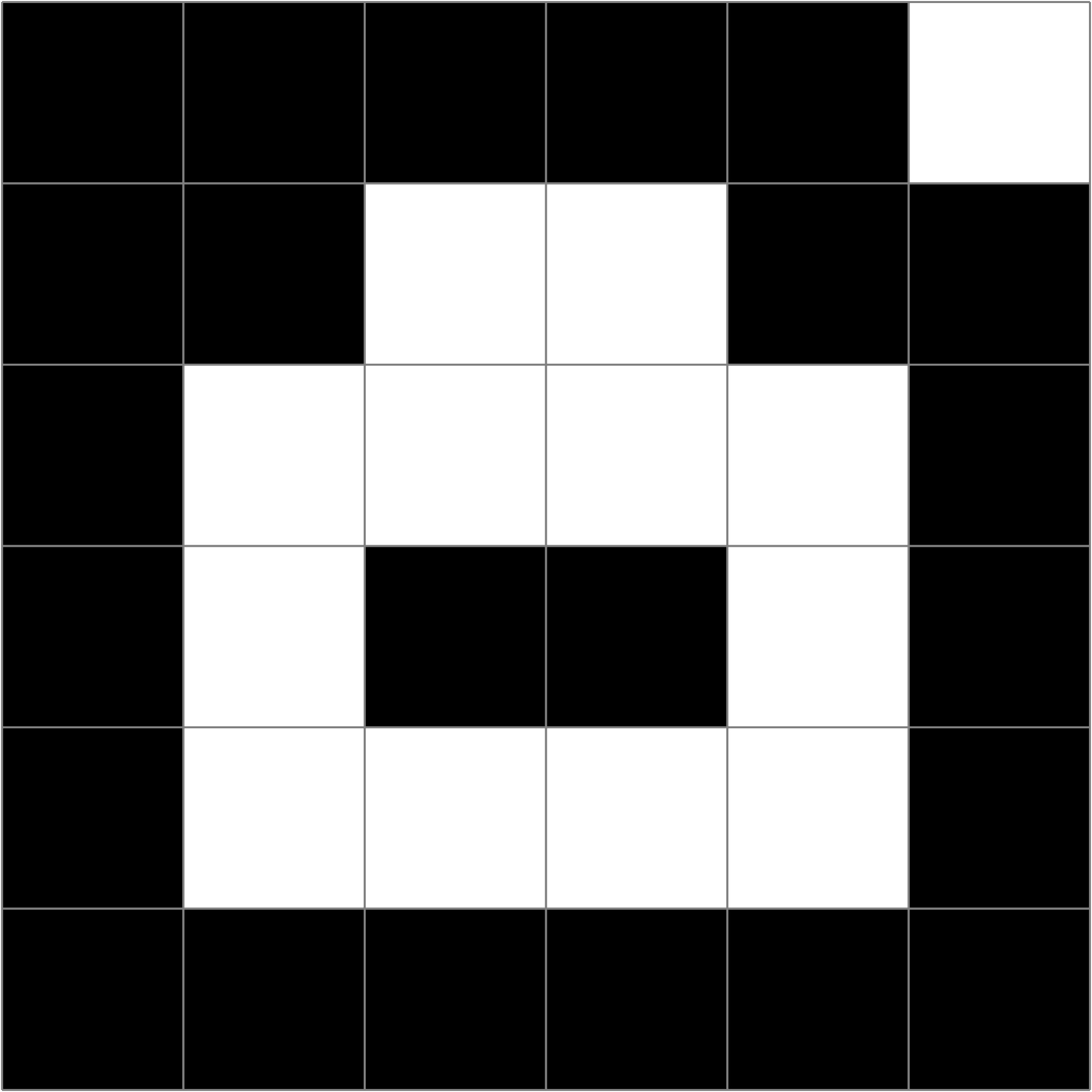}
         \caption{An original binary image $\mathcal{D}$}
         \label{fig:morph_orginal_image}
     \end{subfigure}
     \hspace{5mm}
     \begin{subfigure}[t]{0.25\textwidth}
         \centering
         \includegraphics[width=\textwidth]{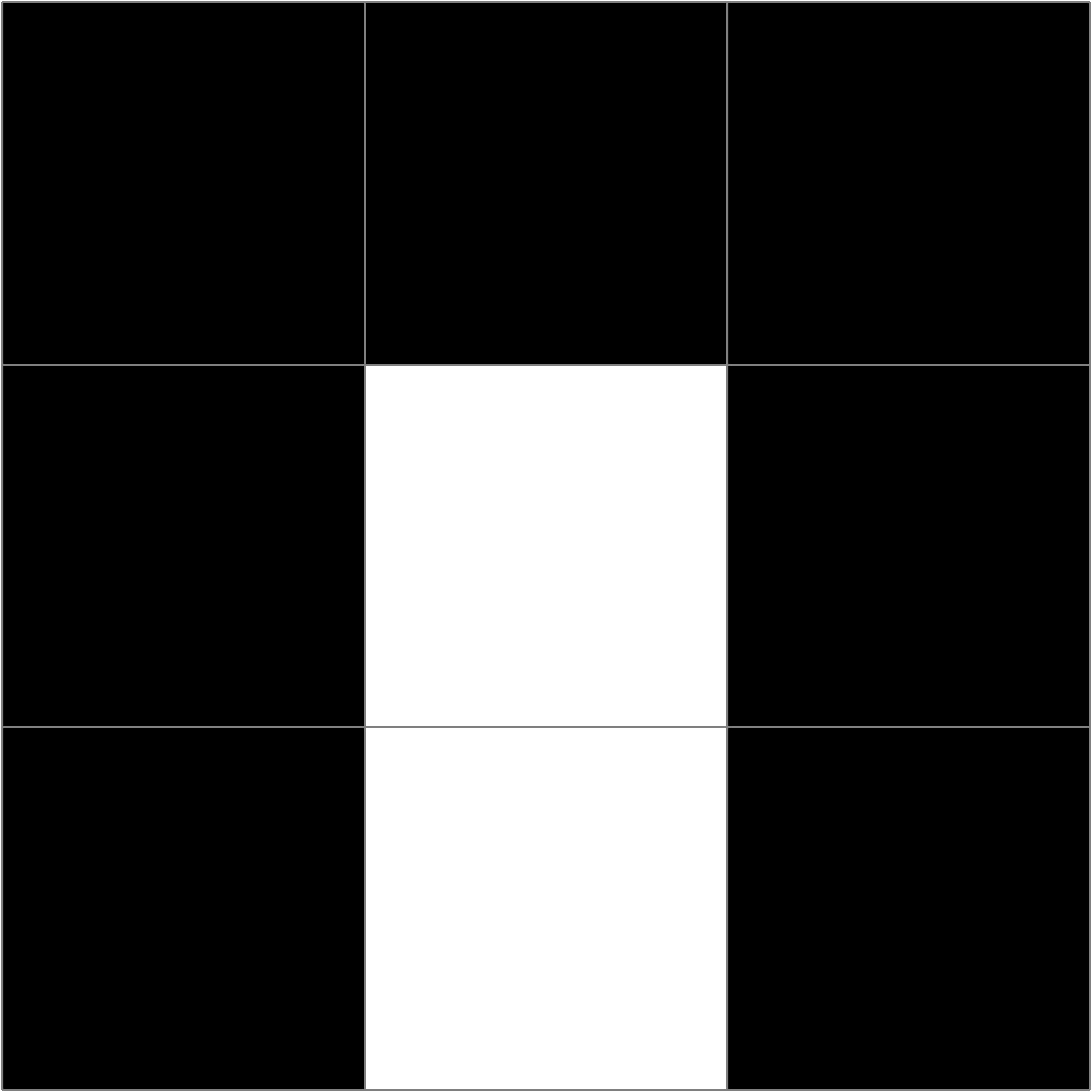}
         \caption{A structuring element $\mathcal{S}$}
         \label{fig:morph_struct_elem}
     \end{subfigure}
          \hspace{5mm}
     \begin{subfigure}[t]{0.25\textwidth}
         \centering
         \includegraphics[width=\textwidth]{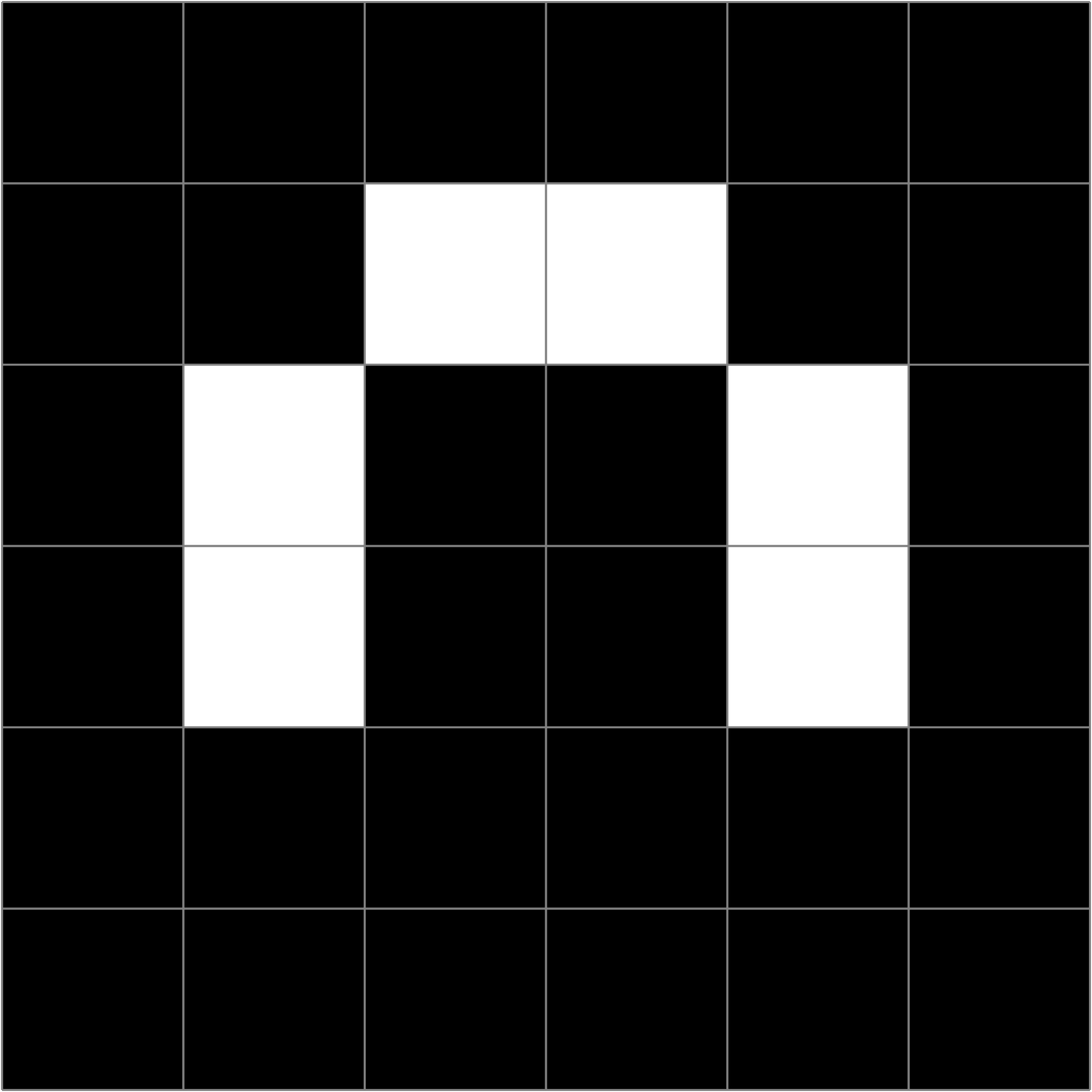}
         \caption{Morphological erosion $\mathcal{D} \ominus \mathcal{S}$}
         \label{fig:morph_erosion}
     \end{subfigure}
               \hspace{5mm}
     \begin{subfigure}[t]{0.25\textwidth}
         \centering
         \includegraphics[width=\textwidth]{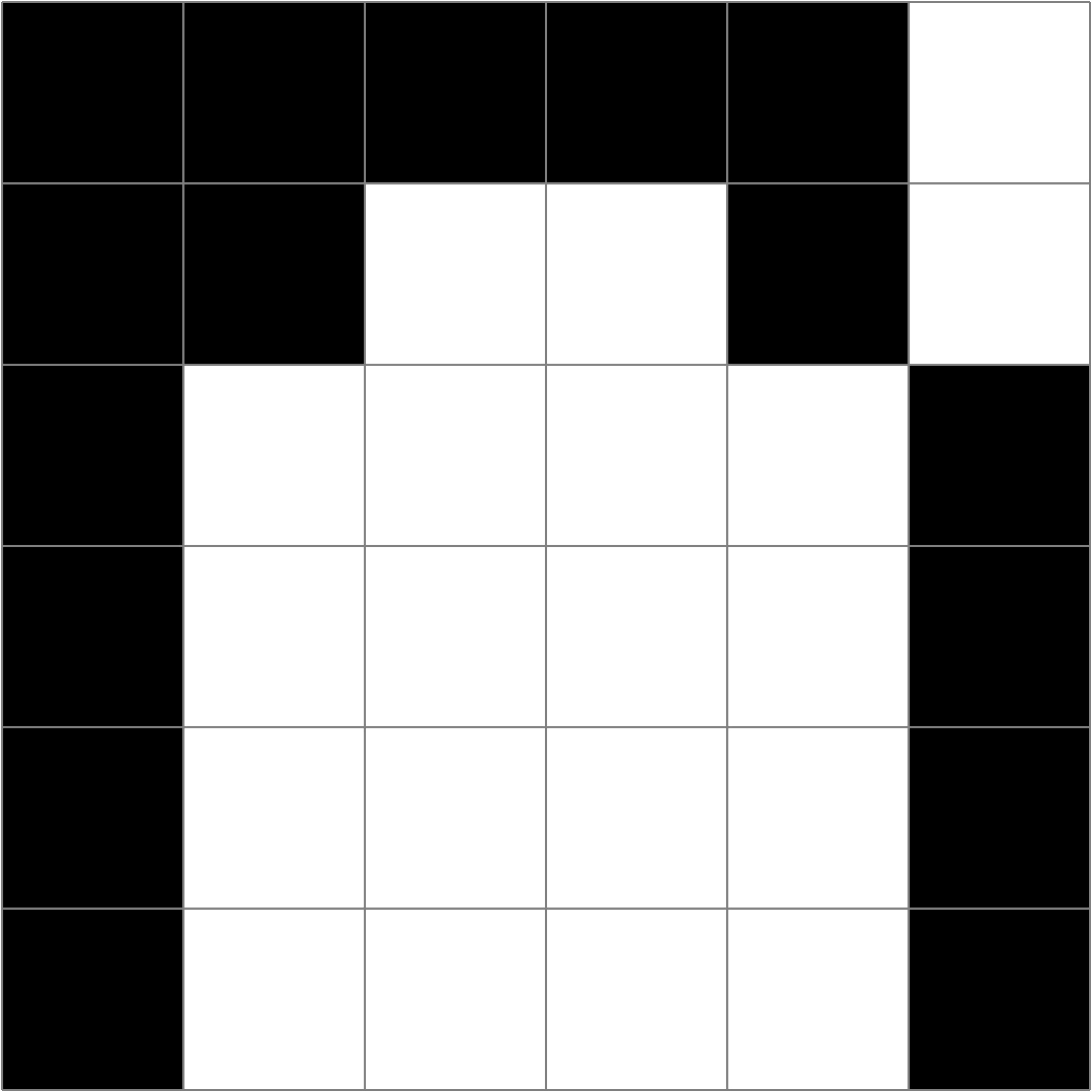}
         \caption{Morphological dilation $\mathcal{D} \oplus \mathcal{S}$}
         \label{fig:morph_dilation}
     \end{subfigure}
               \hspace{5mm}
     \begin{subfigure}[t]{0.25\textwidth}
         \centering
         \includegraphics[width=\textwidth]{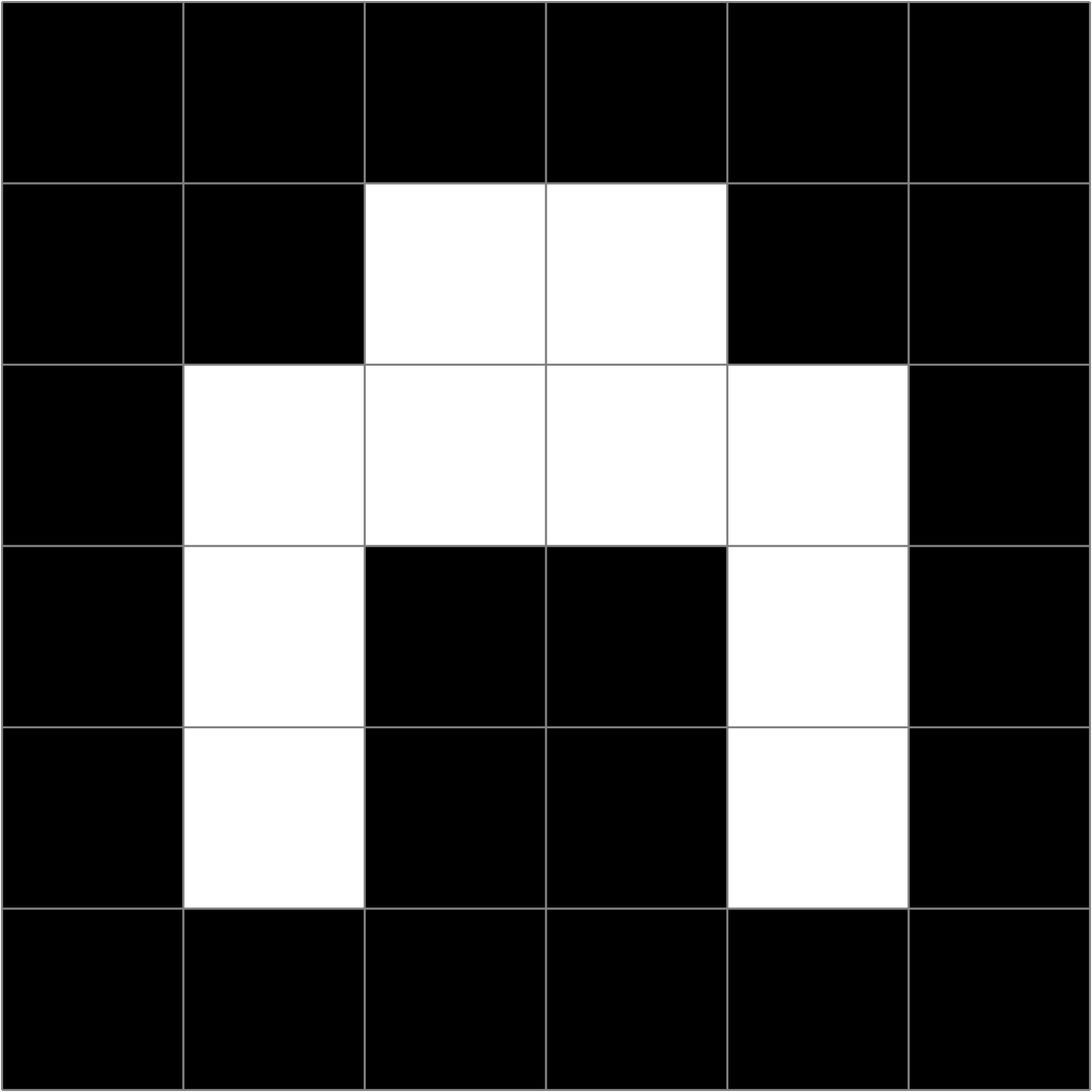}
         \caption{Morphological opening $\mathcal{D} \circ \mathcal{S}$}
         \label{fig:morph_open}
         \end{subfigure}
                   \hspace{5mm}
     \begin{subfigure}[t]{0.25\textwidth}
         \centering
         \includegraphics[width=\textwidth]{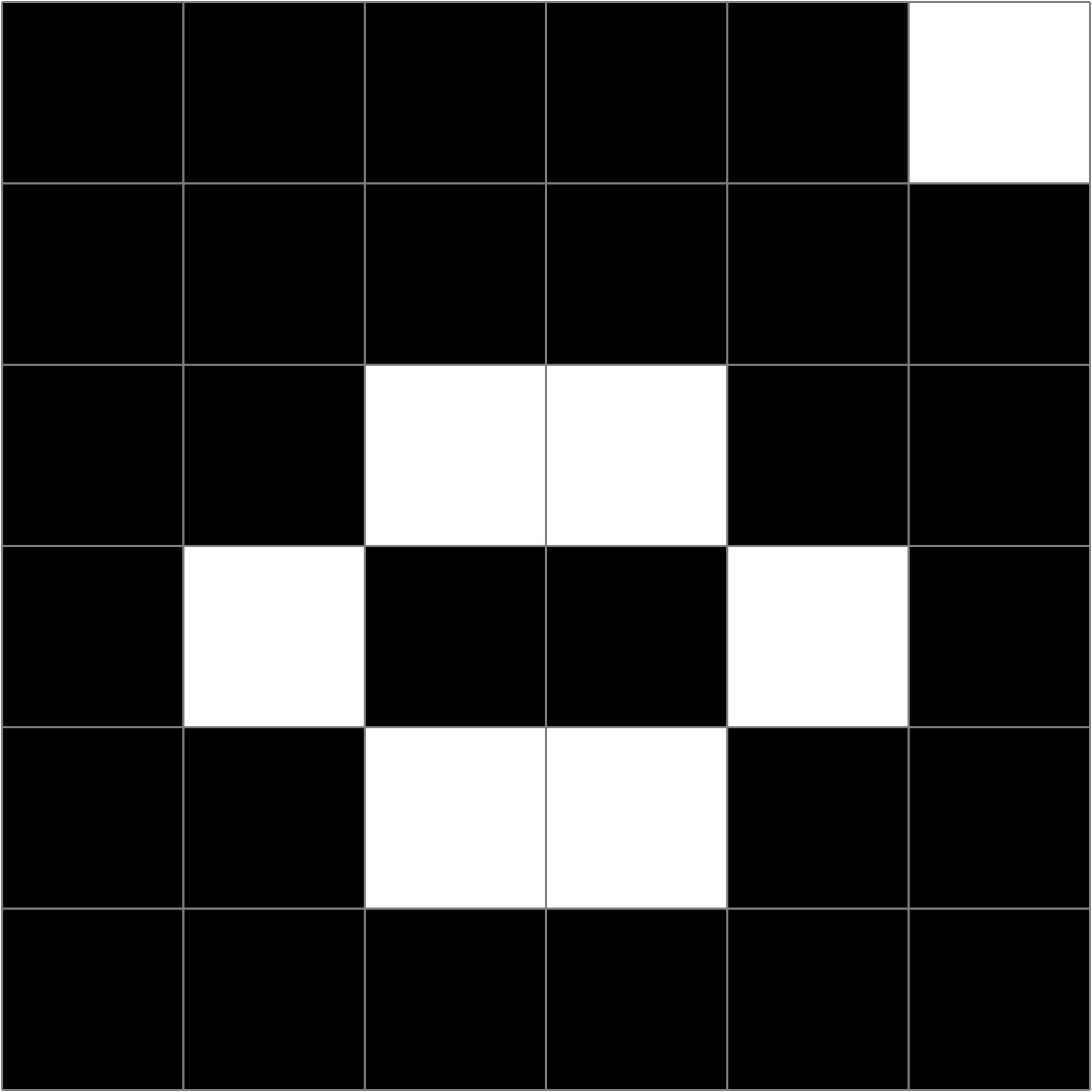}
         \caption{Morphological skeletonization $\mathrm{Skel}(\mathcal{D})$}
         \label{fig:morph_skeleton}
     \end{subfigure}
        \caption{Example of image processing with morphological operations.}
        \label{fig:morph_operations}
\end{figure}

\section{The TILT algorithm} \label{sec:method}

In this section, we show how neighborhoods of boundary components can be estimated in practice when only part of the singularities are known. These neighborhoods can be thought of as 'uncertainty' areas: areas where invisible singularities are supposed to live. For clarity, we suppose that the measuring geometry is fixed so that we see singularities with directions $HL$ and $H\oL$, while singularities with directions $LH$, $HH$, $H\oH$ and $L\oH$ are invisible.  However, the method can be used with any two neighboring directions known or, with slight modifications, the method can be used in cases where more than two directions are known. Before heading to the proposed method, we start by introducing some preprocessing steps and notation. 

\subsection{Preprocessing and notation}

Let $R\in \R^{n\times n}$ be an initial reconstruction computed from the limited-angle tomography data.
Let us denote that $C_{d,W}[R]\in \R^{m\times m}$ is a level $W$ complex wavelet (detail) coefficients matrix of initial reconstruction $R$, where $d\in\{HL, H\oL\}$ denotes the direction and $m=2^W$. 

Let
$C'_{d,W}[R]\in\R^{m\times m}$ elements be
\begin{align*}
     C'_{d,W}[R](i,j)=\frac{| C_{d,W}[R](i,j)|}{\underset{i',j'=1,\dots,n}{\max}| C_{d,W}[R](i',j')|}. 
\end{align*}
We simplify notation $C'_{d,W}[R]$ and simply write $C'_{d}$.
Furthermore, let $C^t_d\in\{0,1\}^{m\times m}$ elements be \begin{align*}
    C^t_d(i,j)=
    \begin{cases}
    1, \text{ if } C'_d(i,j) \geq t \\
    0, \text{ if } C'_d(i,j) < t,
    \end{cases}
\end{align*}
where $t\in\R$.
 
The previous notations should be taken as follows. If a pixel's coefficient is sufficiently large, it is more likely that there is a singularity. To compare complex-valued coefficients in these six coefficient matrices, we take the absolute value of coefficients and normalize them. In this way, we get $C'_d$ for each direction. Now we can choose a threshold-value $t$ to separate singularities from non-singularities. In a matrix $C^t_d$, value $1$ corresponds to singularity and $0$ to non-singularity.

Due to stretching artefacts and imperfect measurements caused by limited-angle problem setting, there might be incorrectly classified pixels. We remove possible isolated singularities and small gaps by performing morphological opening $C^t_d\circ L_{d,l}$, where $L_{d,l}$ is the directed line with line length $l$, see Figure \ref{fig:structurelines}. For simplicity, we denote that $L_{d,l}=L_d$, when line length $l$ is clearly stated.
There may be still small gaps or holes which should be there. Thus we perform morphological dilation with the small $3\times3-$matrix $S$, $S(i,j)=1$ for all $i$ and $j$, as the structuring element.

\begin{figure}
     \centering
     \begin{subfigure}[b]{0.15\textwidth}
         \centering
         \includegraphics[width=\textwidth]{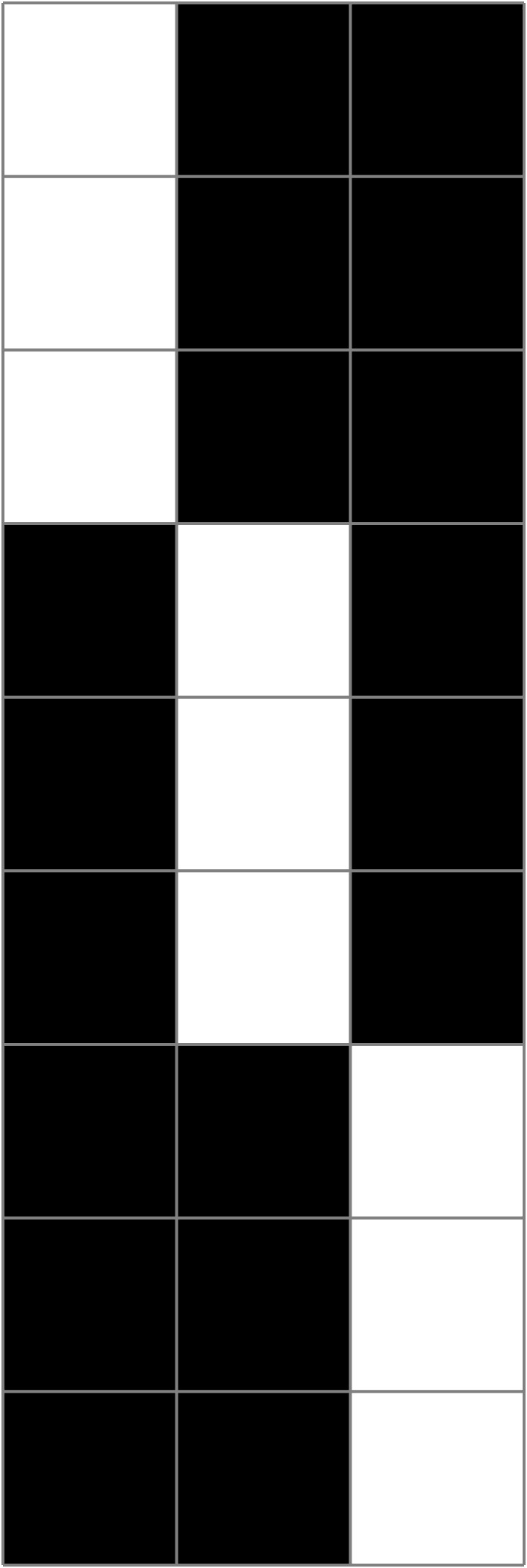}
         \caption{$L_{HL,9}$}
         \label{fig:line_0830}
     \end{subfigure}
     \hspace{2cm}
     \begin{subfigure}[b]{0.15\textwidth}
         \centering
         \includegraphics[width=\textwidth]{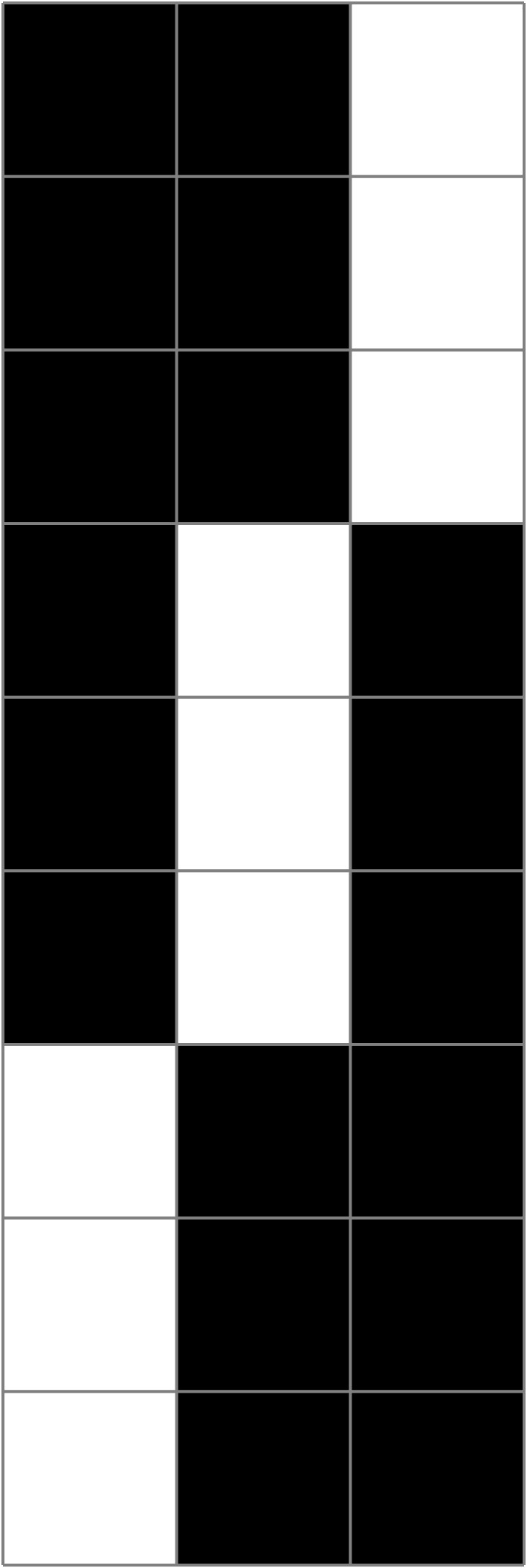}
         \caption{$L_{H\oL,9}$}
         \label{fig:line_0930}
     \end{subfigure}
        \caption{Isolated points can be removed using directed lines as structuring elements in the morphological opening.}
        \label{fig:structurelines}
\end{figure}

\begin{figure}
     \centering
     \begin{subfigure}[t]{0.15\textwidth}
         \centering
         \includegraphics[width=\textwidth]{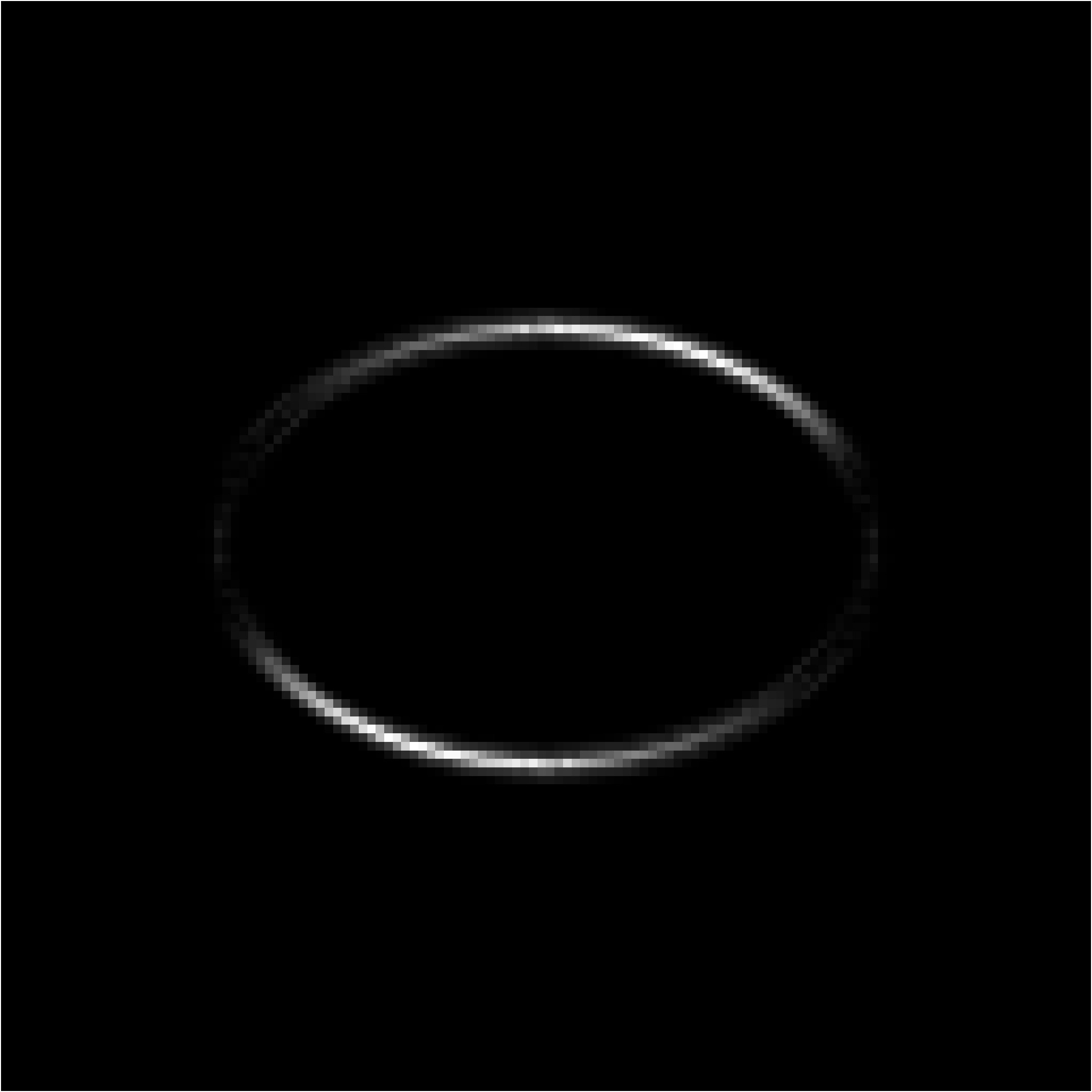}
         \caption{$C'_{LH}$}
         \label{fig:SBa}
     \end{subfigure}
     \begin{subfigure}[t]{0.15\textwidth}
         \centering
         \includegraphics[width=\textwidth]{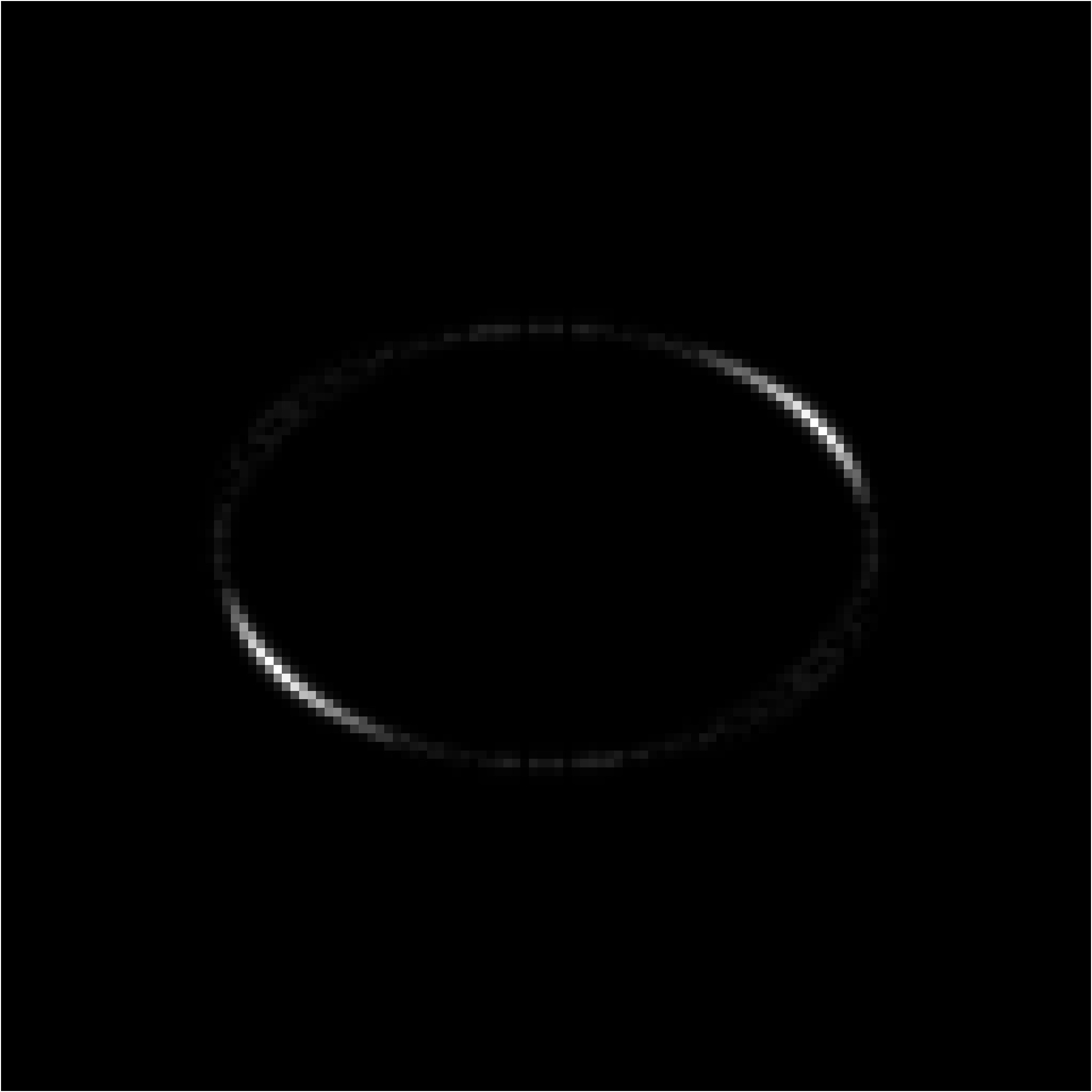}
         \caption{$C'_{HH}$}
         \label{fig:SBb}
     \end{subfigure}
     \begin{subfigure}[t]{0.15\textwidth}         
     \centering
         \includegraphics[width=\textwidth]{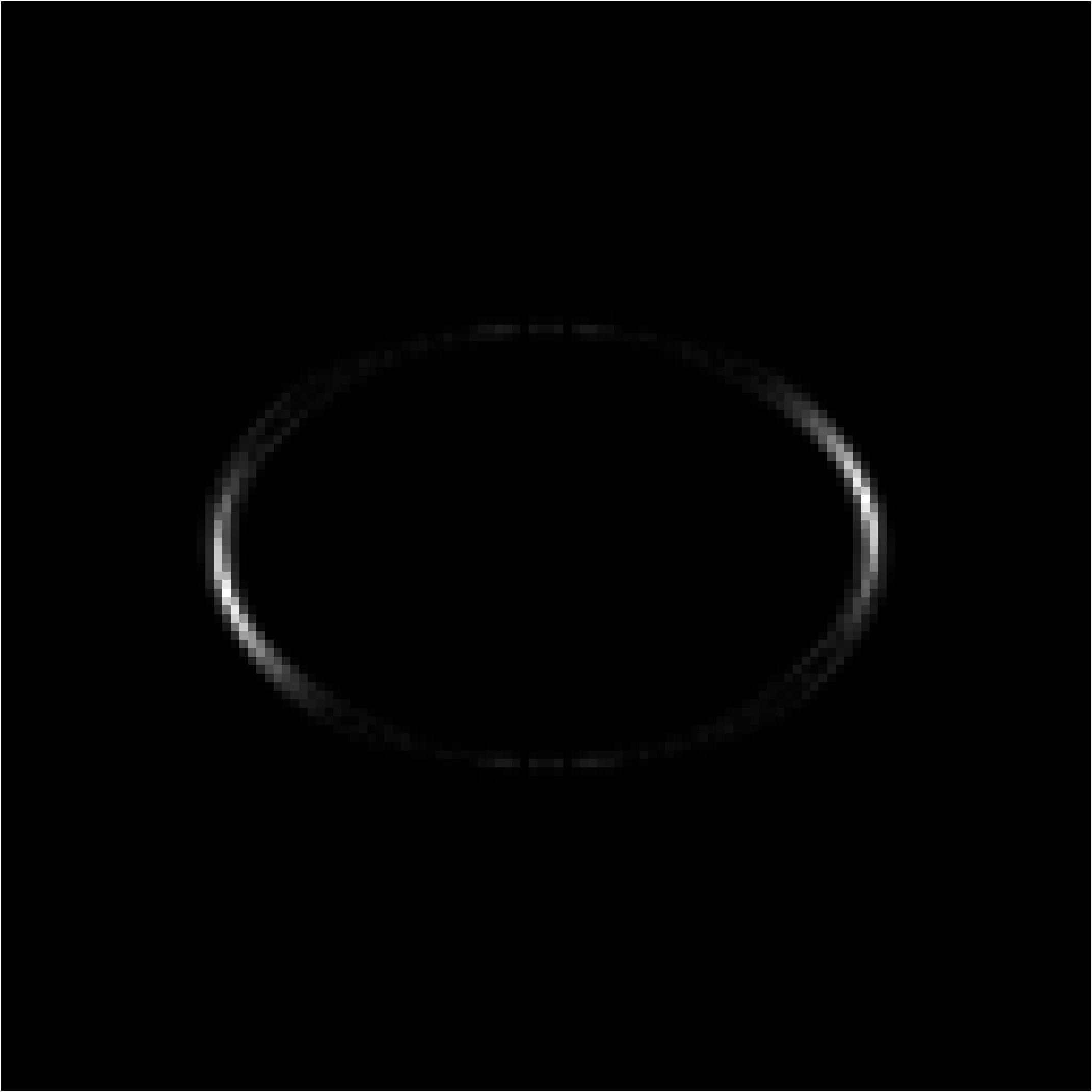}
         \caption{$C'_{HL}$}
         \label{fig:SBc}
     \end{subfigure}
     \begin{subfigure}[t]{0.15\textwidth}
         \centering
         \includegraphics[width=\textwidth]{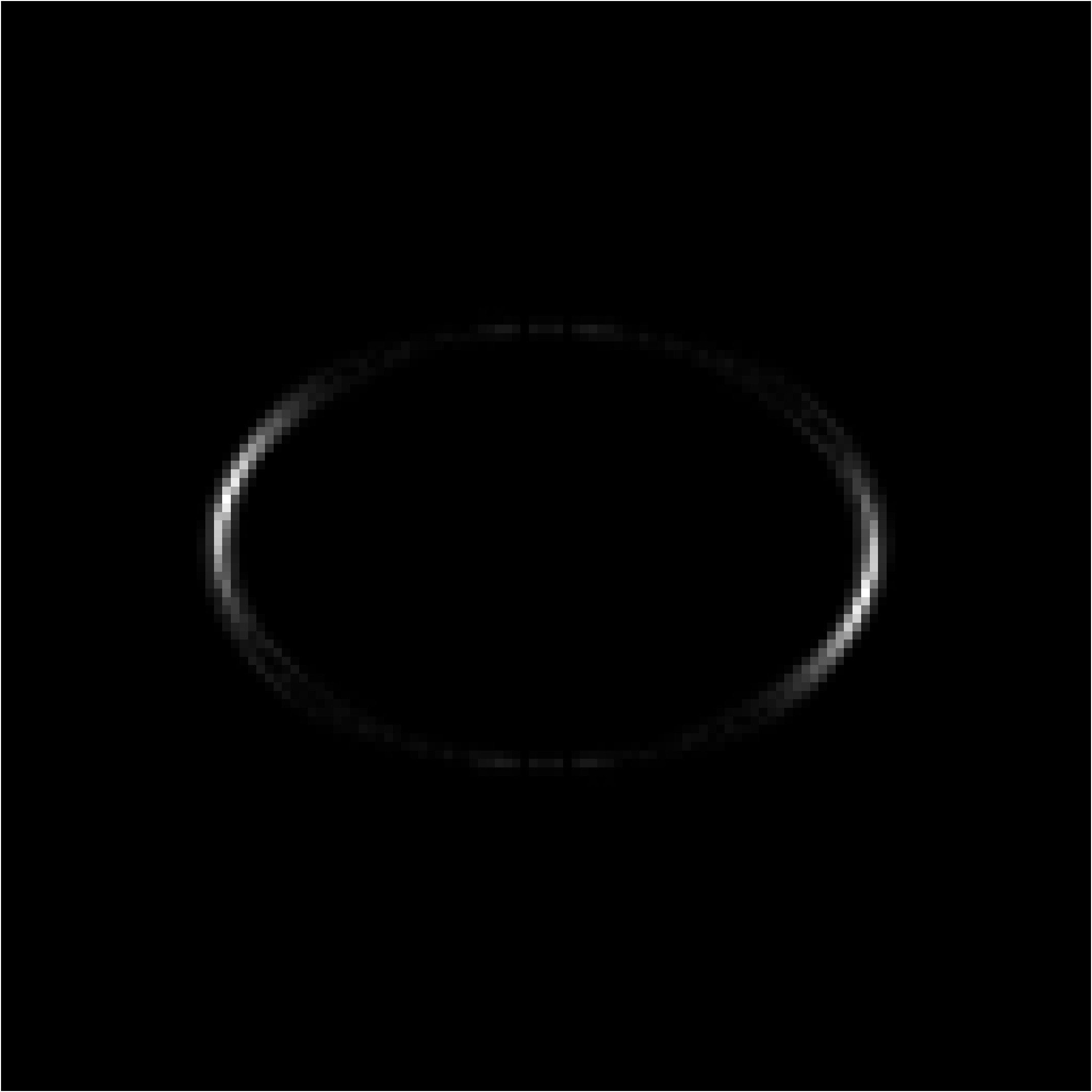}
         \caption{$C'_{H\oL}$}
         \label{fig:SBd}
     \end{subfigure}
     \begin{subfigure}[t]{0.15\textwidth}
         \centering
         \includegraphics[width=\textwidth]{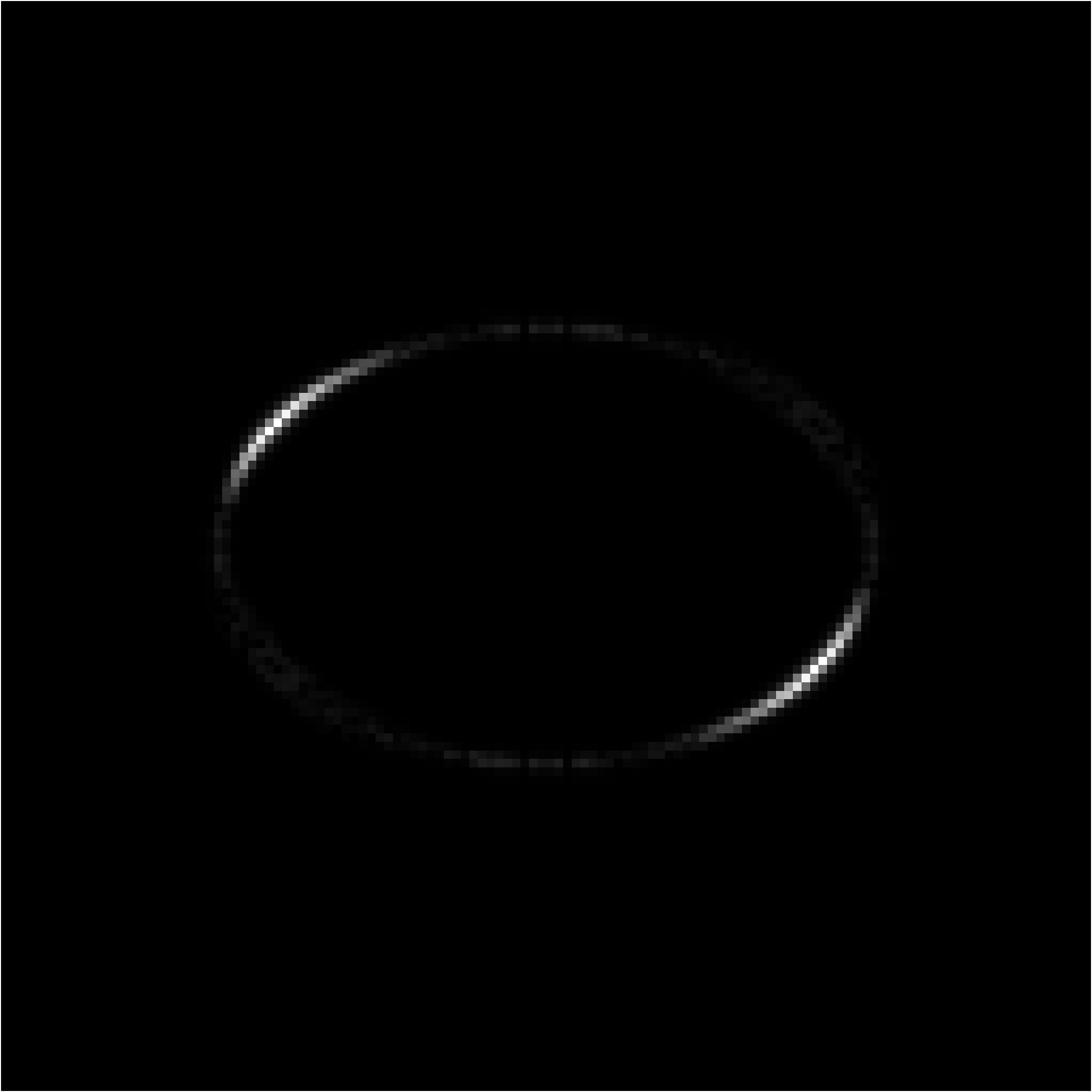}
         \caption{$C'_{H\oH}$}
         \label{fig:SBe}
         \end{subfigure}
     \begin{subfigure}[t]{0.15\textwidth}
         \centering
         \includegraphics[width=\textwidth]{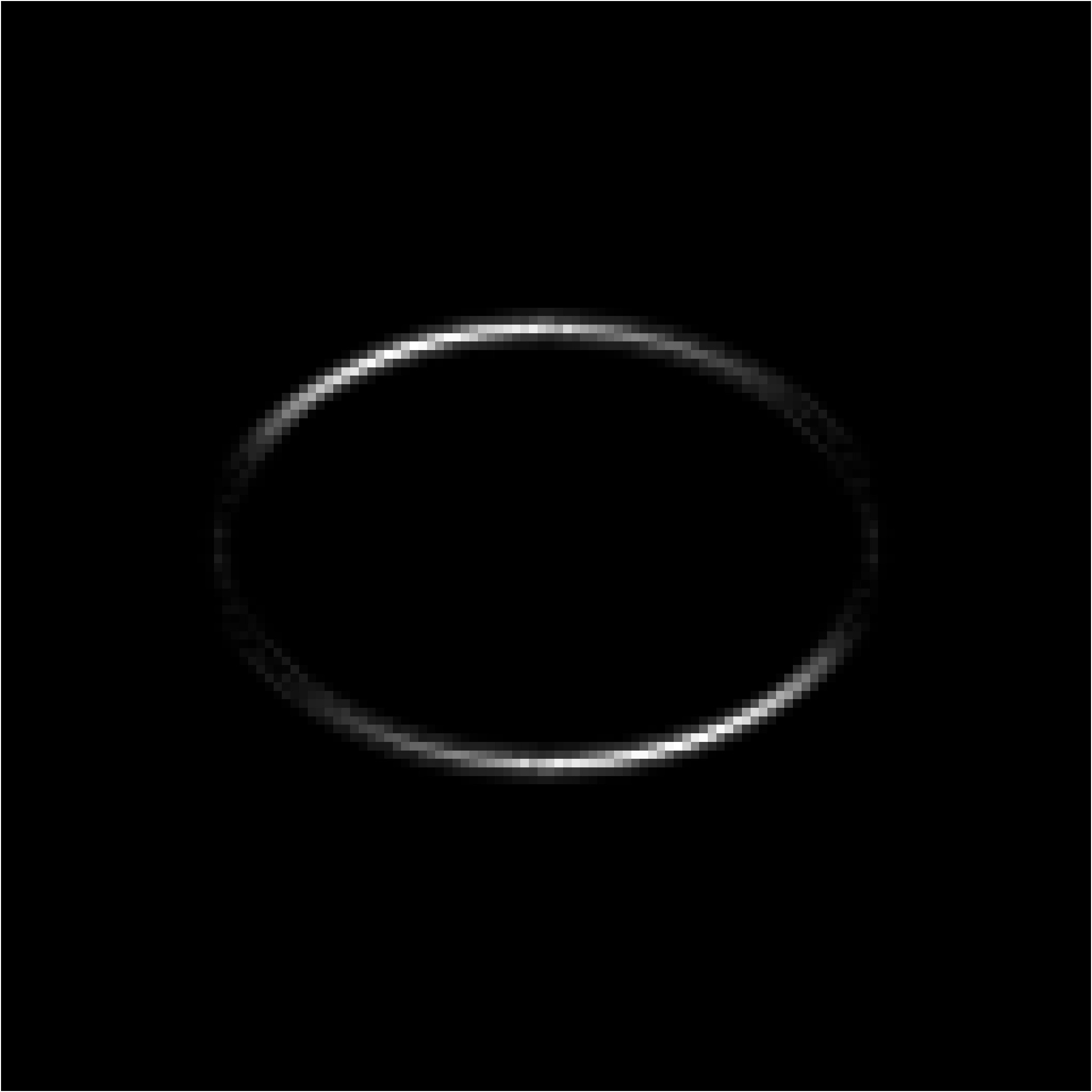}
         \caption{$C'_{L\oH}$}
         \label{fig:SBf}
     \end{subfigure}
     \vfill
     \begin{subfigure}[t]{0.15\textwidth}
         \centering
         \includegraphics[width=\textwidth]{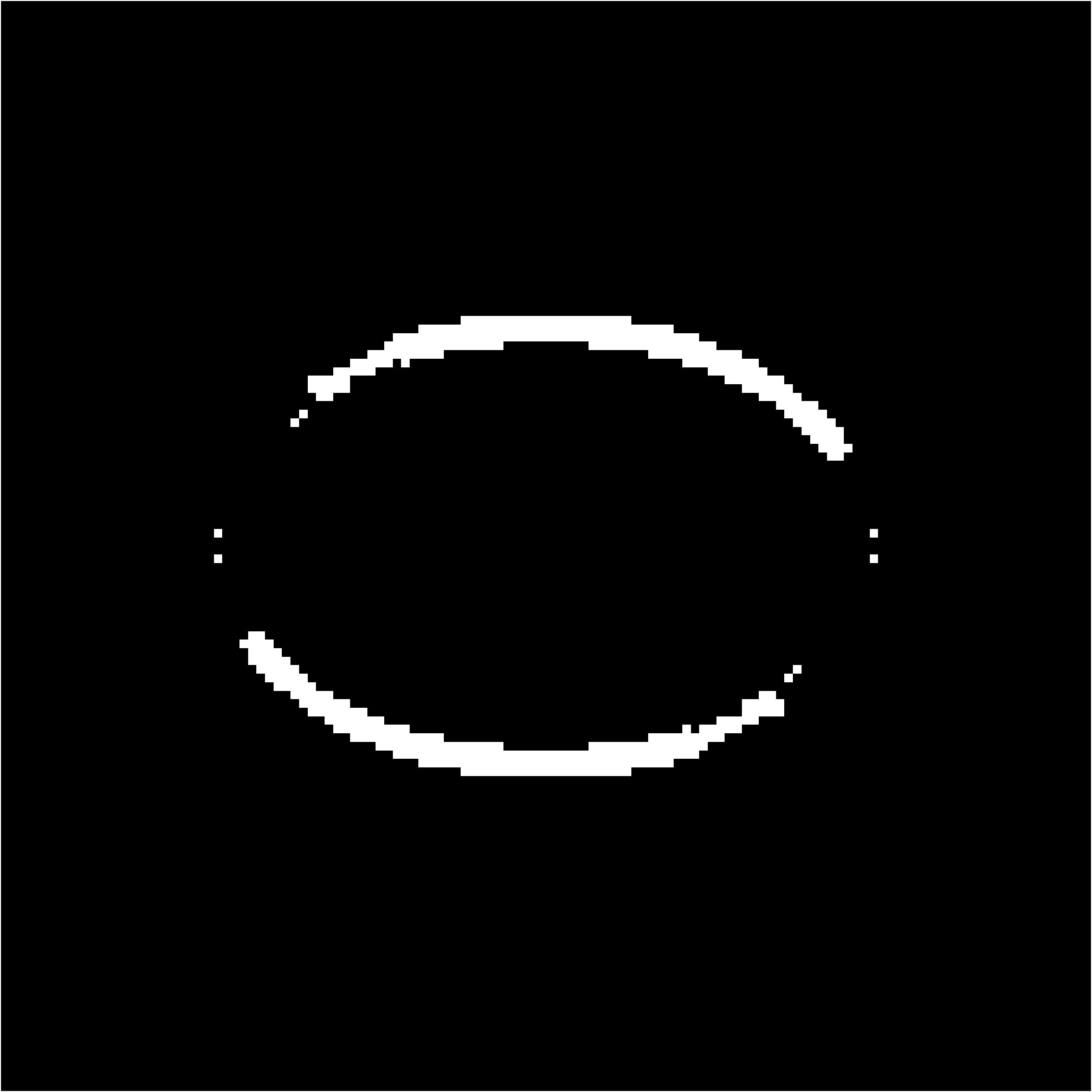}
         \caption{$C^{0.1}_{LH}$}
         \label{fig:SBa_t}
     \end{subfigure}
     \begin{subfigure}[t]{0.15\textwidth}
         \centering
         \includegraphics[width=\textwidth]{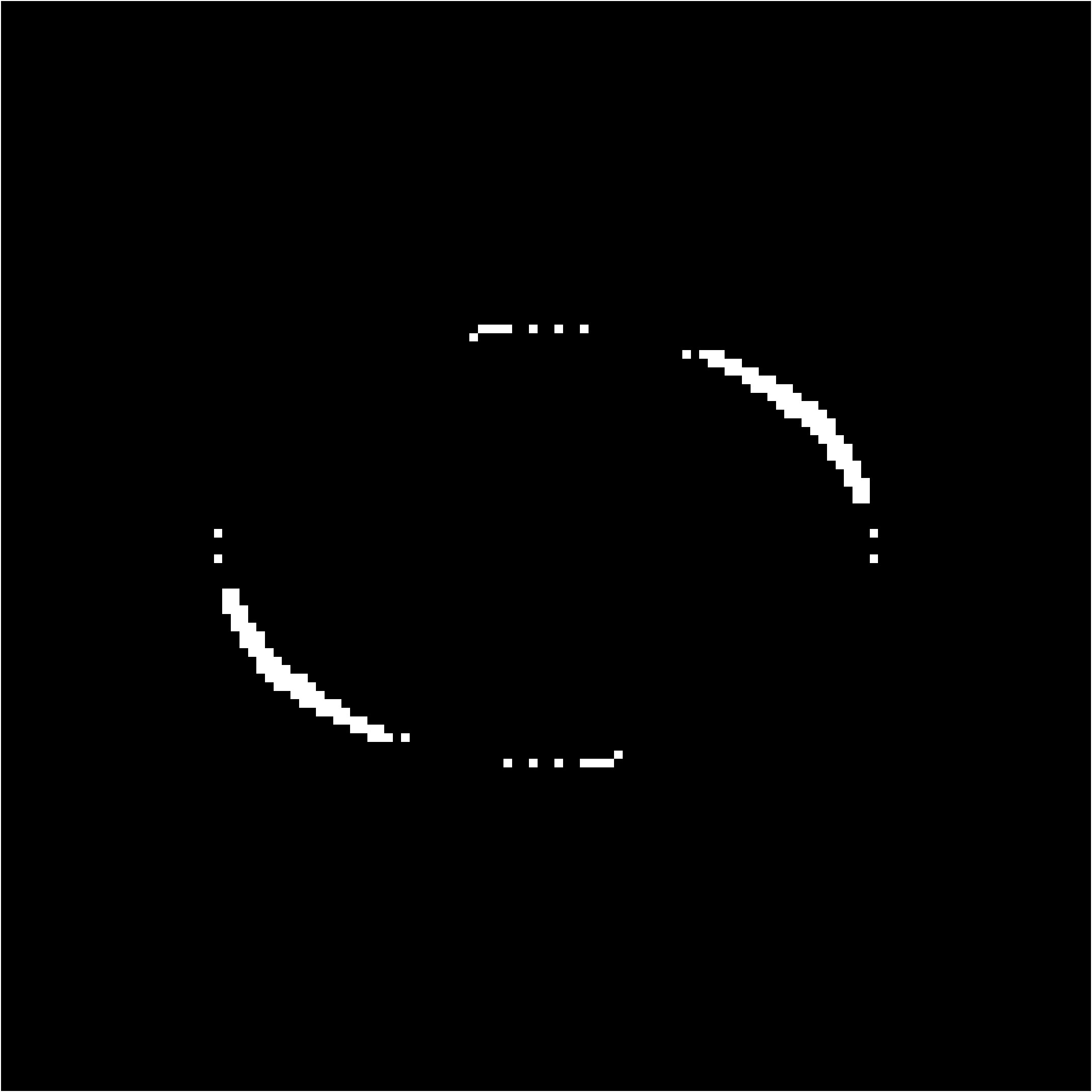}
         \caption{$C^{0.1}_{HH}$}
         \label{fig:SBb_t}
     \end{subfigure}
     \begin{subfigure}[t]{0.15\textwidth}         
     \centering
         \includegraphics[width=\textwidth]{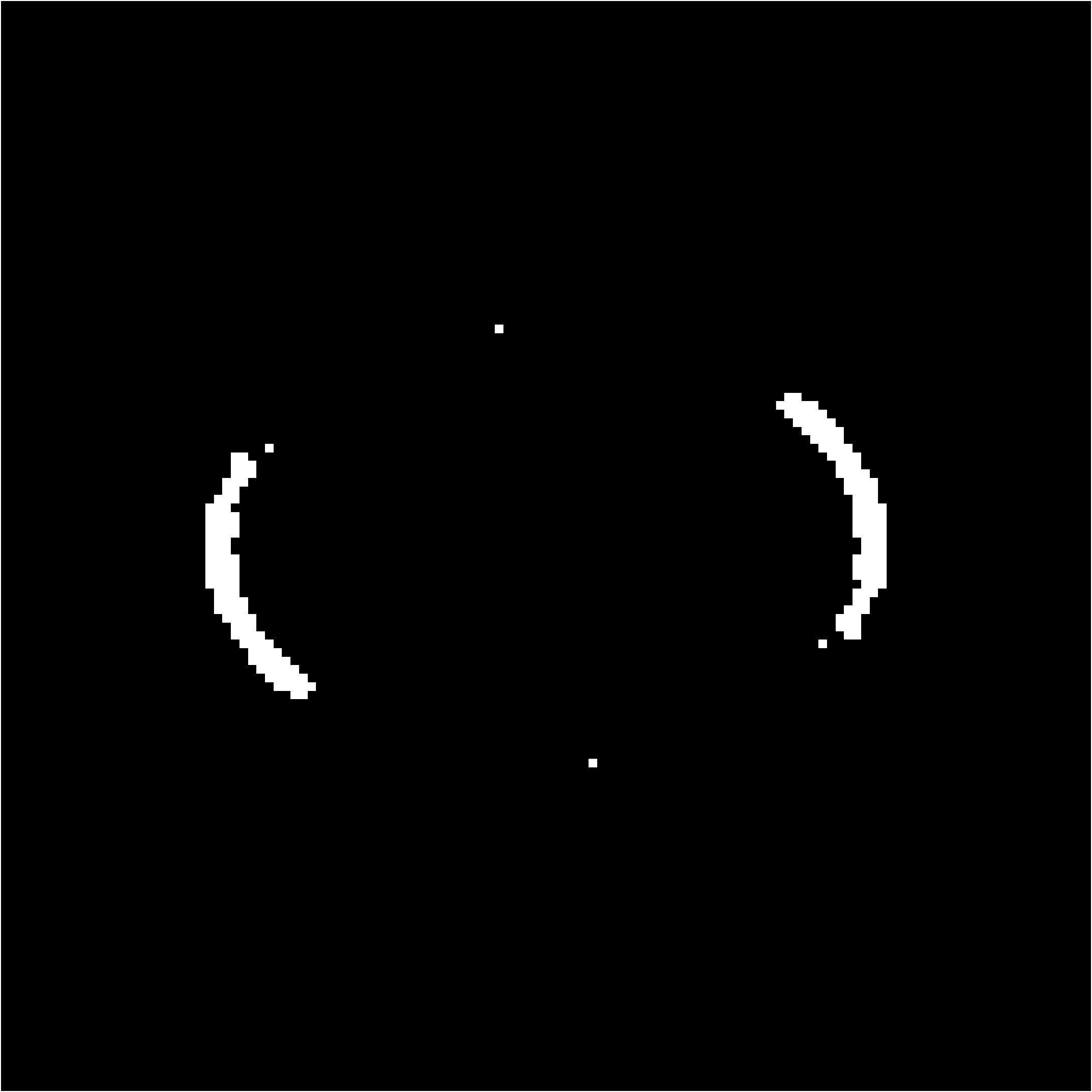}
         \caption{$C^{0.1}_{HL}$}
         \label{fig:SBc_t}
     \end{subfigure}
     \begin{subfigure}[t]{0.15\textwidth}
         \centering
         \includegraphics[width=\textwidth]{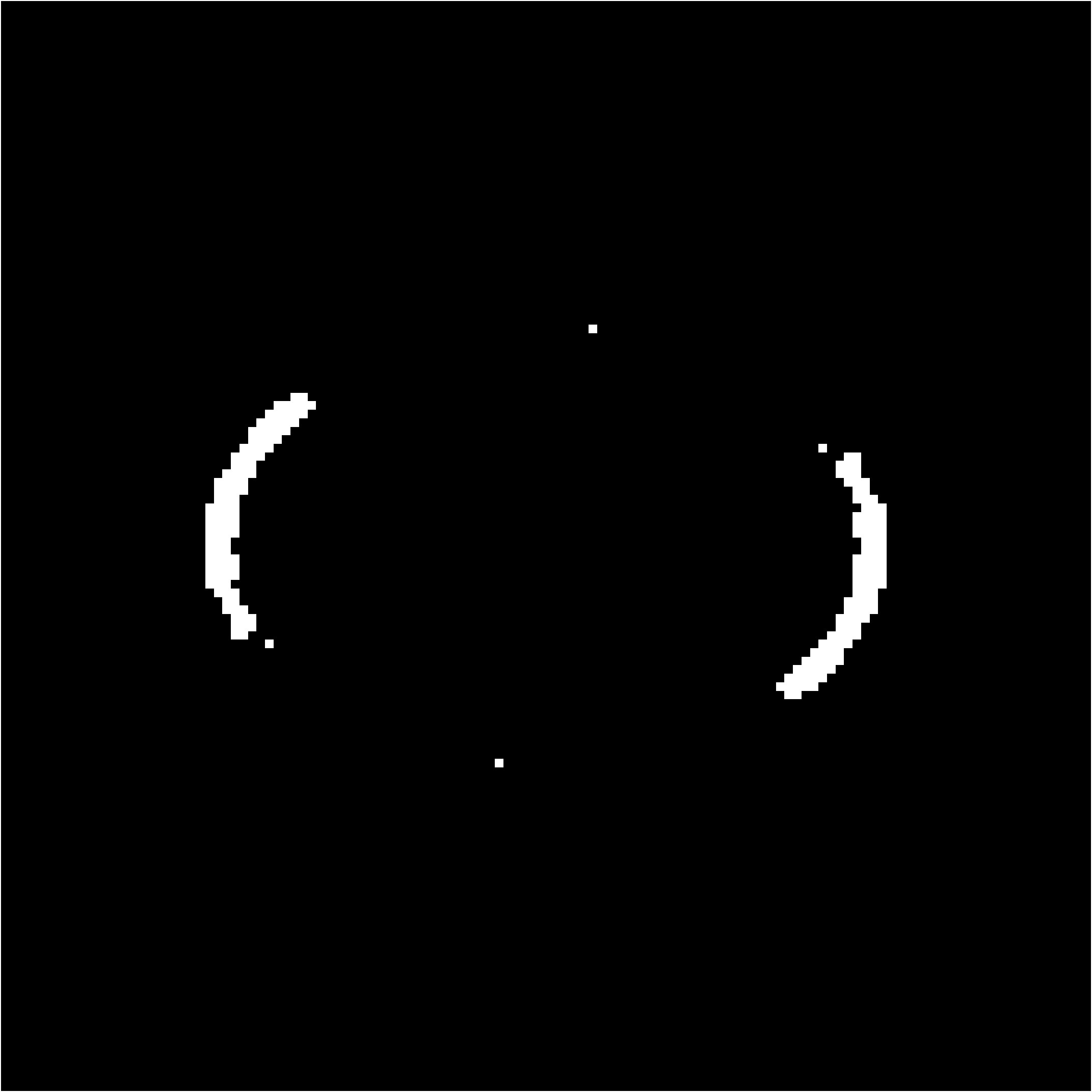}
         \caption{$C^{0.1}_{H\oL}$}
         \label{fig:SBd_t}
     \end{subfigure}
     \begin{subfigure}[t]{0.15\textwidth}
         \centering
         \includegraphics[width=\textwidth]{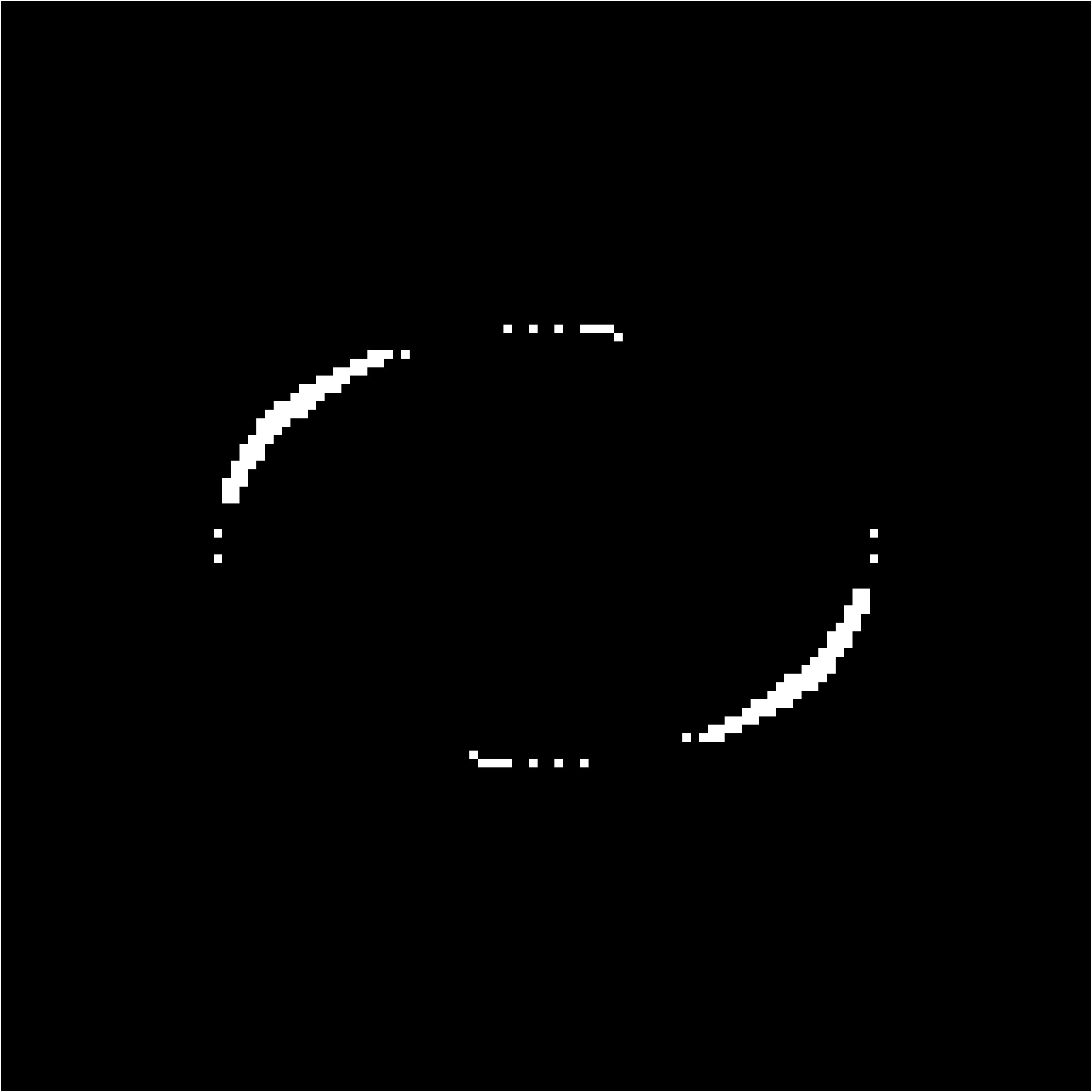}
         \caption{$C^{0.1}_{H\oH}$}
         \label{fig:SBe_t}
         \end{subfigure}
     \begin{subfigure}[t]{0.15\textwidth}
         \centering
         \includegraphics[width=\textwidth]{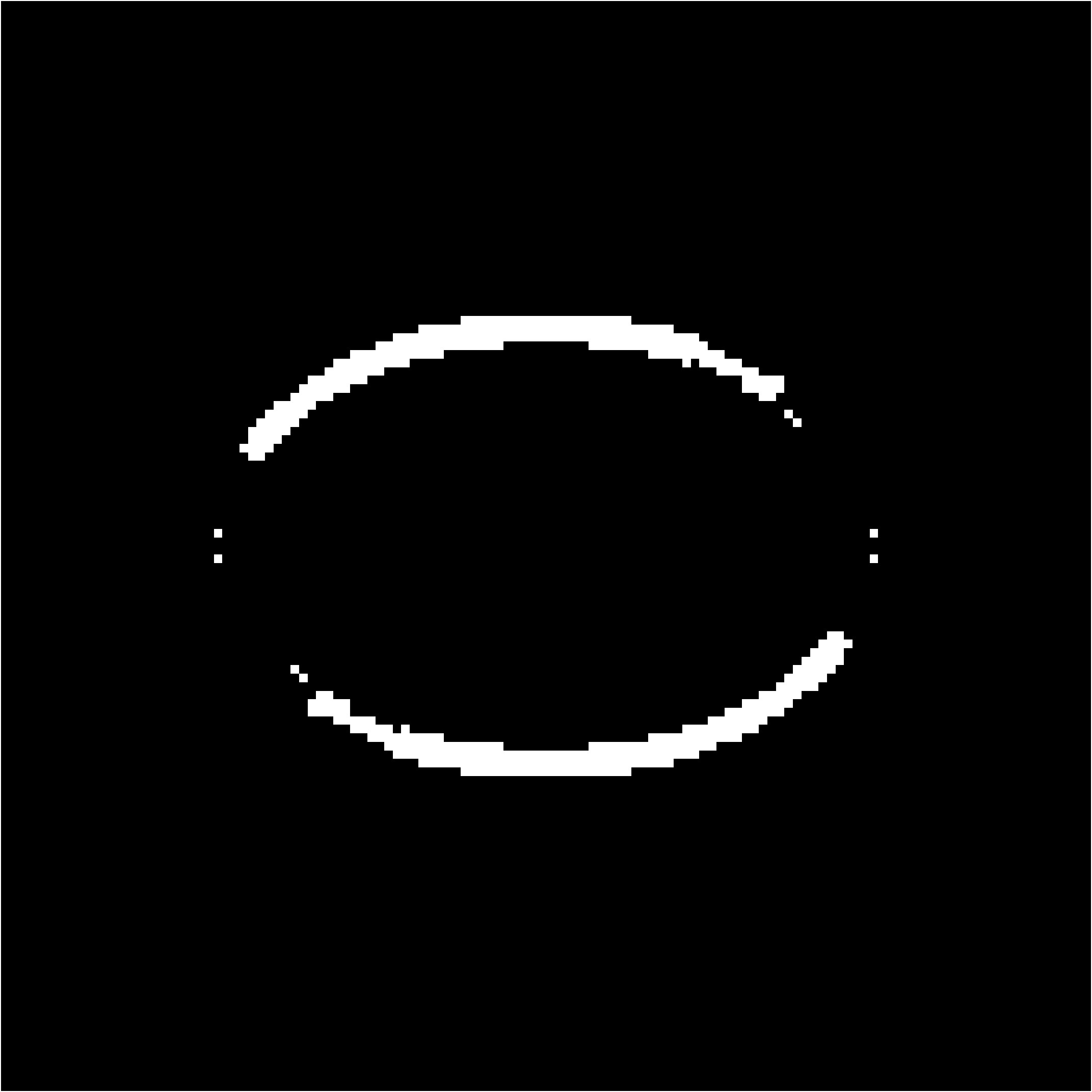}
         \caption{$C^{0.1}_{L\oH}$}
         \label{fig:SBf_t}
     \end{subfigure}
     \begin{subfigure}[t]{0.15\textwidth}
         \centering
         \includegraphics[width=\textwidth]{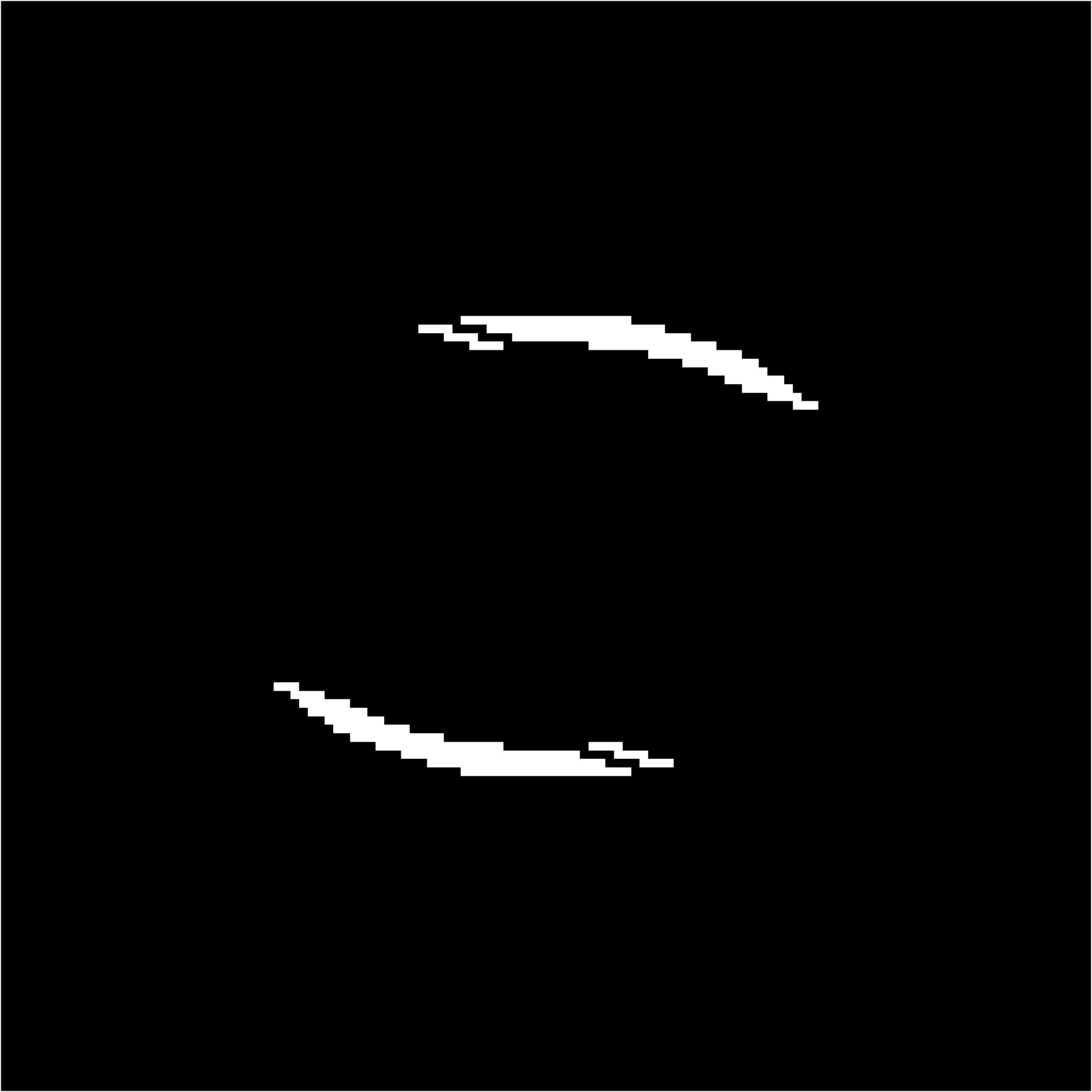}
         \caption{$C^{0.1}_{LH}\circ L_{LH}$}
         \label{fig:SBa_l}
     \end{subfigure}
     \begin{subfigure}[t]{0.15\textwidth}
         \centering
         \includegraphics[width=\textwidth]{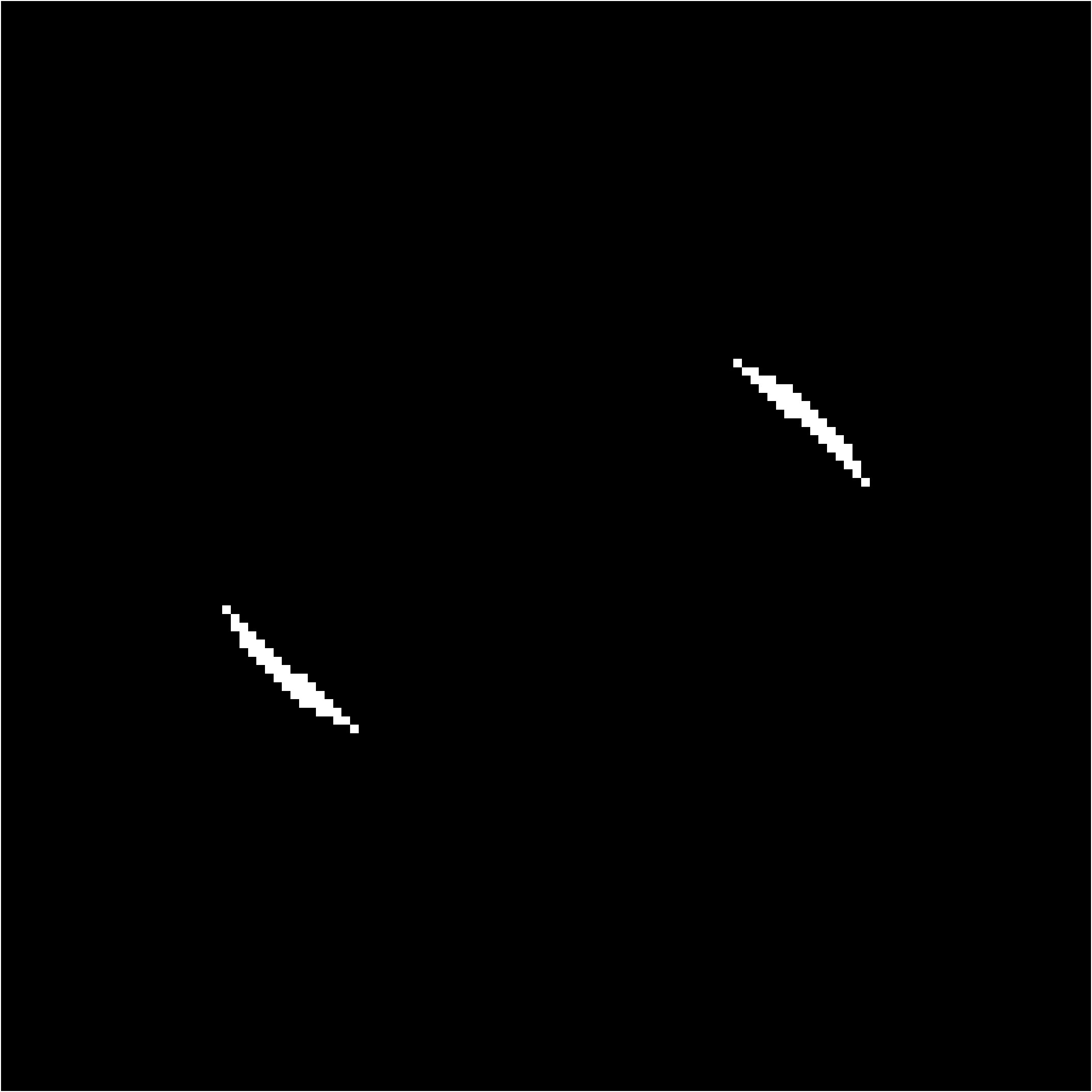}
         \caption{$C^{0.1}_{HH}\circ L_{HH}$}
         \label{fig:SBb_l}
     \end{subfigure}
     \begin{subfigure}[t]{0.15\textwidth}         
     \centering
         \includegraphics[width=\textwidth]{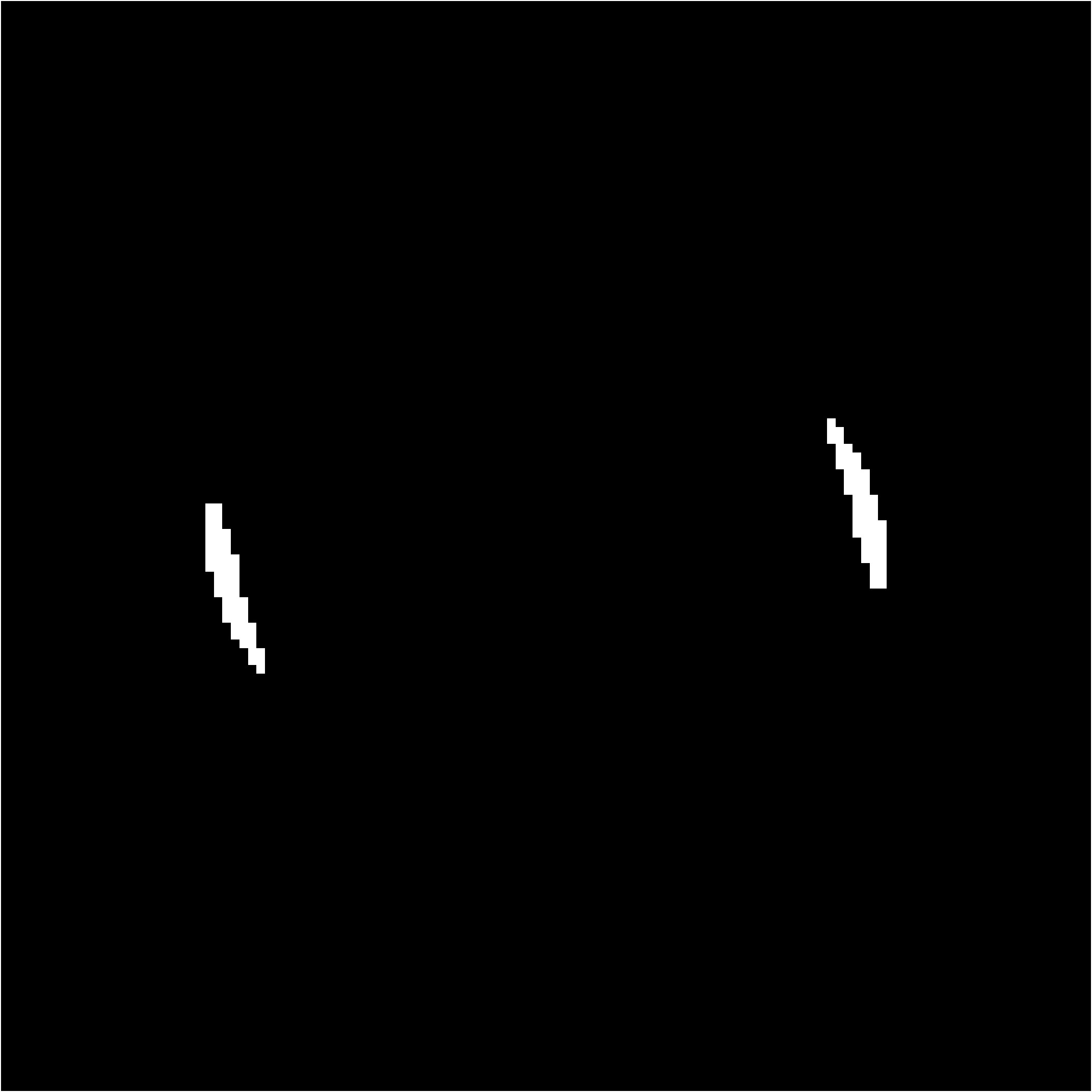}
         \caption{$C^{0.1}_{HL}\circ L_{HL}$}
         \label{fig:SBc_l}
     \end{subfigure}
     \begin{subfigure}[t]{0.15\textwidth}
         \centering
         \includegraphics[width=\textwidth]{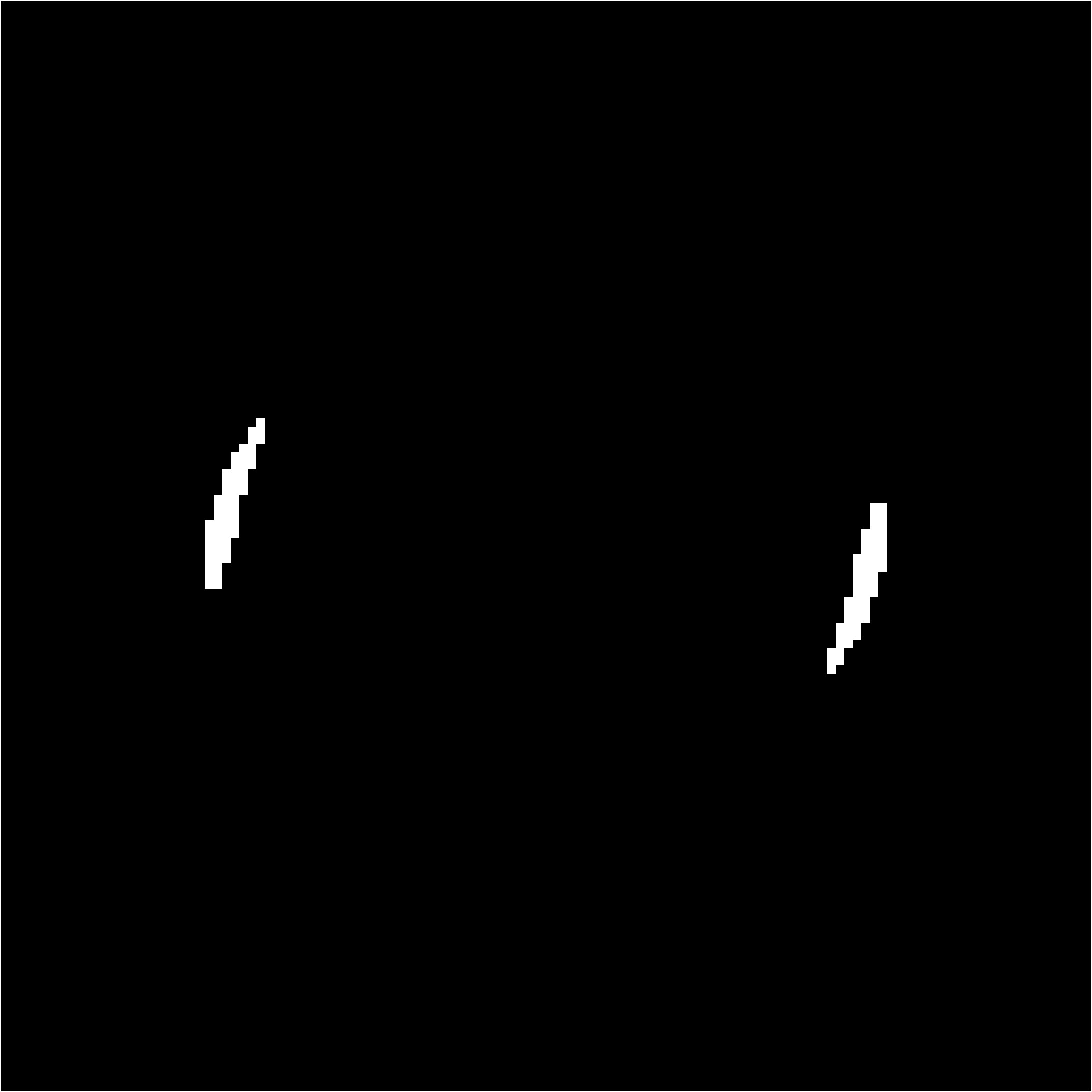}
         \caption{$C^{0.1}_{H\oL}\circ L_{H\oL}$}
         \label{fig:SBd_l}
     \end{subfigure}
     \begin{subfigure}[t]{0.15\textwidth}
         \centering
         \includegraphics[width=\textwidth]{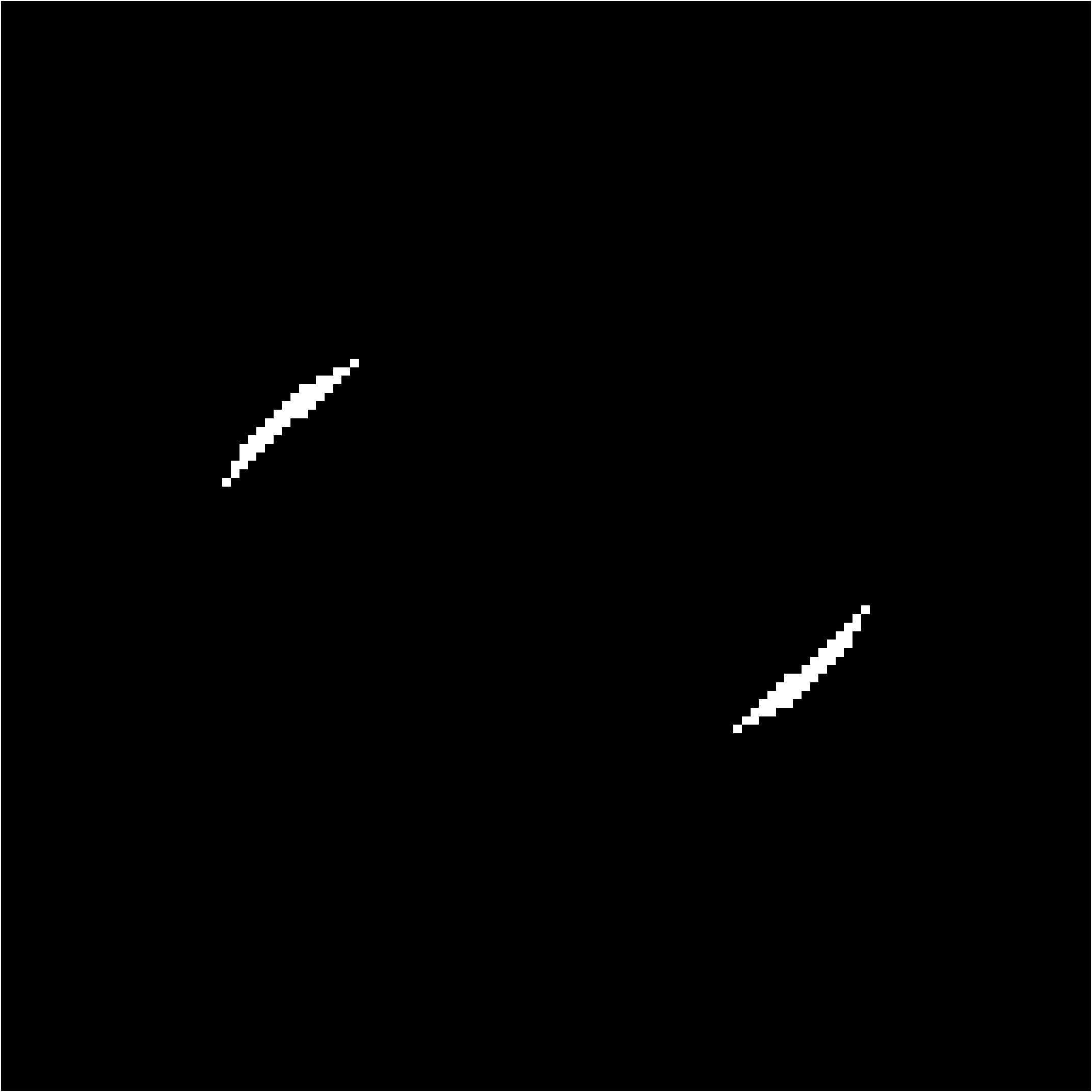}
         \caption{$C^{0.1}_{H\oH}\circ L_{H\oH}$}
         \label{fig:SBe_l}
         \end{subfigure}
     \begin{subfigure}[t]{0.15\textwidth}
         \centering
         \includegraphics[width=\textwidth]{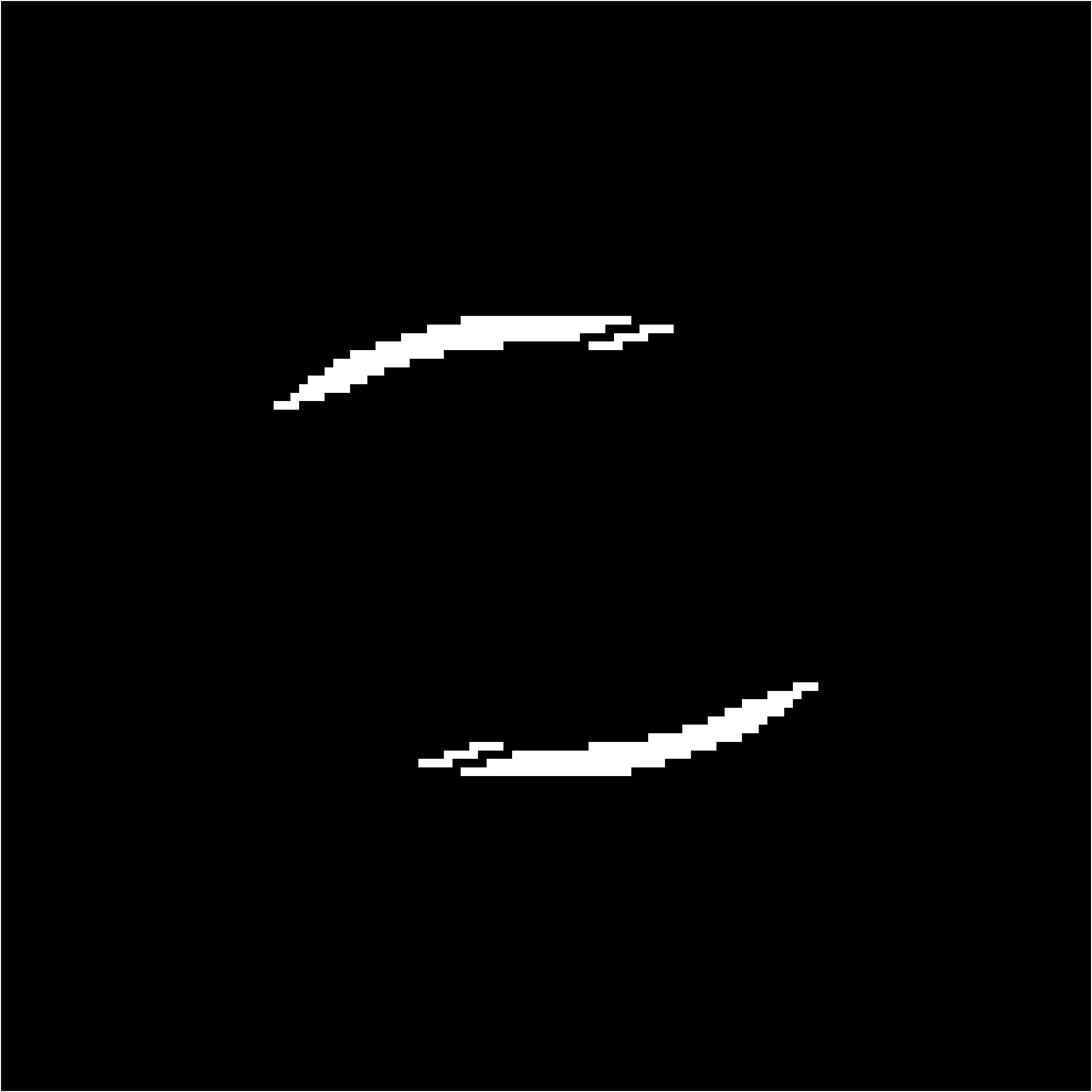}
         \caption{$C^{0.1}_{L\oH}\circ L_{L\oH}$}
         \label{fig:SBf_l}
     \end{subfigure}
     \begin{subfigure}[t]{0.15\textwidth}
         \centering
         \includegraphics[width=\textwidth]{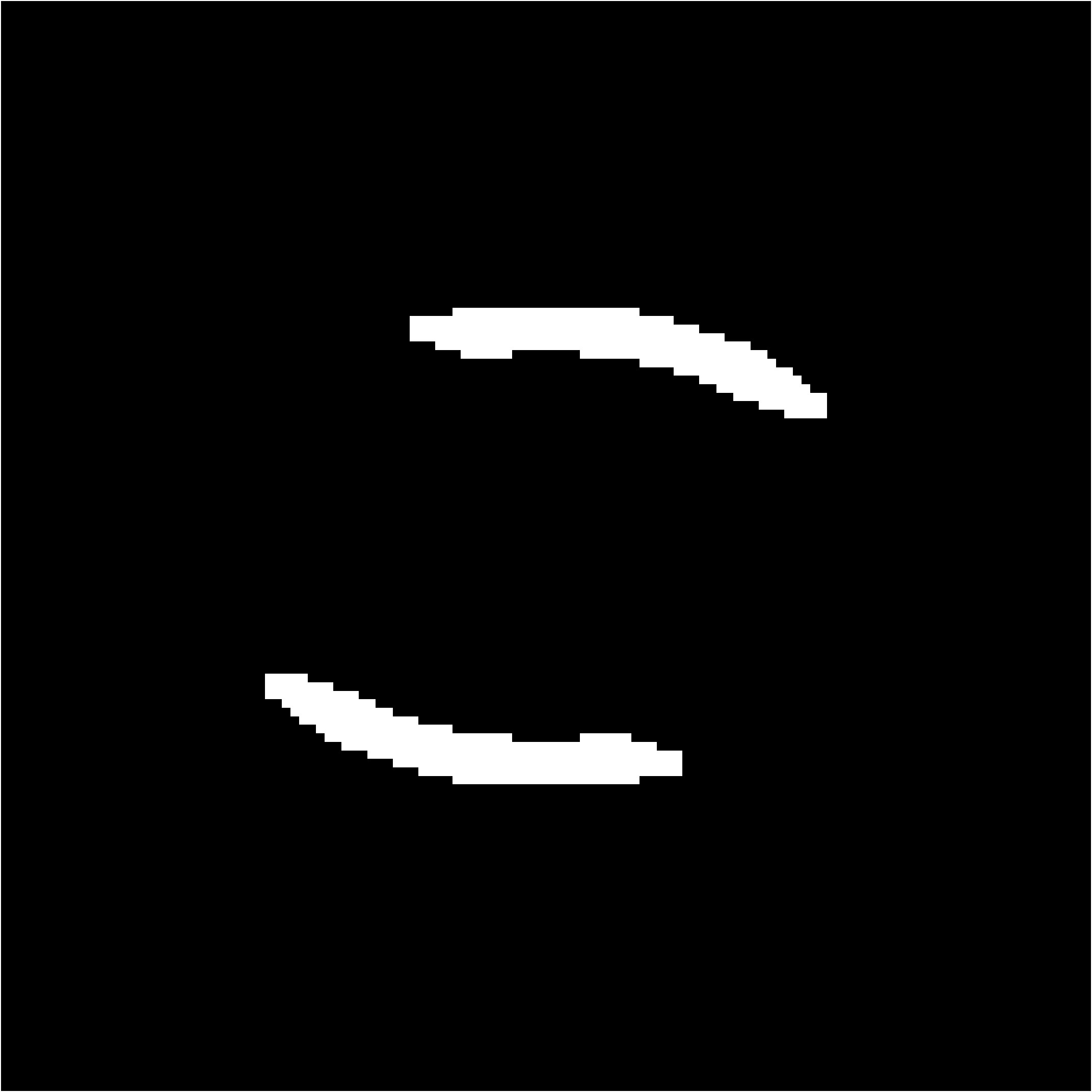}
         \caption{$SB^{0.1}_{LH}$}
         \label{fig:SBa_s}
     \end{subfigure}
     \begin{subfigure}[t]{0.15\textwidth}
         \centering
         \includegraphics[width=\textwidth]{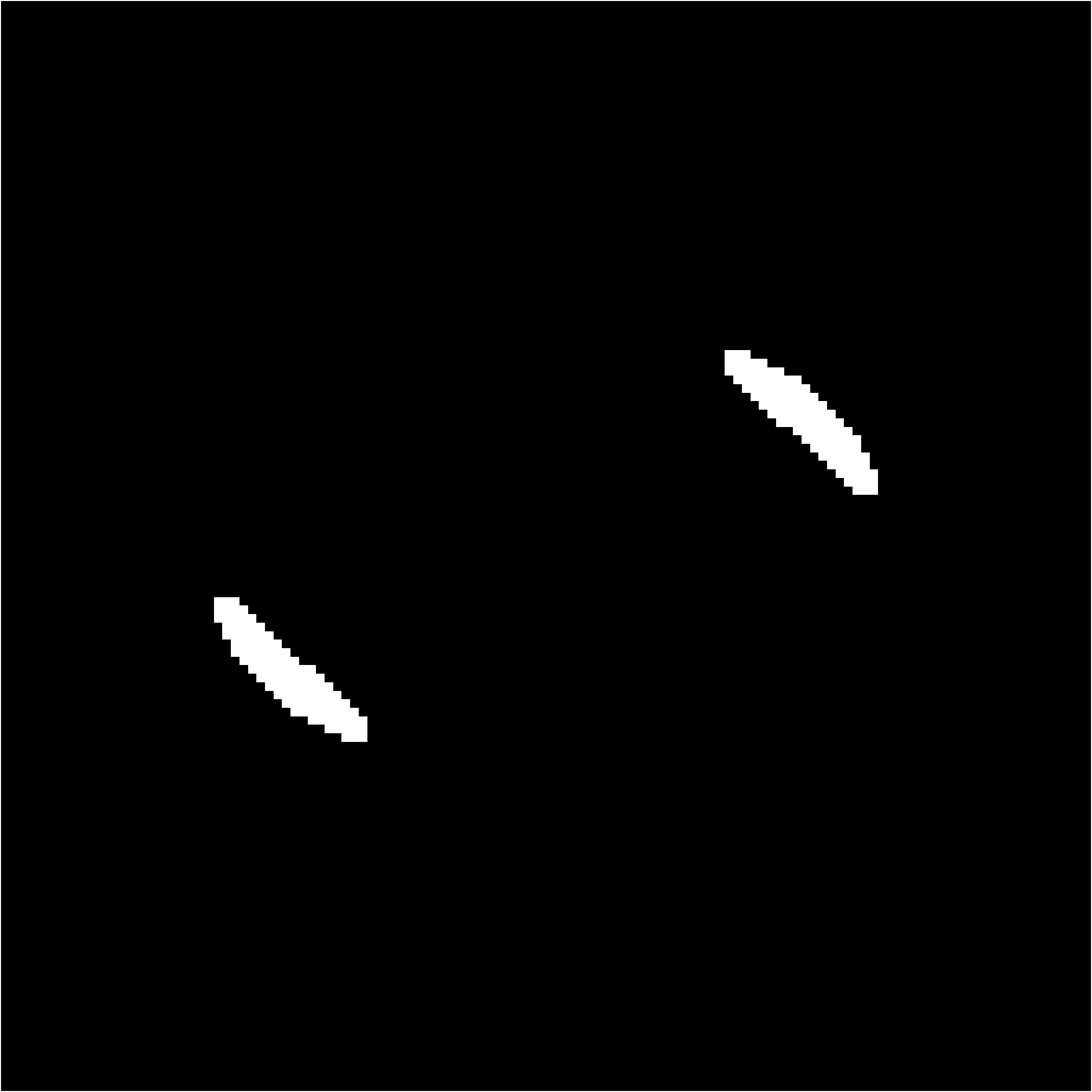}
         \caption{$SB^{0.1}_{HH}$}
         \label{fig:SBb_s}
     \end{subfigure}
     \begin{subfigure}[t]{0.15\textwidth}         
     \centering
         \includegraphics[width=\textwidth]{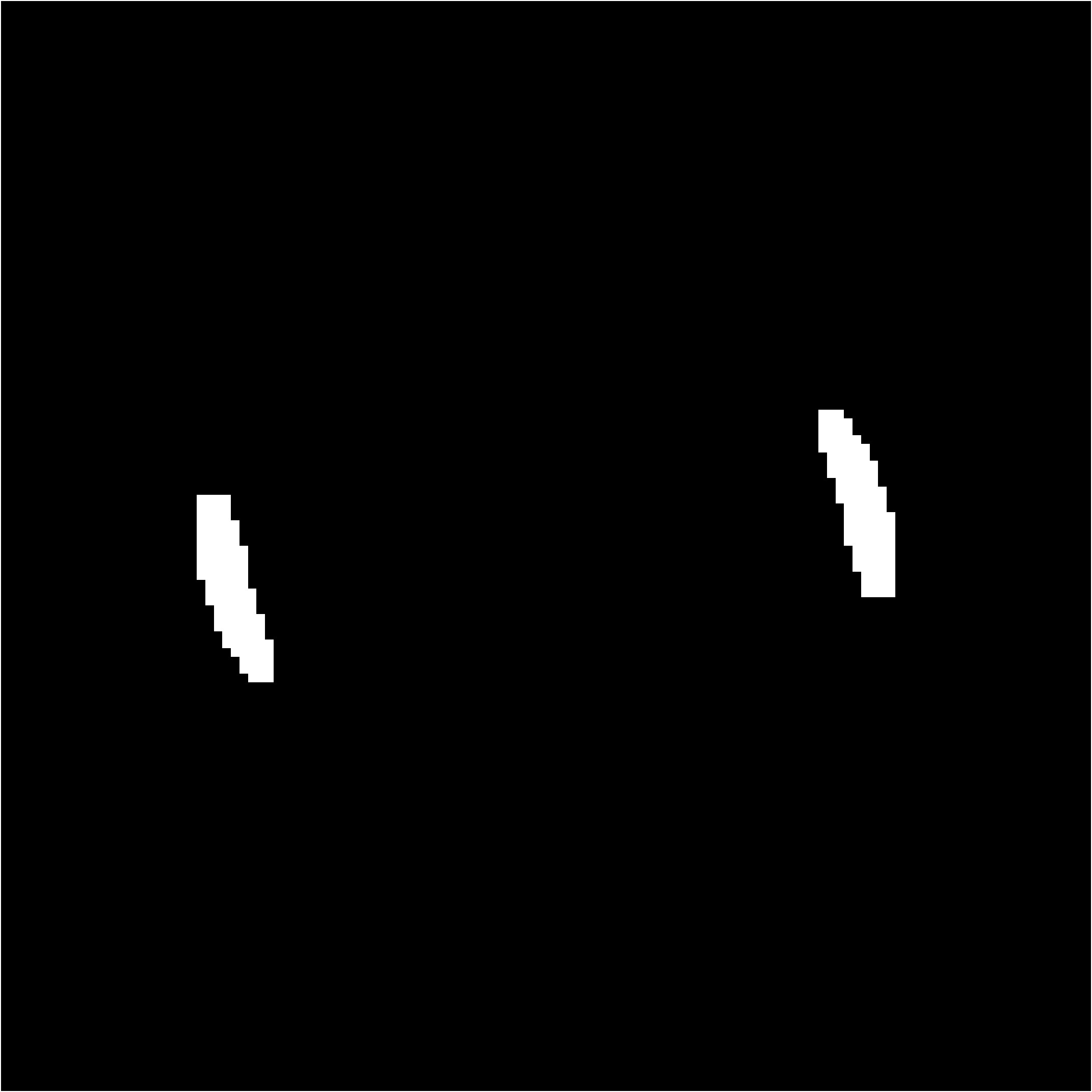}
         \caption{$SB^{0.1}_{HL}$}
         \label{fig:SBc_s}
     \end{subfigure}
     \begin{subfigure}[t]{0.15\textwidth}
         \centering
         \includegraphics[width=\textwidth]{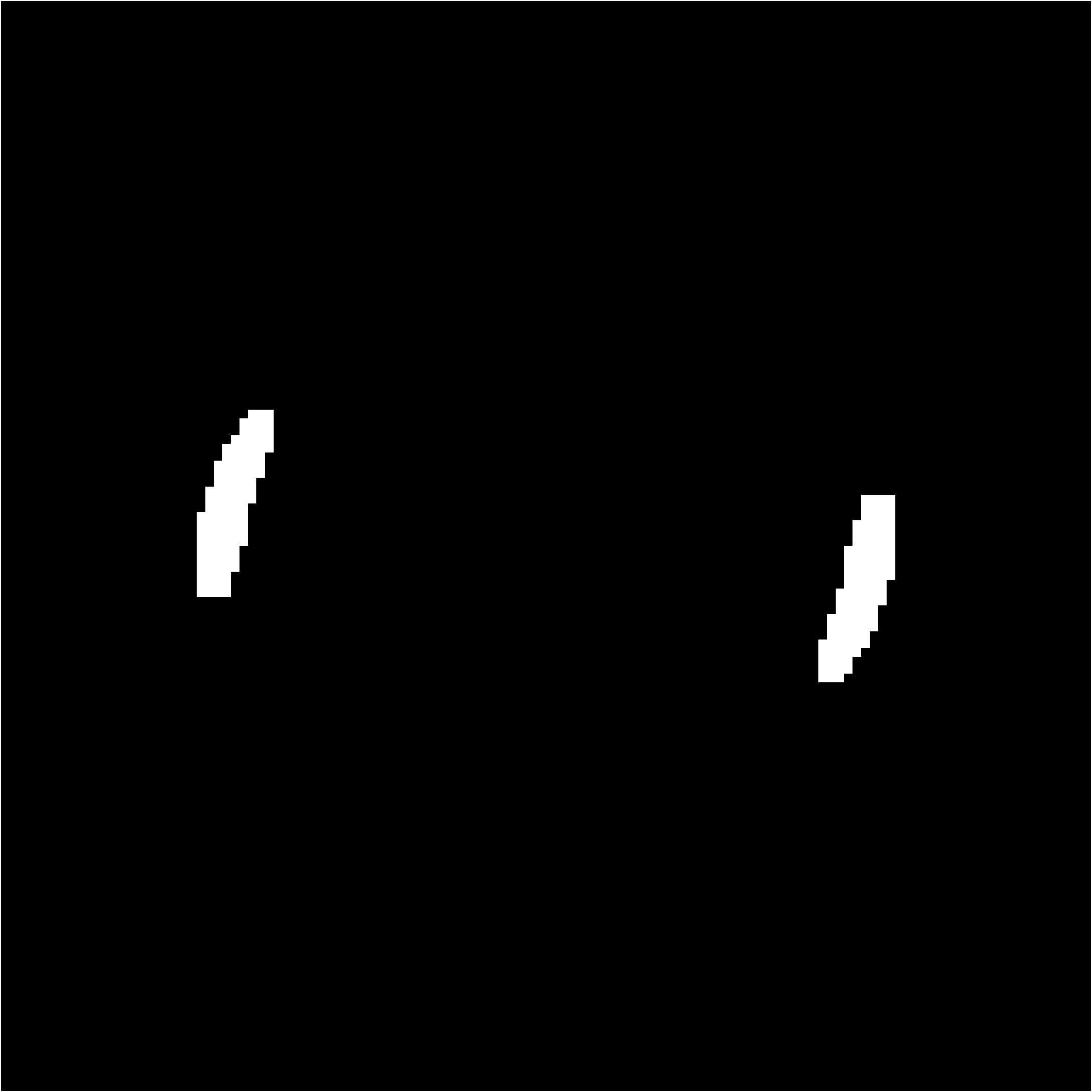}
         \caption{$SB^{0.1}_{H\oL}$}
         \label{fig:SBd_s}
     \end{subfigure}
     \begin{subfigure}[t]{0.15\textwidth}
         \centering
         \includegraphics[width=\textwidth]{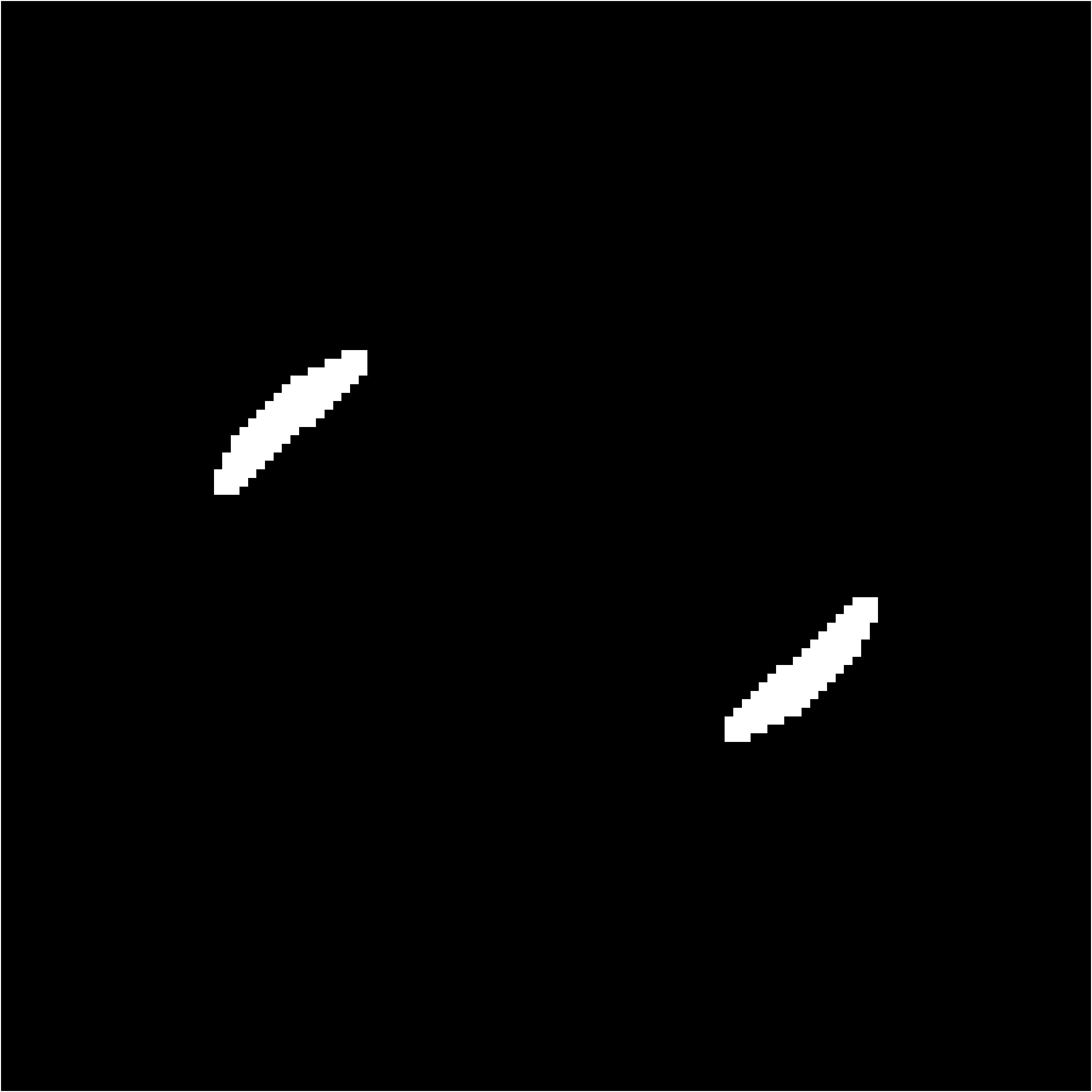}
         \caption{$SB^{0.1}_{H\oH}$}
         \label{fig:SBe_s}
         \end{subfigure}
     \begin{subfigure}[t]{0.15\textwidth}
         \centering
         \includegraphics[width=\textwidth]{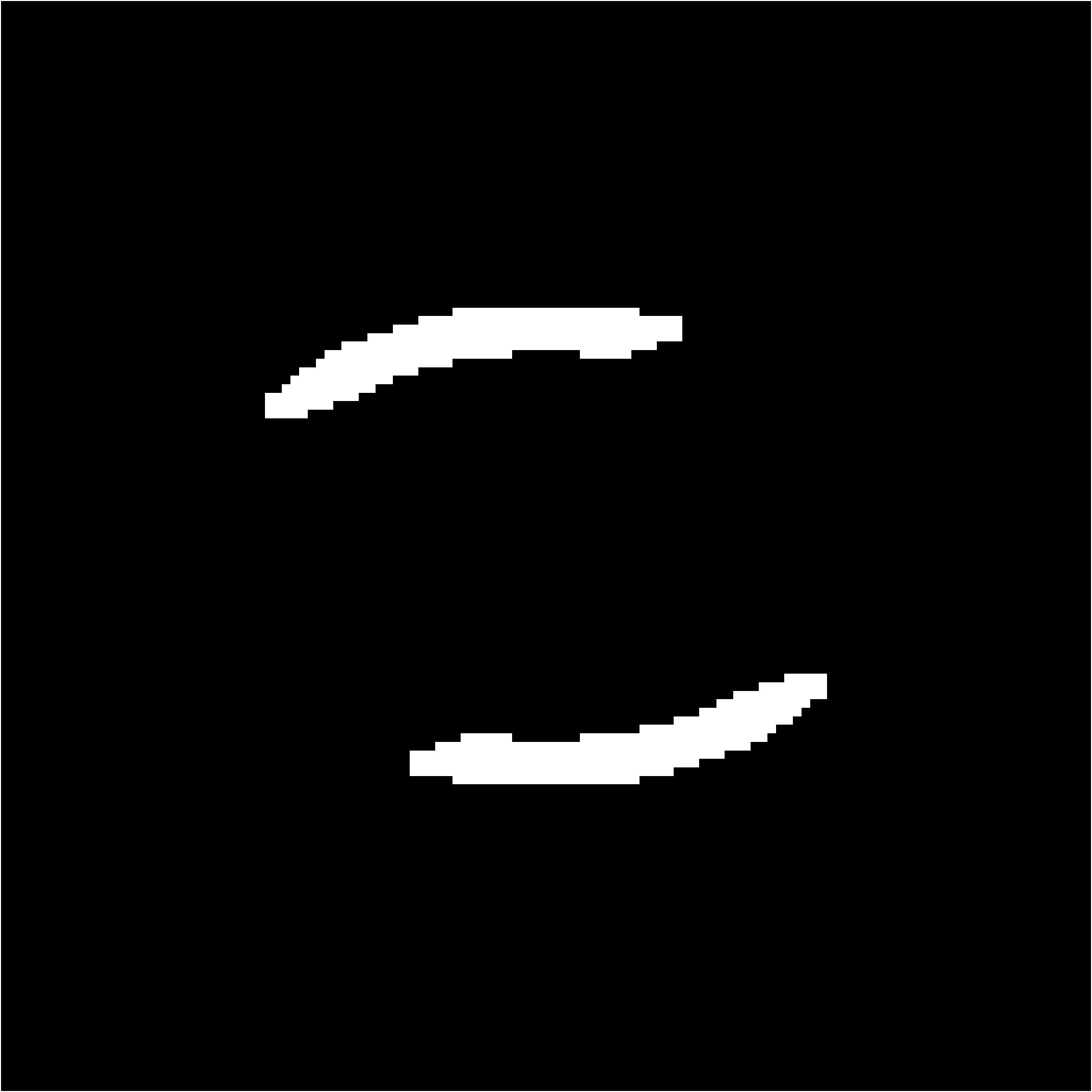}
         \caption{$SB^{0.1}_{L\oH}$}
         \label{fig:SBf_S}
     \end{subfigure}
     
        \caption{Approximation of an ellipse, presented in Figure \ref{fig:ellipse}, singularities and their directions using complex wavelets.}
        \label{fig:stable_singularities}
\end{figure}
We refer to these resulting binary matrices \begin{align*}
    SB^t_d = (C^t_d\circ L_{d})\oplus S, 
\end{align*}
as \emph{subband matrices}, and the singularities, i.e., elements with value $1$, as \emph{subbands}. See Figure $\ref{fig:stable_singularities}$ as an example.

The subband matrices give an approximation of stable singularities,  $SB^t_{LH}$ and $SB^t_{L\oH}$, divided by their directions. Note that in our study we suppose that subbands with directions $LH$, $HH$, $H\oH$, and $L\oH$ are unreliable since the singularities with these directions are invisible. Thus complex wavelet coefficients matrices $C_{d',W}[R]$, where $d'\in\{LH, HH, H\oH, L\oH\}$, are not of interest to us.

Next we will find the endpoints of subbands. Let $SB^t$ be a 3-dimensional matrix of size $m\times m \times 2$, whose elements are 
\begin{align*}
SB^t(x,y,z)=
     \begin{cases}
    SB^t_{HL}(x,y) \text{ if } z=1 \\
     SB^t_{H\oL}(x,y) \text{ if } z=2 \\
     \end{cases}
\end{align*}
Since subbands are multiple pixels wide, it is not clear which pixels are endpoints. Thus we take a morphological skeleton of the matrix $SB^{t}$. Now endpoints are those voxels that have only one voxel in their 26-neighborhood in $\skel(SB^t)$. To be more precise, let $E^t_{HL}$ be a matrix whose elements are
$E^t_{HL}(x,y)=1$ if $\skel(SB^t)(x,y,1)$ has exactly one neighboring voxel in its 26-neighborhood, and $E^t_{HL}(x,y)=0$ otherwise. Similarly $E^t_{H\oL}$ elements are 
$E^t_{H\oL}(x,y)=1$ if $\skel(SB^t)(x,y,2)$ has exactly one neighboring voxel in its 26-neighborhood and $E^t_{H\oL}(x,y)=0$ otherwise. Additionally, one can take a small neighborhood of an endpoint and take an intersection with subbands. One can say that all pixels in this intersection are endpoints. This gives more flexibility since the skeleton might not always reach the end of the subband.   

Futhermore let $E^{t,U}_{d}$ be matrices, where $U$ denotes up.
Elements of $E^{t,U}_{d}$ are
\begin{align*}
E^{t,U}_{d}(x,y) =
     \begin{cases}
     1, &\text{ if } \skel(SB^t)(x-1,y-1,z)=1, \\ & \quad \skel(SB^t)(x-1,y,z)=1, \\
     &\text{ or } \skel(SB^t)(x-1,y+1,z)=1\\
     0, &\text{ otherwise} 
     \end{cases},
\end{align*}
where $z=1$ if $d=LH$, and $z=2$ if $d=L\oH$. Now denote that $E^{t,D}_{d}:= E^{t}_{d}-E^{t,U}_{d}$. Here $D$ denotes down. Matrices $E^{t,U}_d$ contain those endpoints where corresponding subbands continue upwards from endpoints. Similarly, matrices $E^{t,D}_d$ contain those endpoints where subbands continue downwards. We benefit from these divisions of endpoints later when we form neighborhoods.

For computational purposes, we define four basic candywrap masks, which are discrete approximate intersections of metric balls (with respect to metric $D$ defined in Section \ref{sec:cw_dist}), where $\theta$ is a constant for simplicity. First, consider sets
\begin{align*}
\CW_s^{+R} &=\{(x,y) \mid \Tilde{D}((0,0,0),(x,y, \pi/6))<s, x>0\}, \\
\CW_s^{+L} &=\{(x,y) \mid \Tilde{D}((0,0,0),(x,y, \pi/6))<s, x<0\}, \\
\CW_s^{-R} &=\{(x,y) \mid \Tilde{D}((0,0,0),(x,y, -\pi/6))<s, x>0\}, \text{ and } \\
\CW_s^{-L} &=\{(x,y) \mid \Tilde{D}((0,0,0),(x,y, -\pi/6))<s, x<0\}.
\end{align*}
These fours sets are discretized to $N\times N$ binary matrices, called \emph{basic candywrap masks}, where $(x,y)\in[-10,10]\times[-10,10]$. We denote these binary matrices by $CW_s^{+R}, CW_s^{+L} , CW_s^{-R}$ and $ CW_s^{-L}$ respectively, see Figure \ref{fig:basic_candywraps}. 

Here $CW_s^{+R}$ is a neighborhood of a curve, in which the first endpoint is located at $(0,0)$ and has tangential direction $0$ radians, and the second endpoint is located at $(x,y)$, $x>0$, i.e, the curve is on the right-hand-side of the first endpoint, with tangential direction $\pi/6$ radians. The mask $CW_s^{+L}$ gives a similar neighborhood, but the second endpoint is located on the left-hand side of the first endpoint. The elements $CW_s^{-R}$ and $CW_s^{-L}$ are like previous neighborhoods with the difference that the second endpoint has tangential direction $-\pi/6$. 

\begin{figure}[!ht]
     \centering
     \begin{subfigure}[b]{0.2\textwidth}
         \centering
         \includegraphics[width=\textwidth]{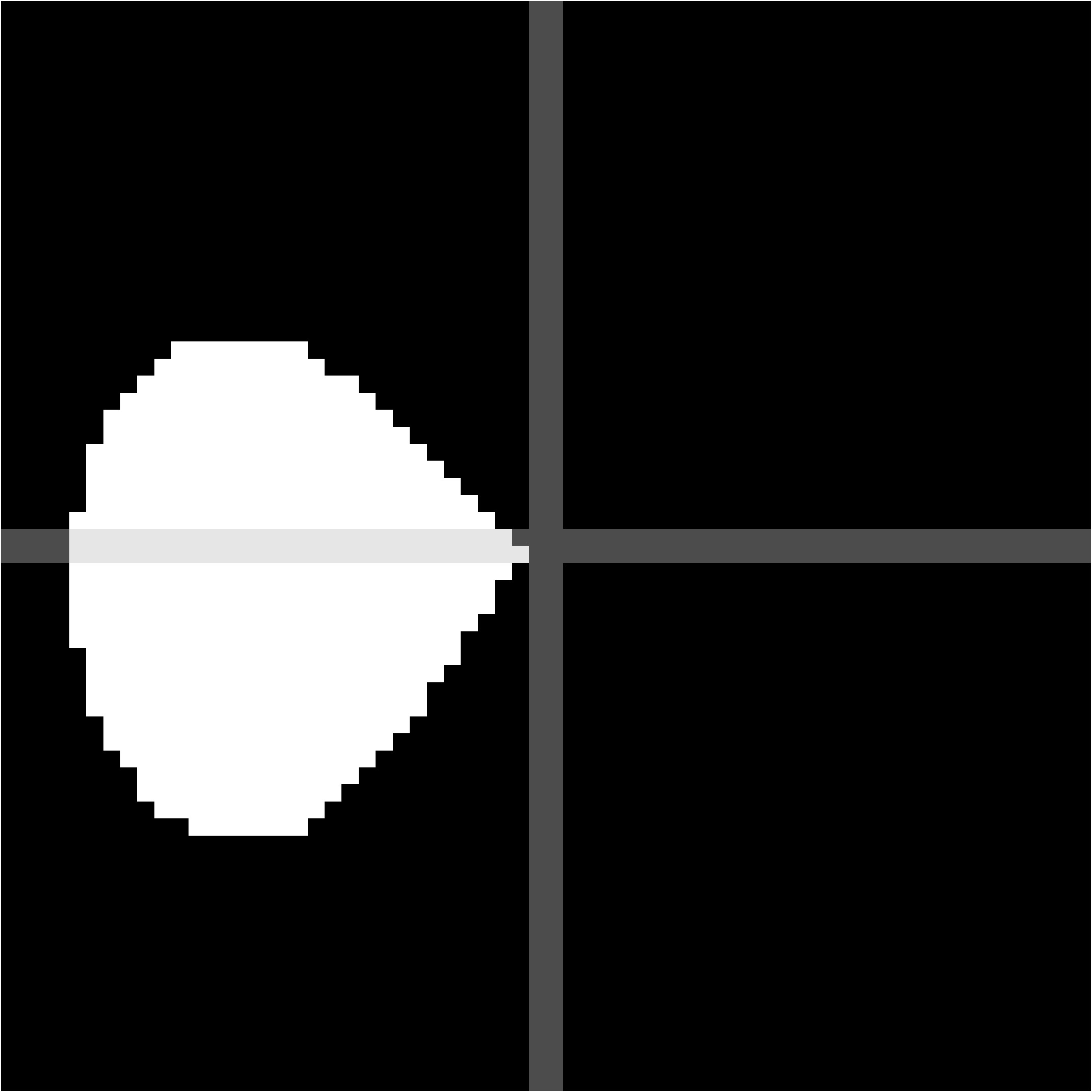}
         \caption{$CW_9^{+L}$}
         \label{fig:upleft-candywrap}
     \end{subfigure}
     \begin{subfigure}[b]{0.2\textwidth}
         \centering
         \includegraphics[width=\textwidth]{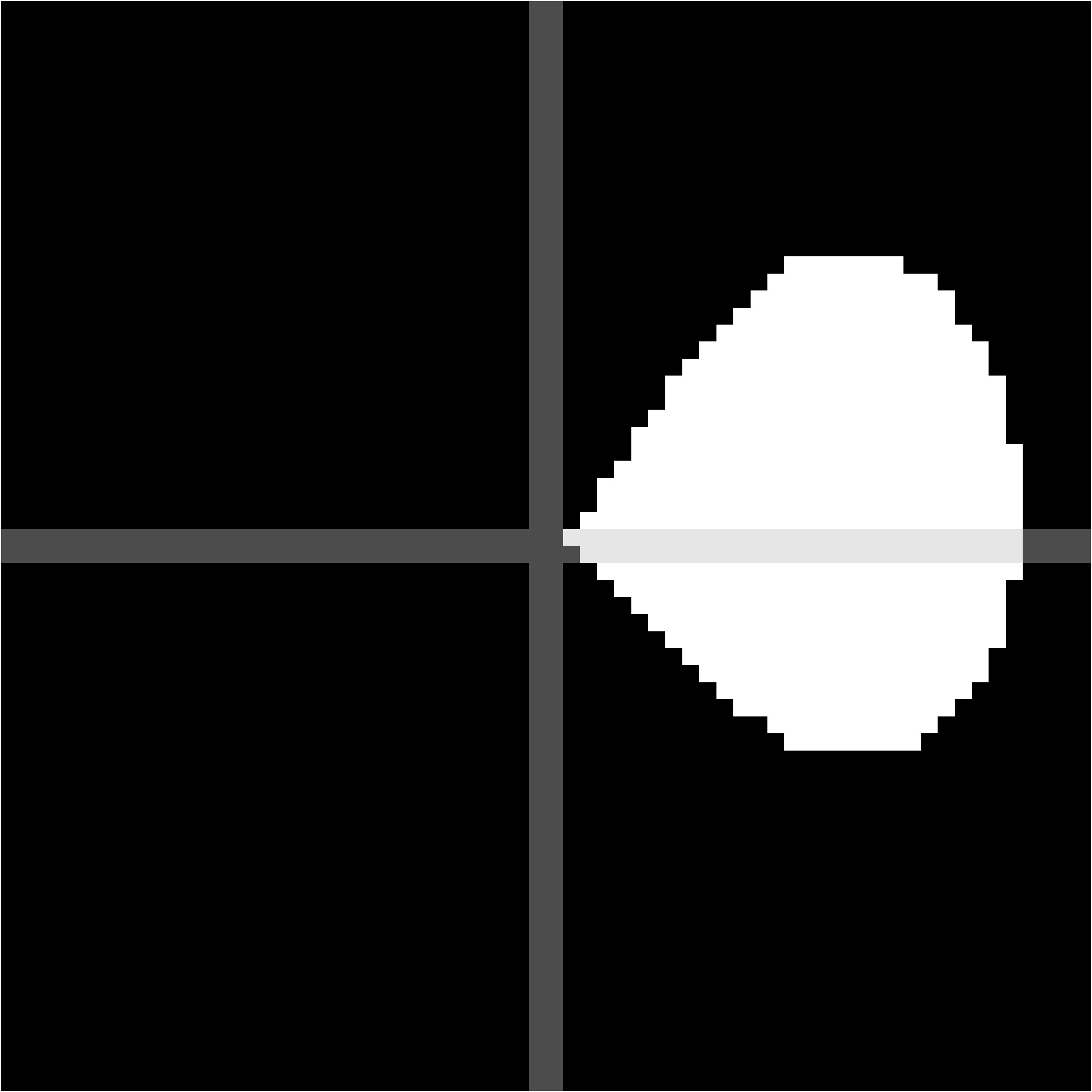}
         \caption{$CW_9^{+R}$}
         \label{fig:upright-candywrap}
     \end{subfigure}
     \begin{subfigure}[b]{0.2\textwidth}
         \centering
         \includegraphics[width=\textwidth]{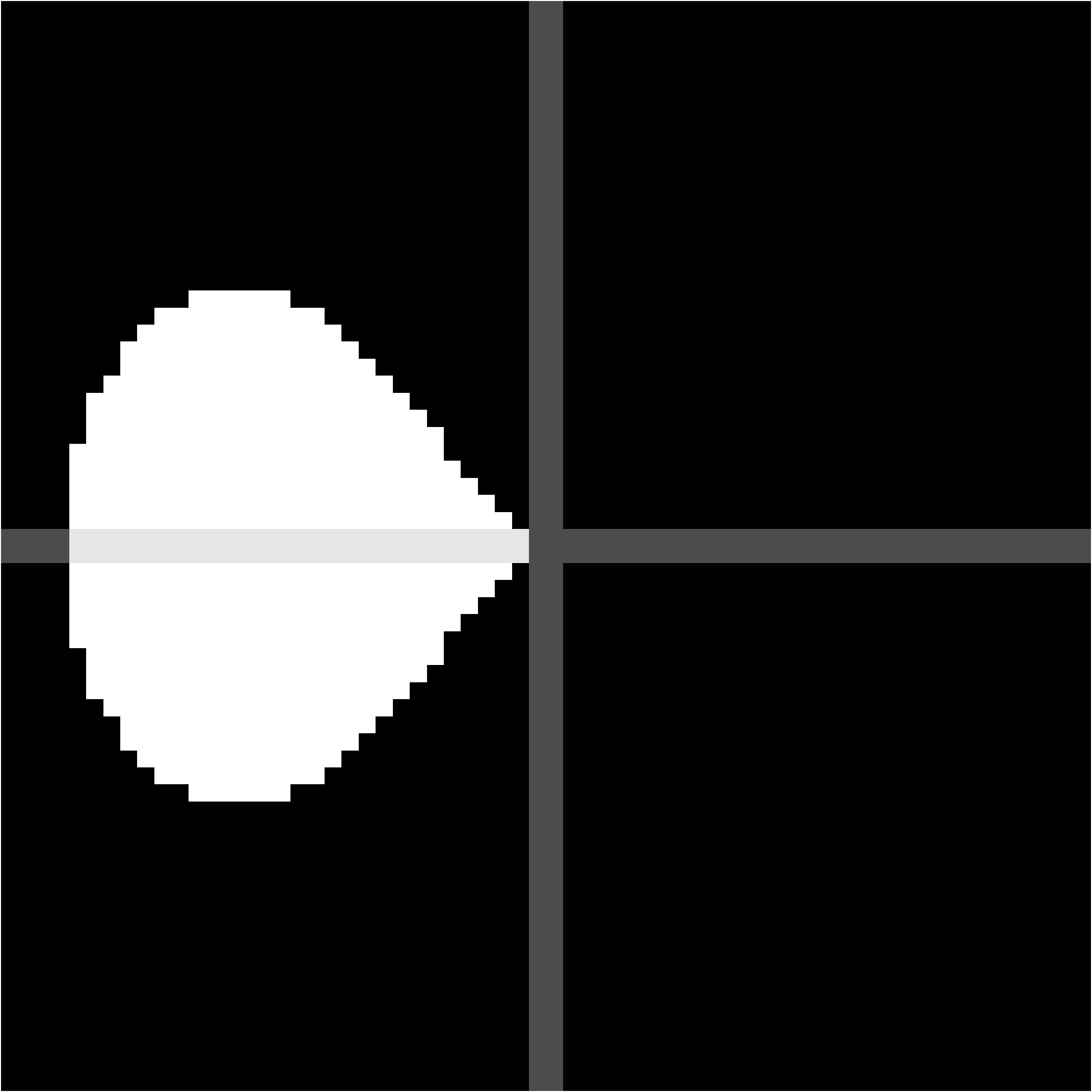}
         \caption{$CW_9^{-L}$}
         \label{fig:downleft-candywrap}
     \end{subfigure}
     \begin{subfigure}[b]{0.2\textwidth}
         \centering
         \includegraphics[width=\textwidth]{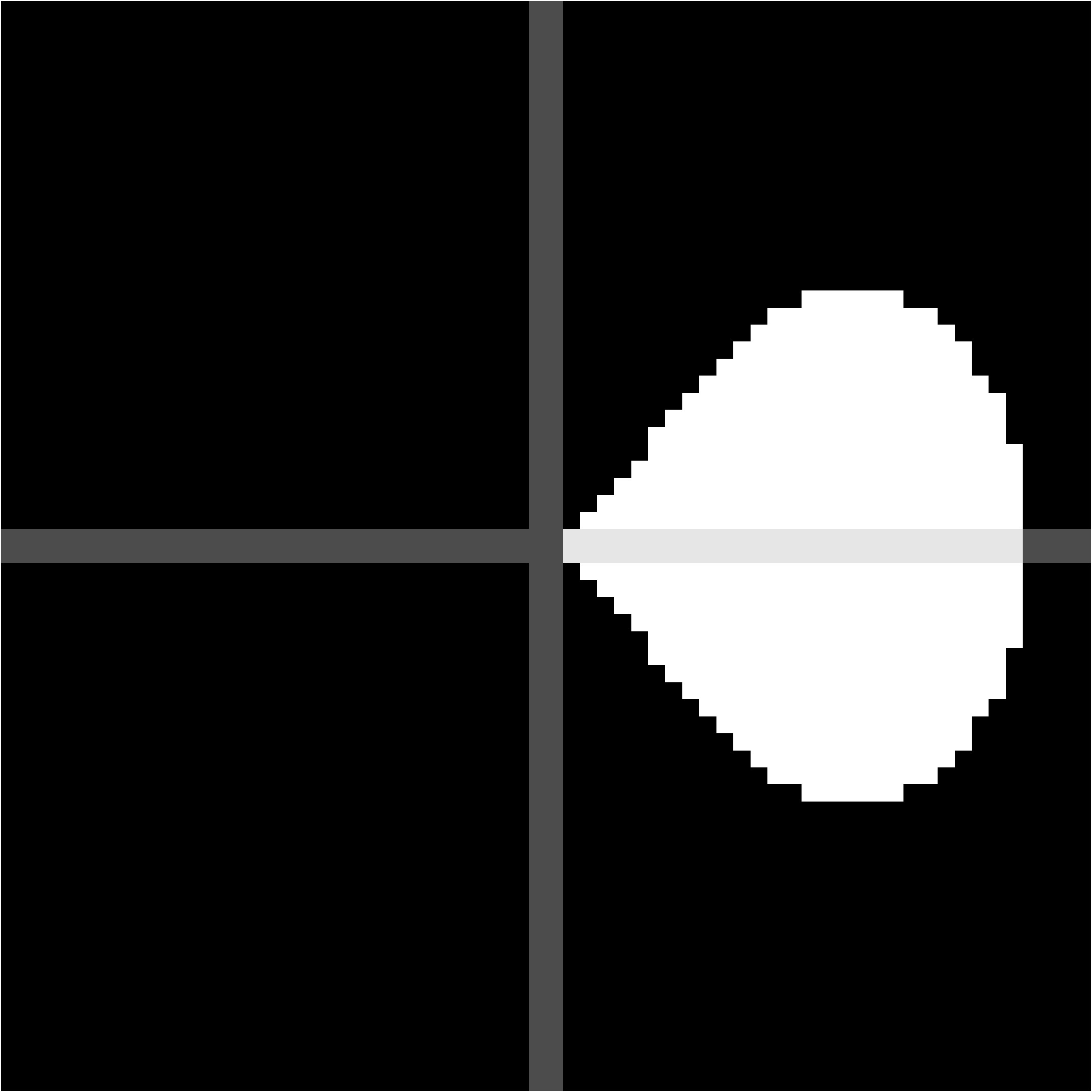}
         \caption{$CW_9^{-R}$}
         \label{fig:downright-candywrap}
     \end{subfigure}
     \caption{Basic candywrap masks are approximate intersections of metric balls (with respect to metric $D$). Here $N=64$. The gray axes are for illustration. }
        \label{fig:basic_candywraps}
\end{figure}

Candywrap masks, which we use for estimating neighborhoods of boundary curves, are created by rotating basic candywraps. Consider that we have a point $(0,0)$ with tangential direction $\theta_0$ degrees. This point is a supposed or known first endpoint of the boundary curve. Suppose then that the second endpoint has tangential direction $\theta_1 = \theta_0 + \frac{\pi}{6}$, $\theta_0 \in (-\pi, \pi]$. The aim is to estimate a neighborhood where a boundary curve lies with the assumption that the distance between two endpoints is at the most $s$. 

That can be done by rotating the first endpoint by $-\theta_0$ radians. Then the first (rotated) endpoint has tangential direction $0$ radians and the second endpoint has tangential direction $\frac{\pi}{6}$. Now $CW_s^{+R}$ defines the neighborhood, where a rotated boundary curve lies if we suppose the second endpoint to lie on the right-hand side of the first endpoint. Similarly, $CW_s^{+L}$  defines the neighborhood where the rotated boundary curve lies if the second endpoint is on the left-hand side of the first endpoint. The last thing is to rotate the neighborhood back by $+\theta_0$ radians.

If we suppose that the second endpoint has tangential direction $\theta_1 = \theta_0 - \frac{\pi}{6}$ degrees we can do a similar neighborhood estimation using $CW_s^{-R}$ or $CW_s^{-L}$ depending supposed placement of the boundary curve. We denote that $CW_{s,\theta}^{*\bullet}$ is $CW_s^{*\bullet}$ rotated by $\theta$ degrees, where $*$ is $+$ or $-$ and $\bullet$ is $R$ or $L$.

Now neighborhoods of missing singularities, i.e., missing parts of boundary curves or subbands, can be achieved by doing dilation with properly chosen endpoint matrices and rotated candywrap masks, which we will go through. Denote that 
\begin{align*}
    DM_{s,HH}^\bullet &= E_{HL}^{t,\bullet} \oplus CW^{+\bullet}_{s,-\frac{\pi}{3}} \text{ and}\\
    DM_{s,H\oH}^\bullet &= E_{H\oL}^{t,\bullet} \oplus CW^{-\bullet}_{s,\frac{\pi}{6}},
\end{align*}
where $\bullet$ is $U$ or $D$. 
Secondly, let 
\begin{align*}
    DM_{s,LH}^\bullet &= DM_{s,HH}^{\bullet} \oplus CW^{+\bullet}_{s,-\frac{\pi}{6}} \text{ and}\\
    DM_{s,L\oH}^\bullet &= DM_{s,H\oH}^{\bullet} \oplus CW^{-\bullet}_{s,\frac{\pi}{3}}.
\end{align*} Furthermore we define that $DM_{d',s}$ is a matrix whose elements are
\begin{align*}
    DM_{s,d'}(x,y)=
    \begin{cases}
        1, \text{ if }  DM_{s,d'}^U(x,y)=1 \\
        1, \text{ if }  DM_{s,d'}^D(x,y)=1\\
        0, \text{ otherwise }
    \end{cases}.
\end{align*}
Now $DM_{s,d'}$ gives size $s$ neighborhood estimation of unknown subband $d'$. Note that here dilation acts as a shifting operator. Moreover, note that for estimating neighborhoods of subbands $L\oH$ and $LH$ we perform dilations using neighborhood estimations of subbands $H\oH$ and $HH$. This is because we do not know the exact placement of endpoints in those subbands.      

Let $A_s$ be a 3D matrix of size $m\times m \times 13$, whose elements are
\begin{align*}
    A_s(x,y,z)=
    \begin{cases}
        SB^t_{HL}(x,y), \text{ if }  z=1,7,13 \\
        SB^t_{H\oL}(x,y), \text{ if }  z=2,8 \\
        DM_{s,H\oH}(x,y), \text{ if }  z=3,9 \\
        DM_{s,L\oH}(x,y), \text{ if }  z=4,10 \\
        DM_{s,LH}(x,y), \text{ if }  z=5,11 \\
        DM_{s,HH}(x,y), \text{ if }  z=6,12 
    \end{cases}.
\end{align*}

Finally, suppose that $X$ and $Y$ are $\{0,1\}^{n \times n}$ matrices, $n\in \N$, now
$X\setminus Y$ is a $\{0,1\}^{n \times n}$ matrix whose elements are 
\begin{align*}
    X\setminus Y(x,y) = \begin{cases}
        1 \text{ if } (X-Y)(x,y) =1 \\
        0, \text{ otherwise}.
    \end{cases}
\end{align*}

\subsection{Method description}

In this section, we will show how the estimation of the unknown boundaries can be done in practice. The complete process is also summarized in Figure \ref{fig:workflow}.

The final estimation of these neighborhoods is achieved iteratively through forming filtrations of multiple different neighborhood estimations and finding the estimation which contains components that fulfill Proposition \ref{prop:lifting} in a computational sense. Thus we will formulate a statement whose correctness we will check during the iterations.

\begin{stm}\label{stm:component}
There exists a component matrix $C$ of $A_s$ for which the following is true:
there exists $x$ and $y$ so that $C(x,y,1)=1$ and $C(x,y,13)=1$.
\end{stm}

The proposed method starts by doing an initial reconstruction $R$. We choose a wavelet level $W$, a threshold value $t$, and a line length $l$. We compute subband matrices $SB^t_d$ from the reconstruction $R$. Furthermore we compute matrices $E^{t,\bullet}_{d}$. Next, we choose a candywrap size $N$, an initial distance $s>0$, and a stepsize $z>0$. With these initial settings, we create $A_s$ and check if Statement \ref{stm:component} holds. If such a component matrix $C$ exists, it means that we have found an estimation of a neighborhood of some boundary curve in three-dimensional space. Now the projection of $C$ to its first two coordinates gives an estimation of a boundary neighborhood in two-dimensional space. 

If a component consistent with Statement 1 is found, we continue neighborhood estimation as follows. We update that \begin{align*}
SB^t_{HL}=SB^t_{HL}\setminus C_{\mid z=1},\hspace{5mm} &SB^t_{H\oL}=SB^t_{H\oL}\setminus C_{\mid z=2}, \\
E^{t,\bullet}_{HL}=E^{t,\bullet}_{HL}\setminus C_{\mid z=1},\hspace{5mm} &E^{t,\bullet}_{H\oL}=E^{t,\bullet}_{H\oL}\setminus C_{\mid z=2}, 
\end{align*}
where  $C_{\mid z=i}$ is a 2D matrix whose elements are $C_{\mid z=i}(x,y) = C(x,y,i)$.
In other words, we remove parts of the found boundary from known subbands. Next, we compute $DM_{s,d'}$ and furthermore $A_s$ using these updated matrices. We compute components of $A_s$ again and check if the component matrix which fulfills Statement \ref{stm:component} is found.

If component matrix $C$ is not found, we update that $s=s+z$ and compute $DM_{s,d'}$ and $A_s$ using the updated size, i.e. we increase the estimation distance. 

Next, we check if $C$ is found. We process the method as explained above depending on the result. If $A_s$ does not contain any components, we have found all neighborhood estimations for every boundary curve. See the pseudocode \ref{pseudocode} for the method.

In other words, we consider a filtration $Q_{s_0}\subset Q_{s_1} \subset \cdots \subset Q_{s_K}$, where $Q_{s_i}\rightleftharpoons A_{s_i}$, $s_0$ is the initial $s$, and $s_{i}=s+iz$, when $i=1,\dots K$. Furthermore, we consider homology groups $H_0(Q_{s_i})$. We are looking for a birth time of a component that connects the bottom and top layers. The iterative approach, removing parts of the found neighborhood from known subbands and forming new filtration, is done because we want to minimize the risk that two boundary curves' neighborhoods merge together, giving us misleading results.     

If estimations of boundary curves are desired, a user may use, for example, splines to do so. See for example Figures \ref{fig:annulus_spline}, \ref{fig:shapes_spline}, and  \ref{fig:blob_spline}.

\begin{algorithm}
\caption{Estimations of boundaries' neighborhoods}
\label{pseudocode}
\begin{algorithmic}
\Require{$s>0$, $z>0$}
\Repeat
\State{Compute component matrices of $A_s$}
\If{$\exists C$ a component matrix of $A_s$ and  $x,y$, such that, $C(x,y,1)=1$ and $C(x,y,1)=13$}
    \State{save $C$}
    \State{$SB^t_{HL} \gets SB^t_{HL}\setminus C_{\mid z=1}$}
    \State{$SB^t_{H\oL} \gets SB^t_{H\oL}\setminus C_{\mid z=2}$}
    \State{$E^{t,\bullet}_{HL} \gets E^{t,\bullet}_{HL}\setminus C_{\mid z=1}$}
    \State{$E^{t,\bullet}_{H\oL} \gets E^{t,\bullet}_{H\oL}\setminus C_{\mid z=2}$}
    \State{Compute $DM_{s,d'}$}
    \State{Compute $A_{s}$}
\Else{}
\State{$s \gets s+z$}
\EndIf{}
\Until{$A_s$ has components}
\end{algorithmic}
\end{algorithm}

\usetikzlibrary{positioning, fit,calc}
\tikzstyle{block} = [draw=black, thick, text width=3cm, minimum height=1cm, align=center]  
\tikzstyle{arrow} = [thick,->,>=stealth, text width=3cm]
\tikzstyle{line} = [thick,-,, text width=3cm]

\begin{figure}[!ht]
    \centering
    \begin{tikzpicture} 
    \node (sino) {\includegraphics[width=2.5cm]{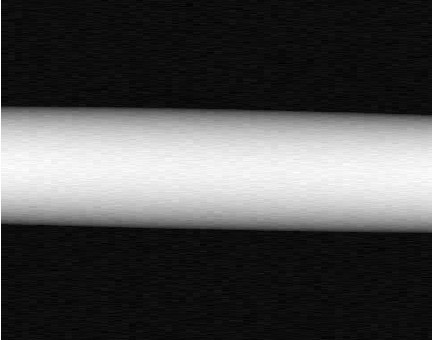}};  
    \node[right =4.5cm of sino] (rec) {\includegraphics[width=2.5cm]{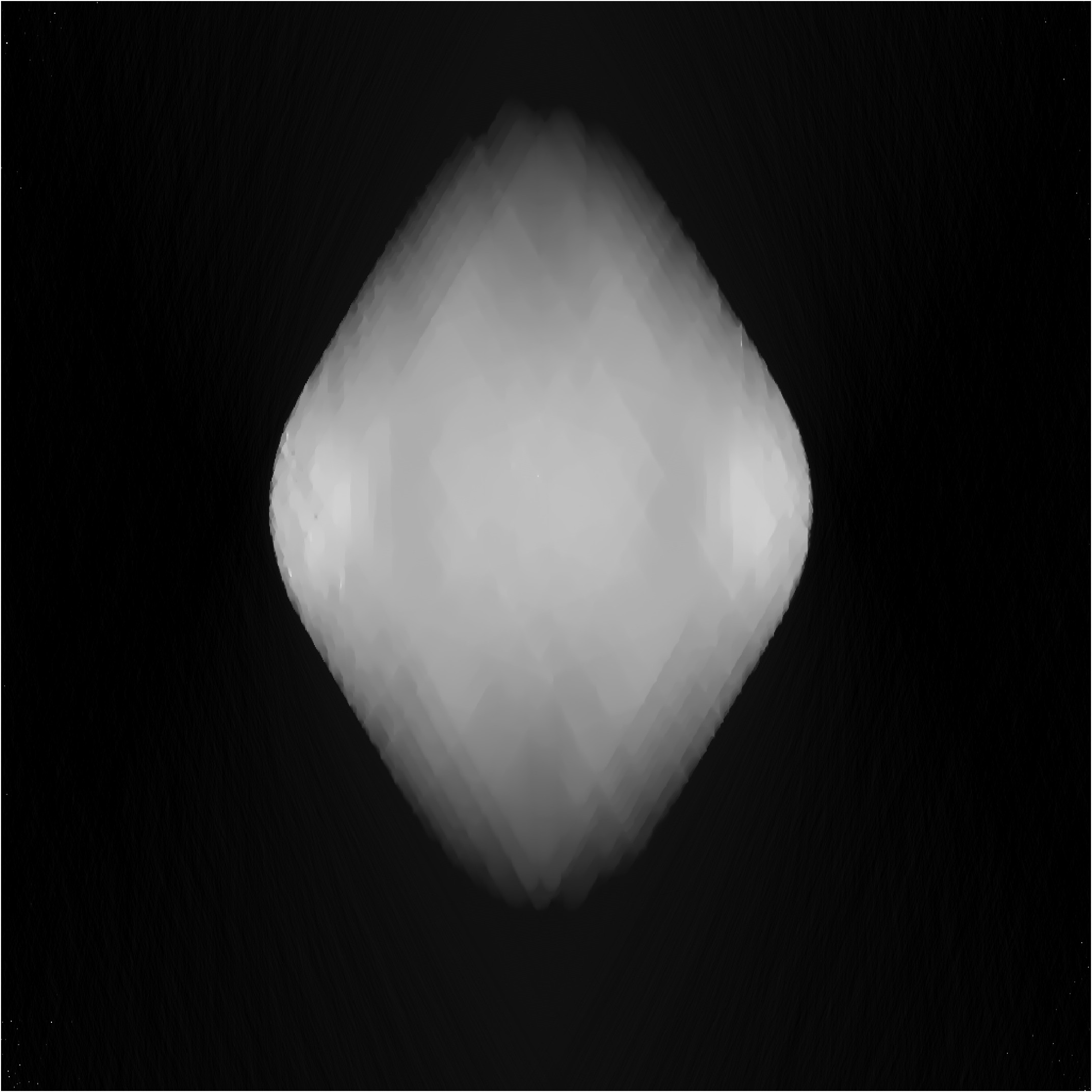}};
    \node[below =1.8cm of rec] (complex) {\includegraphics[width=2.5cm]{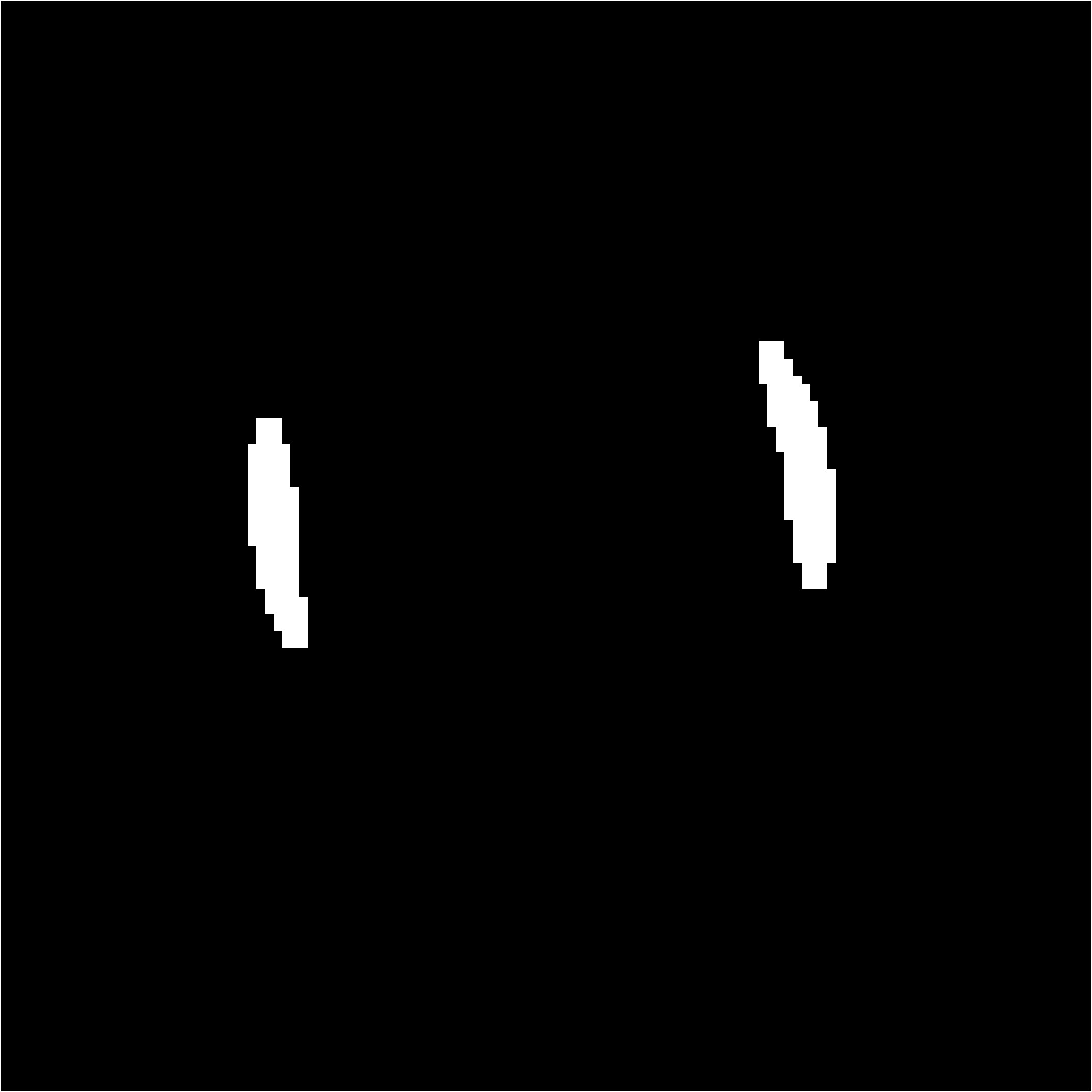} \includegraphics[width=2.5cm]{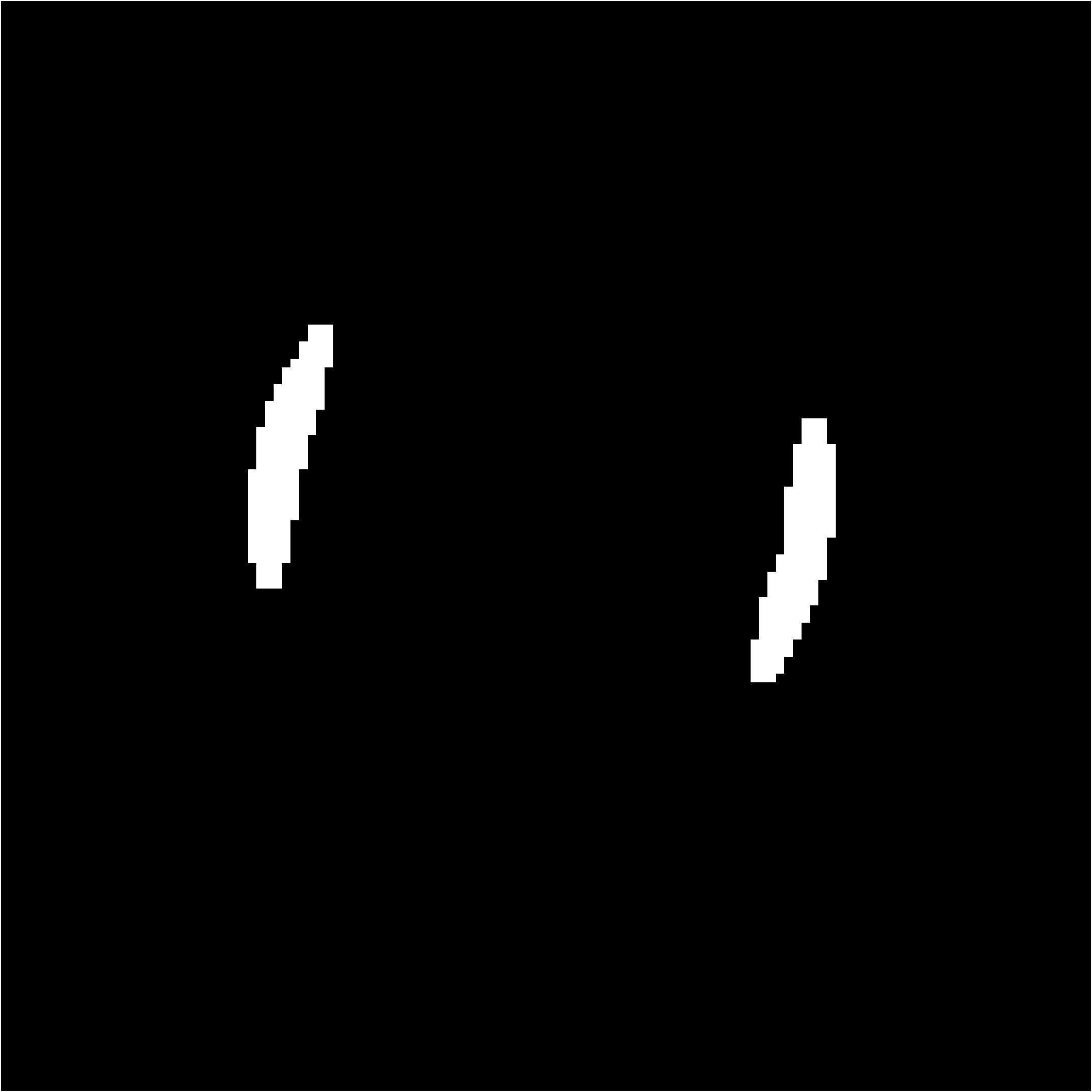}};  
    \node[left =4.5cm of complex] (endpoints) {\includegraphics[width=2.5cm]{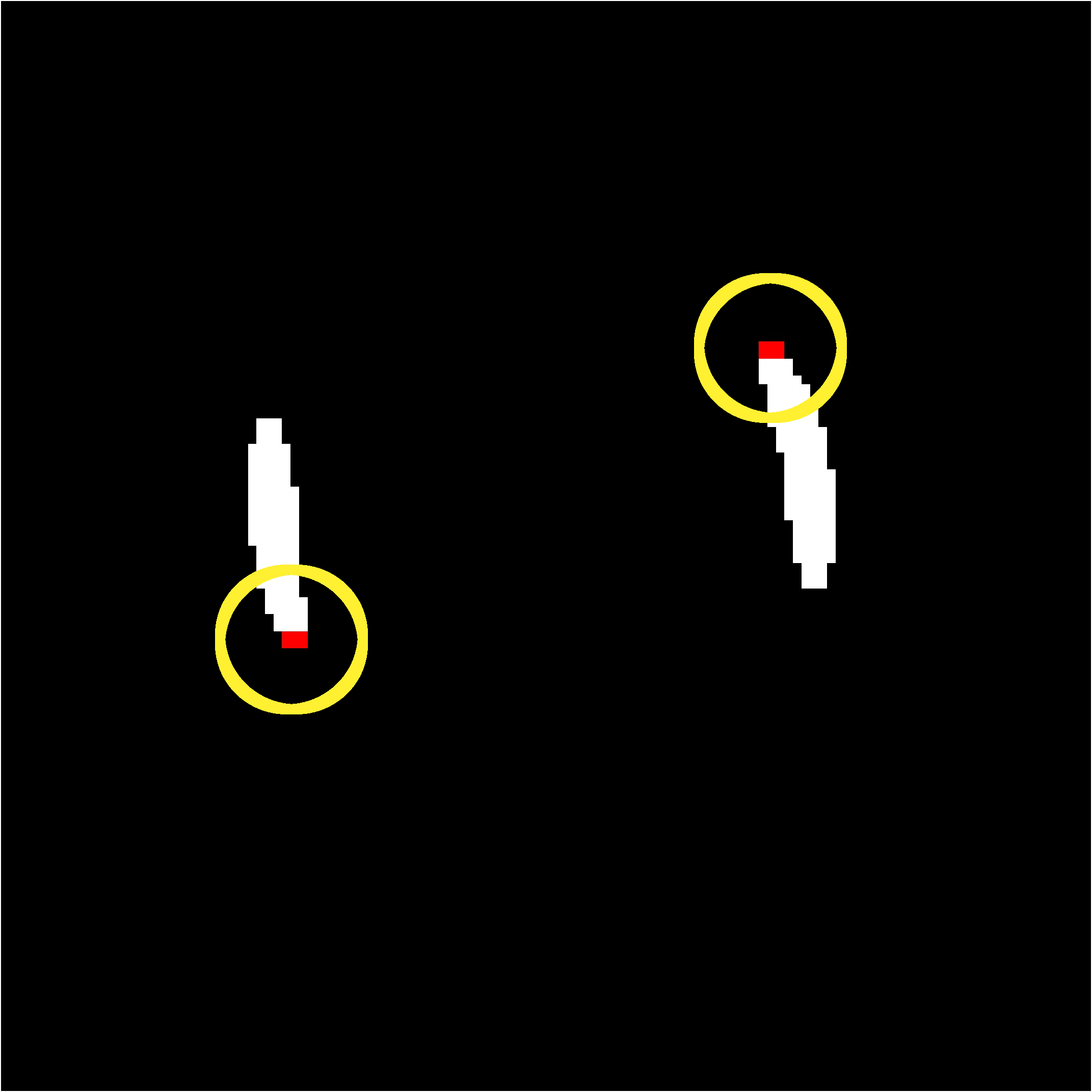} \includegraphics[width=2.5cm]{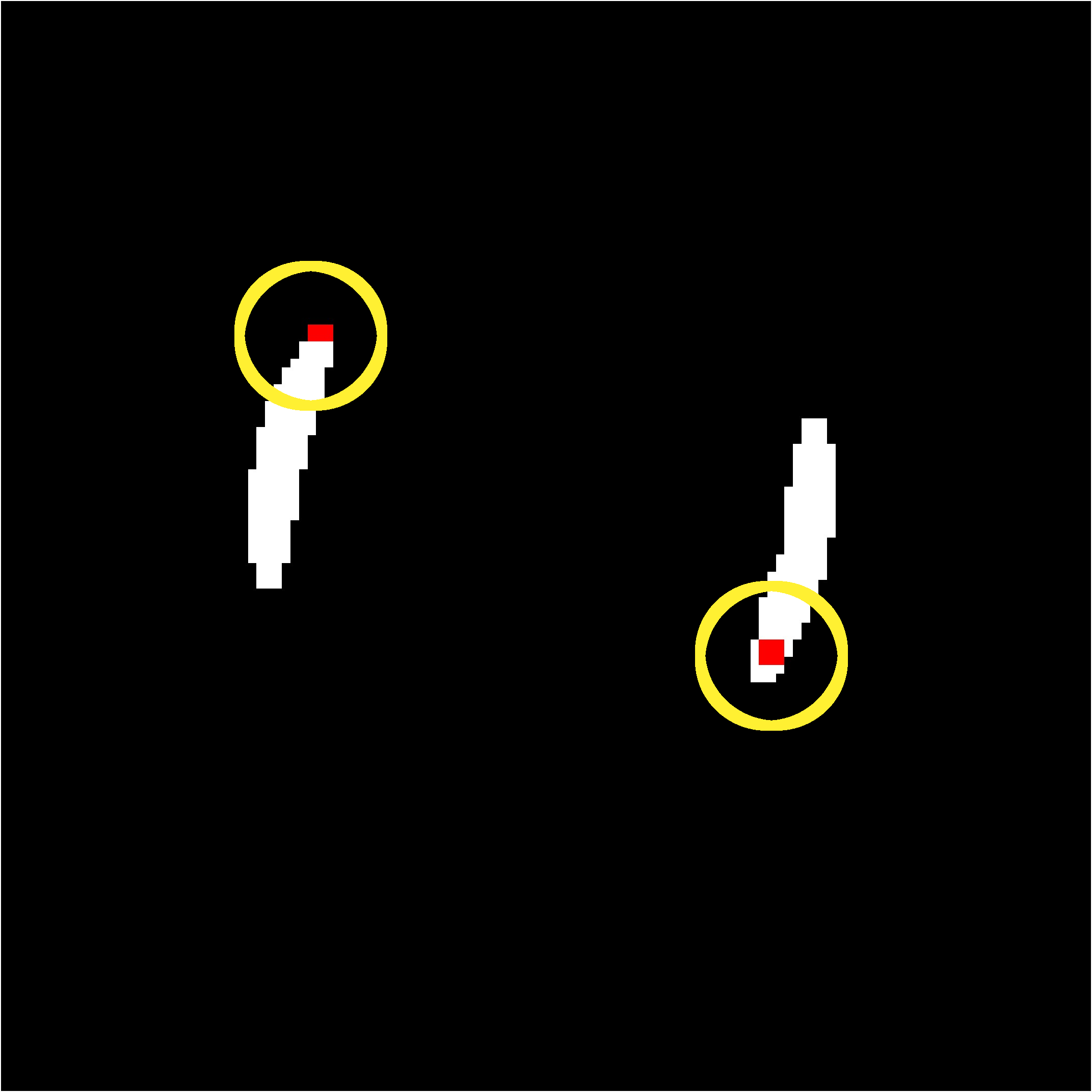}};
    \node[rectangle, draw, below right =1cm of endpoints]  (dil) {\includegraphics[width=2.5cm]{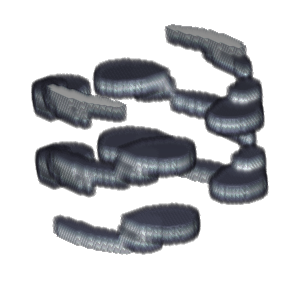}};
    \node[circle, draw, below =2cm of dil](p)  {};
    \node[rectangle, draw, right =1.5cm of p](compIf){Statement \ref{stm:component}};
    \node[rectangle, draw, left =1.5cm of p](nocomp){There is no components in $A_s$.};
    \node[diamond,minimum height = 1.2cm,minimum width = 1.2cm, draw](no) at (5,-10) {no};
    \node[diamond,minimum height = 1.2cm,minimum width = 1.2cm, draw,right = 1cm of no](yes) {yes};
    \node(spline) at (-1,-15){\includegraphics[width=2.5cm]
    {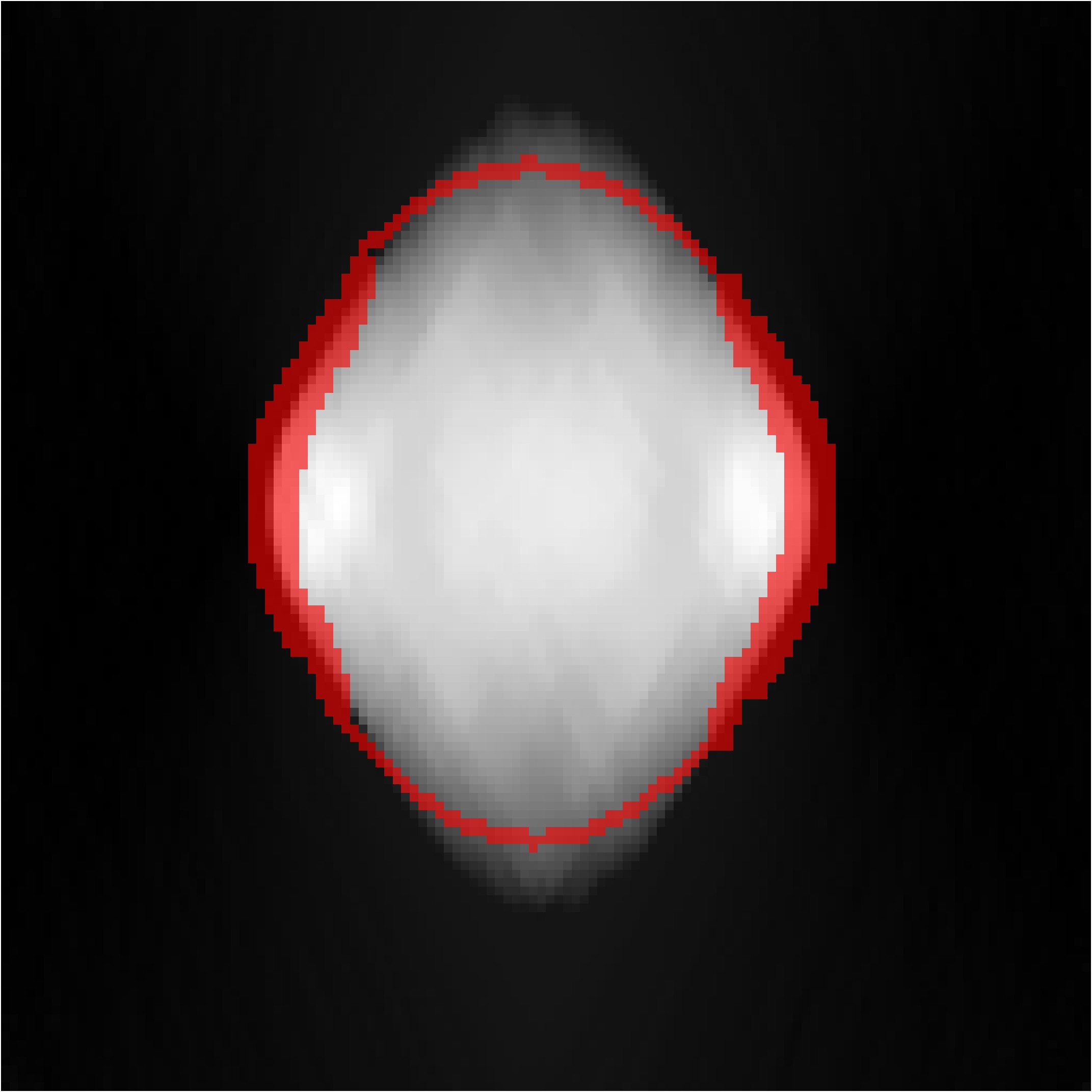}};
    \node(est) at (3,-15){\includegraphics[width=2.5cm]
    {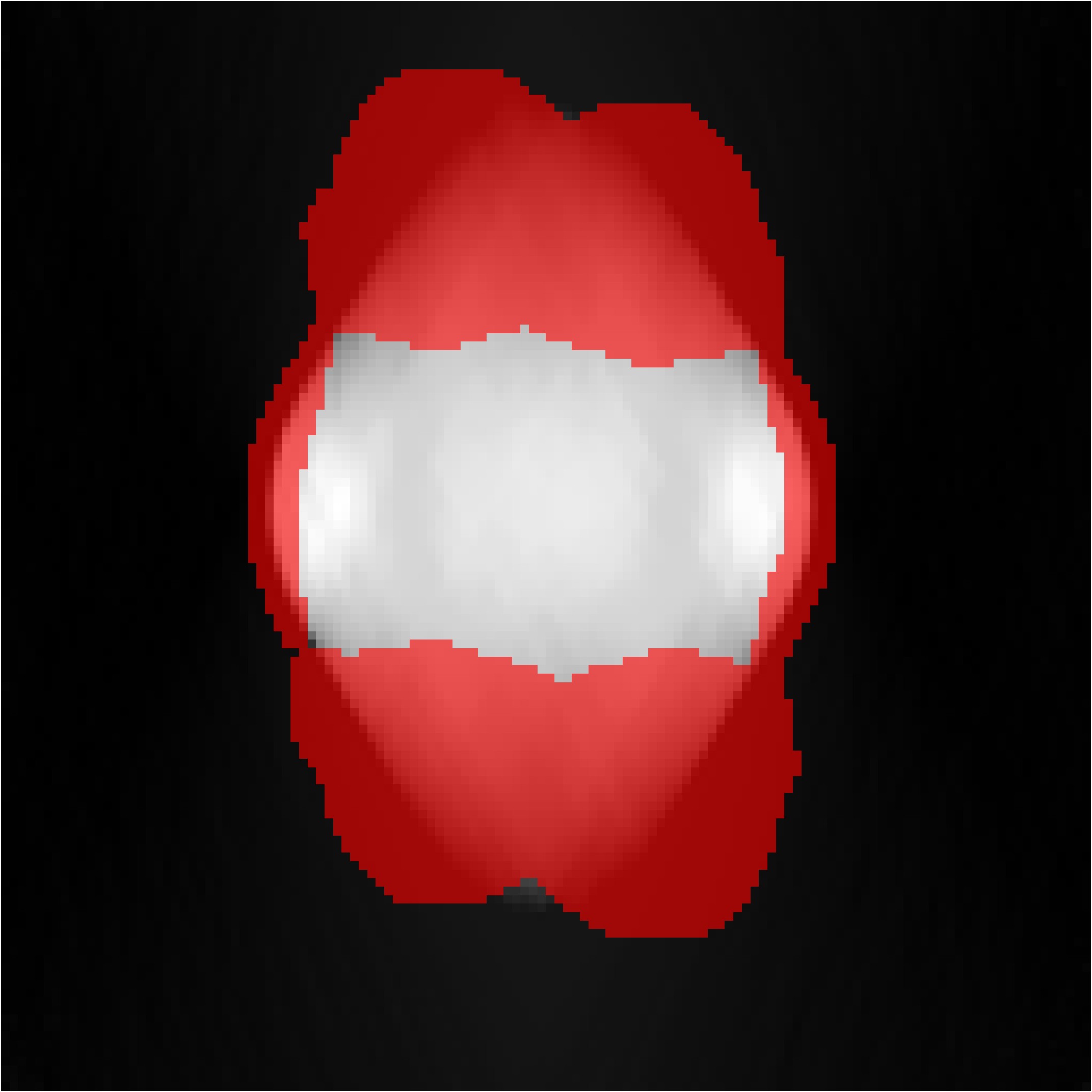}};
\draw[arrow] (sino)--node [text width=3.5cm,midway,above]{1. Reconstruction $R$} (rec);
\draw[arrow] (rec)--node [text width=4cm,midway,left]{2. Subband matrices of $R$}(complex);
\draw[arrow] (complex)--node [text width=3.5cm,midway,above]{3. Find endpoints \\ (red points inside yellow circles)} (endpoints);
\draw[arrow] (endpoints)--node [text width=4.7cm,midway,below left]{4. Morphological dilations using candywrap masks as structuring elements} (dil);
\draw[arrow](dil)--node [text width=3.5cm,midway,left]{5. Compute components of $A_s$}(p);
\draw[arrow](p)--node [text width=1cm,midway,above]{6.}(compIf);
\draw[arrow](p)--node [text width=1cm,midway,above]{7.}(nocomp);

\draw[arrow] (compIf)--node [text width=4cm,midway,right]{} (no);
\draw[arrow] (compIf)--node [text width=4cm,midway,right]{} (yes);
\draw[arrow] (no)--node [text width=4cm,midway, above right]{6.1. Grow $s$} (dil);
\draw[arrow] (yes)--node [text width=4cm,midway,right]{6.2. Remove found component} (complex);
\draw[arrow] (nocomp)--node [text width=2.5cm,midway,left]{7.1. Spline fills of missing parts} (spline);
\draw[arrow] (nocomp)--node [text width=4cm,midway,above right]{7.2. Projection of found components} (est);

\end{tikzpicture}  
    \caption{The workflow of TILT is as follows. 1. Compute an initial reconstruction $R$ from limited-angle tomographic data. 2. Compute subband matrices of $R$ using complex wavelets. 3. Find endpoints of subbands. 4. Perform morphological dilations using a rotated candywrap mask formed using current distance $s$. 5. Compute components of the 3D matrix $A_s$. 6. Check if Statement \ref{stm:component} holds. 6.1. If Statement holds, remove the found component and return to step 3. 6.2 If Statement does not hold, grow $s$ with chosen step size and return to step 4. This step is where persistent homology appears. 7. If $A_s$ does not contain any components, all boundary neighborhoods are supposed to be found. 7.1 Do spline fill to complete missing parts. 7.2. Project found neighborhoods.}
    \label{fig:workflow}
\end{figure}

\section{Results} \label{sec:results}

To test the TILT method, we created four binary phantoms. They are progressively more complicated, pushing the limits of our assumptions to explore the scope of applicability of TILT.
\begin{itemize}
    \item[(i)] Annulus phantom: not simply connected but only one component, Figure \ref{fig:target_annulus}.
    \item[(ii)] Several components: group of ellipses, Figure \ref{fig:target_multipleshapes}.
    \item[(iii)] One component, but non-convex, Figure \ref{fig:target_blob}.
    \item[(iv)] One component, non-convex, high curvature, \ref{fig:target_makkara}.
\end{itemize}
We simulated X-ray measurements of phantoms using a parallel-beam geometry of a 60-degree opening angle with 61 projection images, spanning from -30 to 30 degrees. We added normally distributed pseudo-random noise to sinograms with  $3\%$ relative noise.

To form the initial reconstructions, we adapted the deconvolution code in \cite{bredies2014recovering} for computing TV-regularized reconstructions of the phantoms. The complex wavelet coefficients were computed using Kingsbury Q-shift 2-D Dual-tree complex wavelet transform. We took $3\times3$ neighborhood of skeleton's endpoints and intersect those with subbands, and considered all these points in intersections to be endpoints. For candywrap masks, we used $N=64$. Following the method explained in previous sections, we estimated neighborhoods of boundaries. Furthermore, we did spline fills for missing parts. Figure \ref{fig:workflow} shows the entire workflow of TILT. As an example, the computations for phantom \ref{fig:target_annulus} after an initial reconstruction took slightly less than 40 seconds using an 11th Gen Intel(R) Core(TM) i5-1135G7 processor. For each extra boundary component, computations took an additional three to eight seconds.   


\begin{figure}[!ht]
    \centering
    \begin{subfigure}[t]{0.3  \textwidth}
    \includegraphics[width=\textwidth]{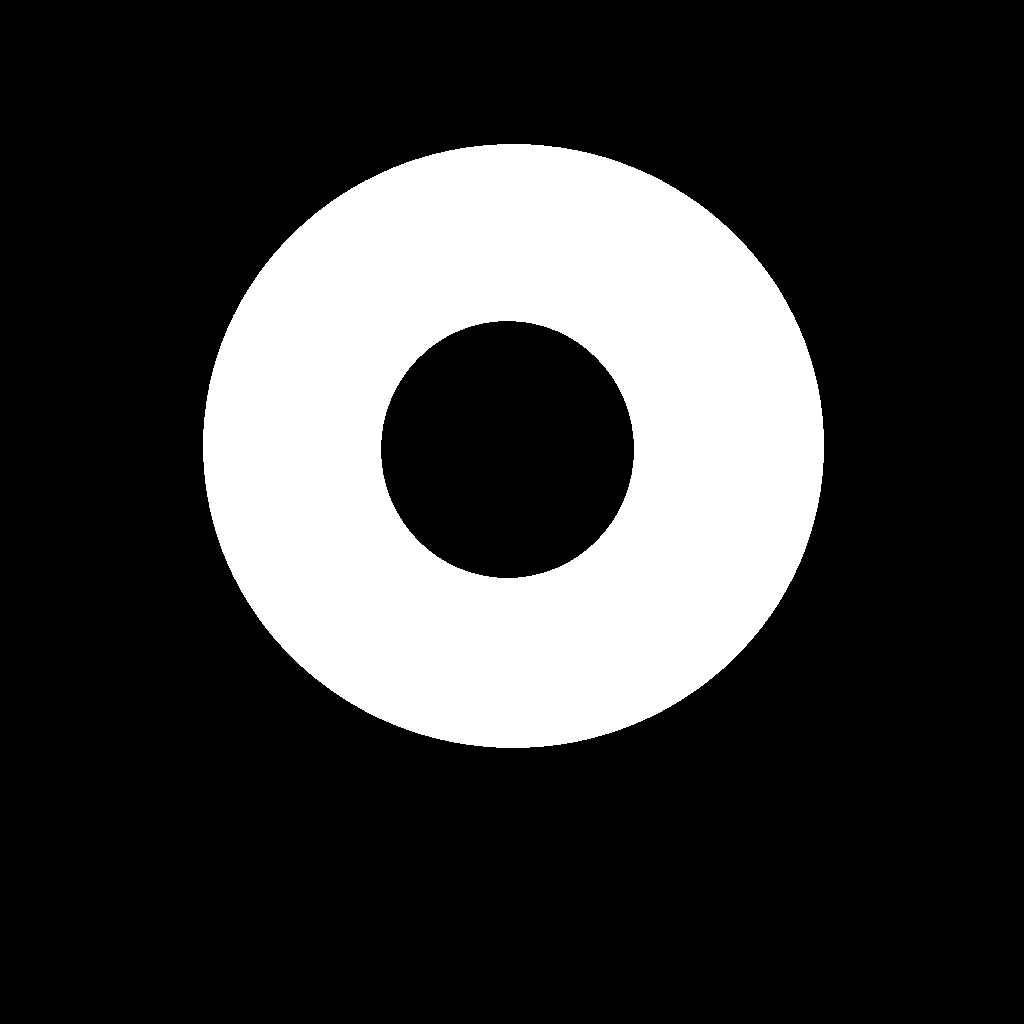} 
    \caption{A target phantom}
    \label{fig:target_annulus}
    \end{subfigure}
    \centering
    \hspace{2mm}
    \begin{subfigure}[t]{0.3  \textwidth}
    \includegraphics[width=\textwidth]{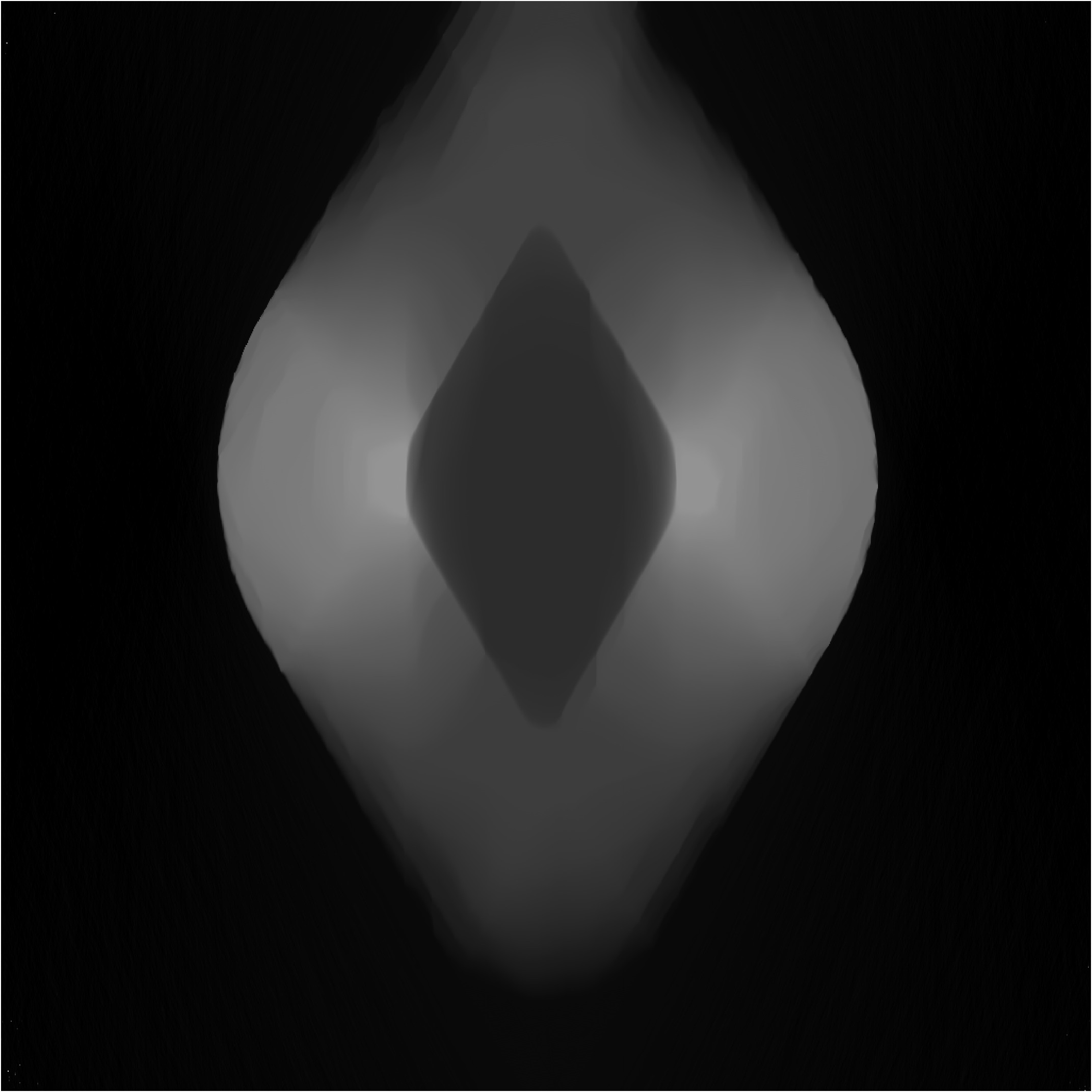} 
    \caption{An initial reconstruction of size $1024 \times 1024$ computed using total variation regularization with regularization parameter $\alpha=1$ and 2000 iterations.}
    \end{subfigure}
    \hspace{2mm}
    \begin{subfigure}[t]{0.3  \textwidth}
    \includegraphics[width=\textwidth]{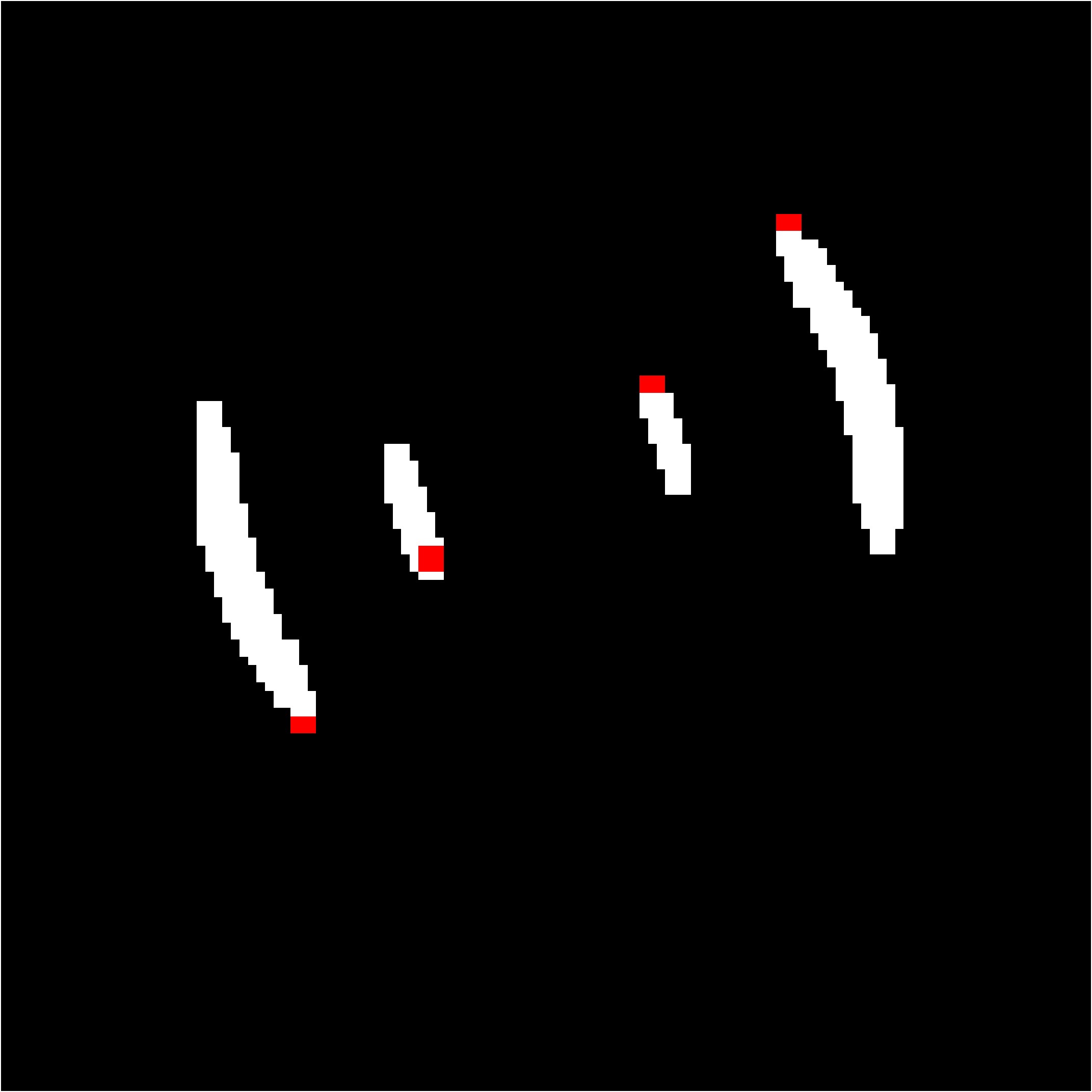} 
    \caption{Subbands $HL$ (white) of wavelet level $7$ computed using threshold value $t=0.09$ and line length $l=9$. Endpoints marked as red over the subbands.}
    \end{subfigure}
    \begin{subfigure}[t]{0.3  \textwidth}
    \includegraphics[width=\textwidth]{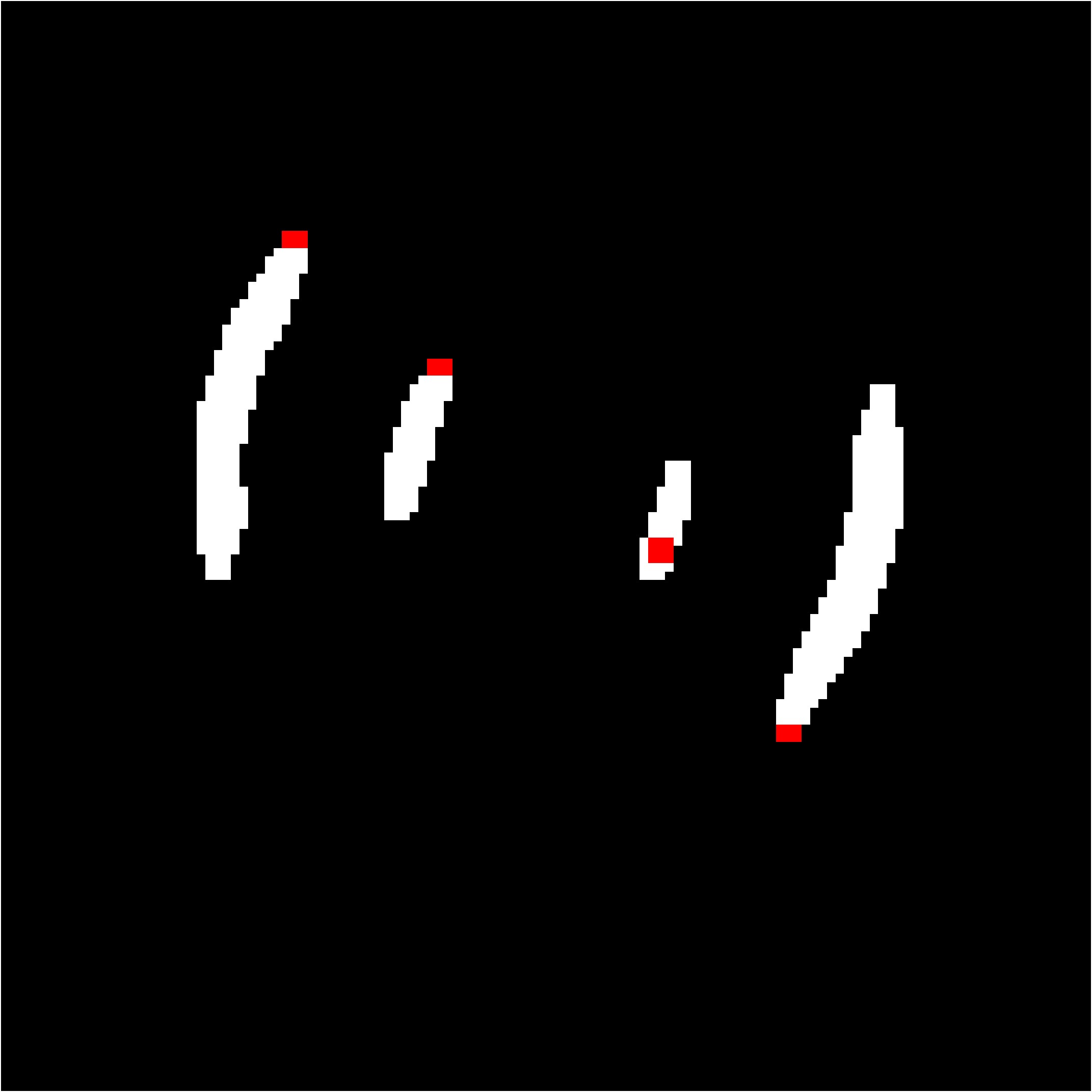} 
    \caption{ Subbands $H\oL$ (white) of wavelet level $7$ computed using threshold value $t=0.09$ and line length $l=9$. Endpoints marked as red over the subbands.}
    \end{subfigure}
    \hspace{2mm}
    \begin{subfigure}[t]{0.3  \textwidth}
    \includegraphics[width=\textwidth]{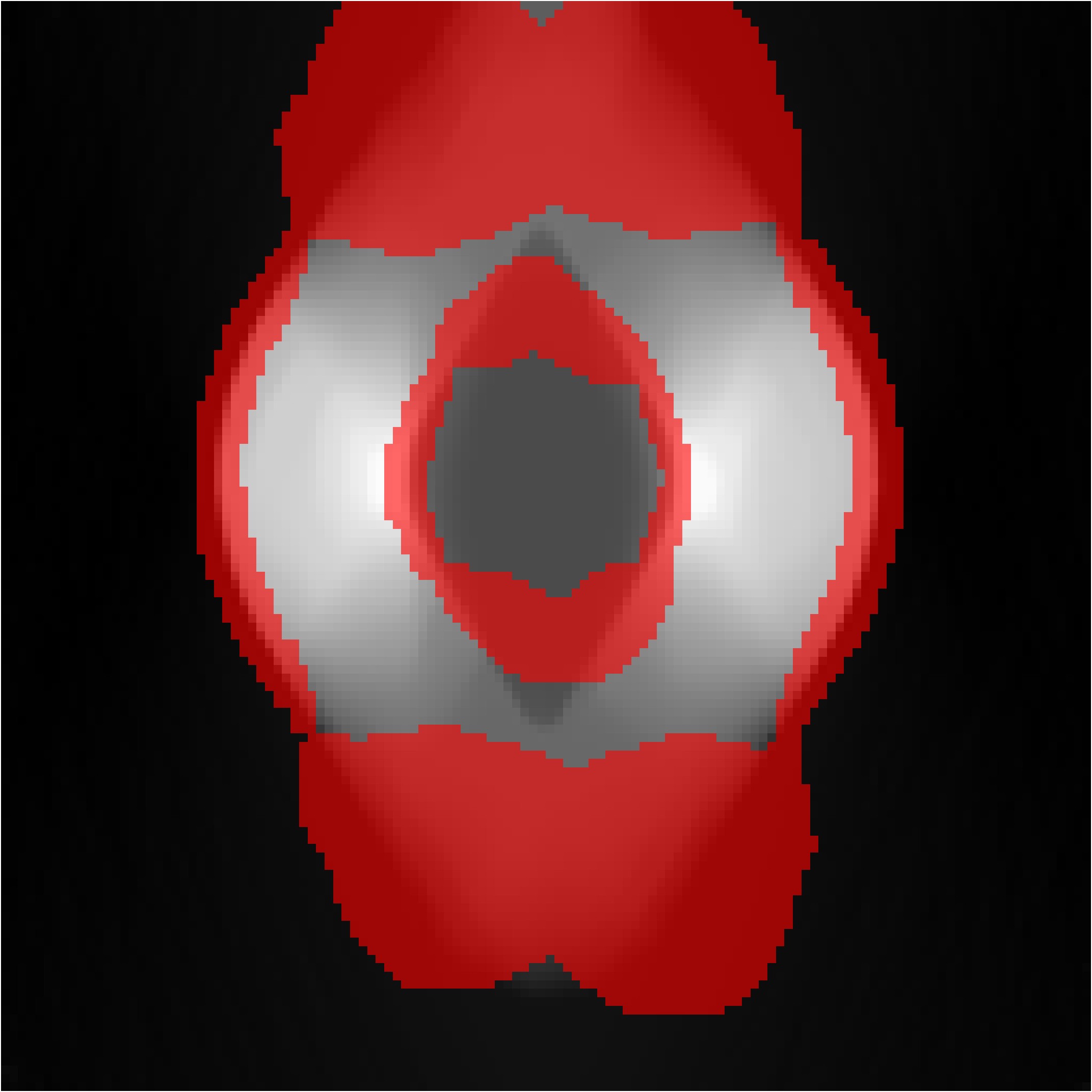} 
    \caption{Boundary neighborhood estimations over bicubic interpolated initial reconstruction.}
    \end{subfigure}
    \hspace{2mm}
\begin{subfigure}[t]{0.3  \textwidth}
    \includegraphics[width=\textwidth]{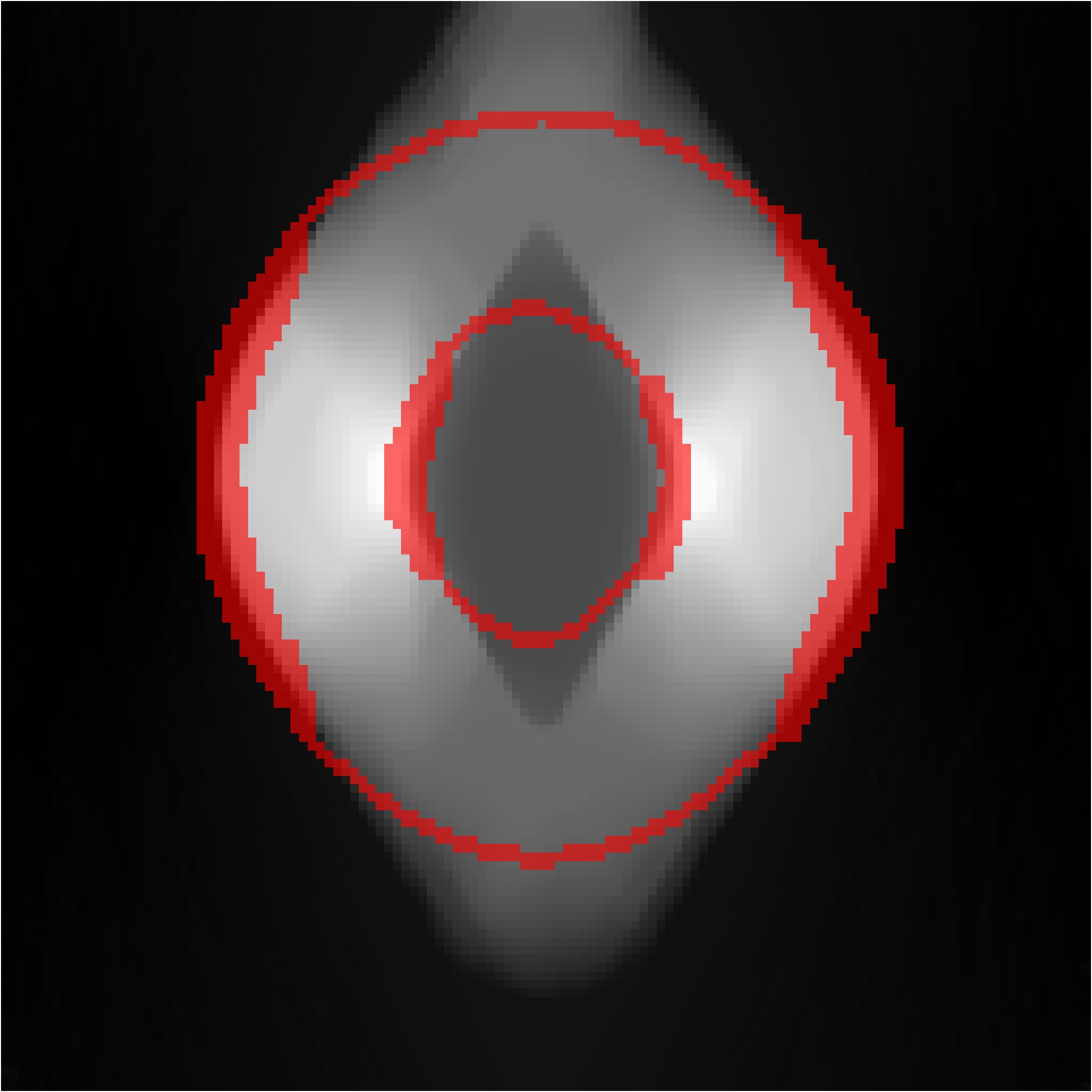} 
    \caption{Spline fill over bicubic interpolated initial reconstruction}
    \label{fig:annulus_spline}

    \end{subfigure}
    \begin{subfigure}[t]{0.3  \textwidth}
    \includegraphics[width=\textwidth]{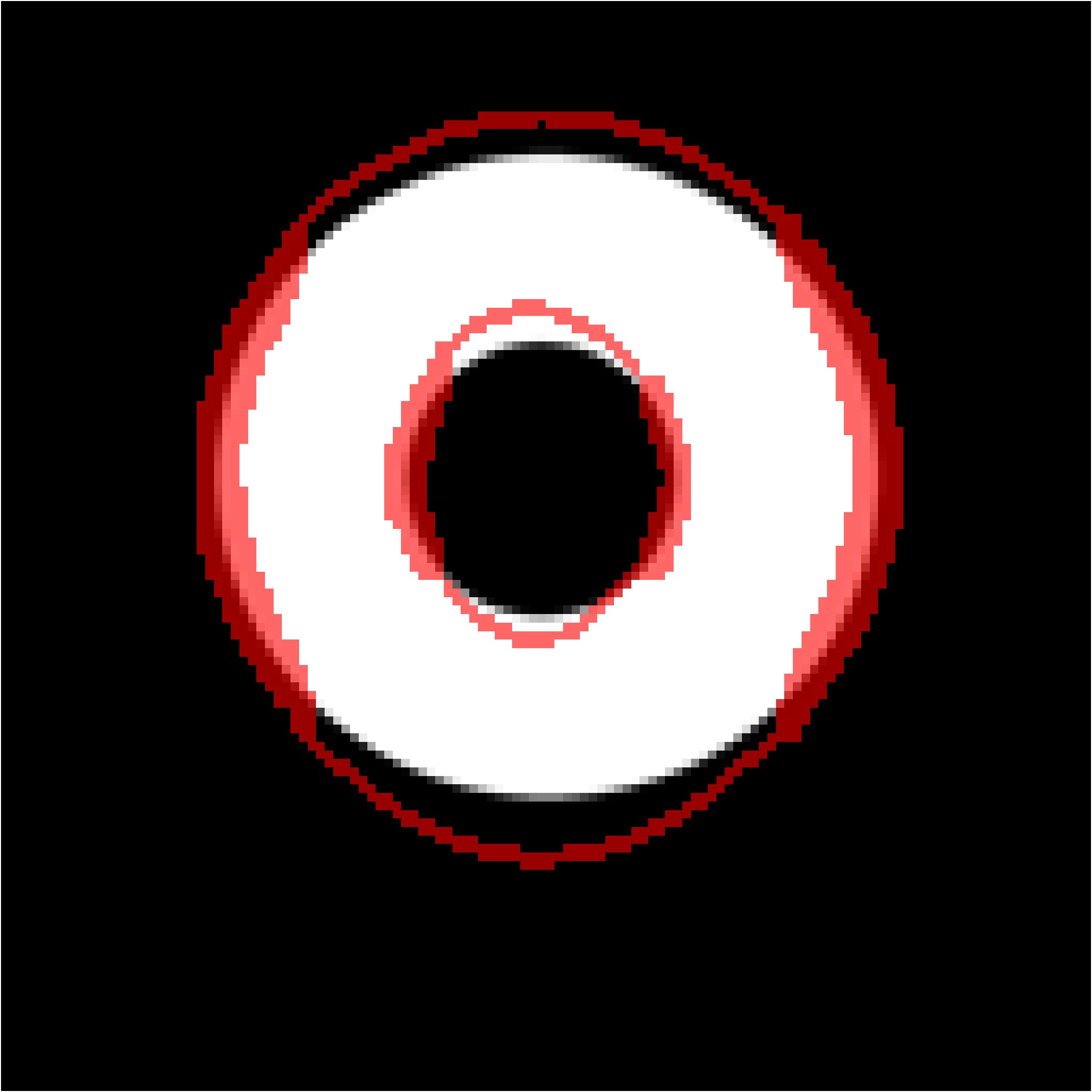} 
    \caption{Spline fill over bicubic interpolated target object.}

    \end{subfigure}
       \hspace{2mm}
\begin{subfigure}[t]{0.6  \textwidth}
    \includegraphics[width=\textwidth]{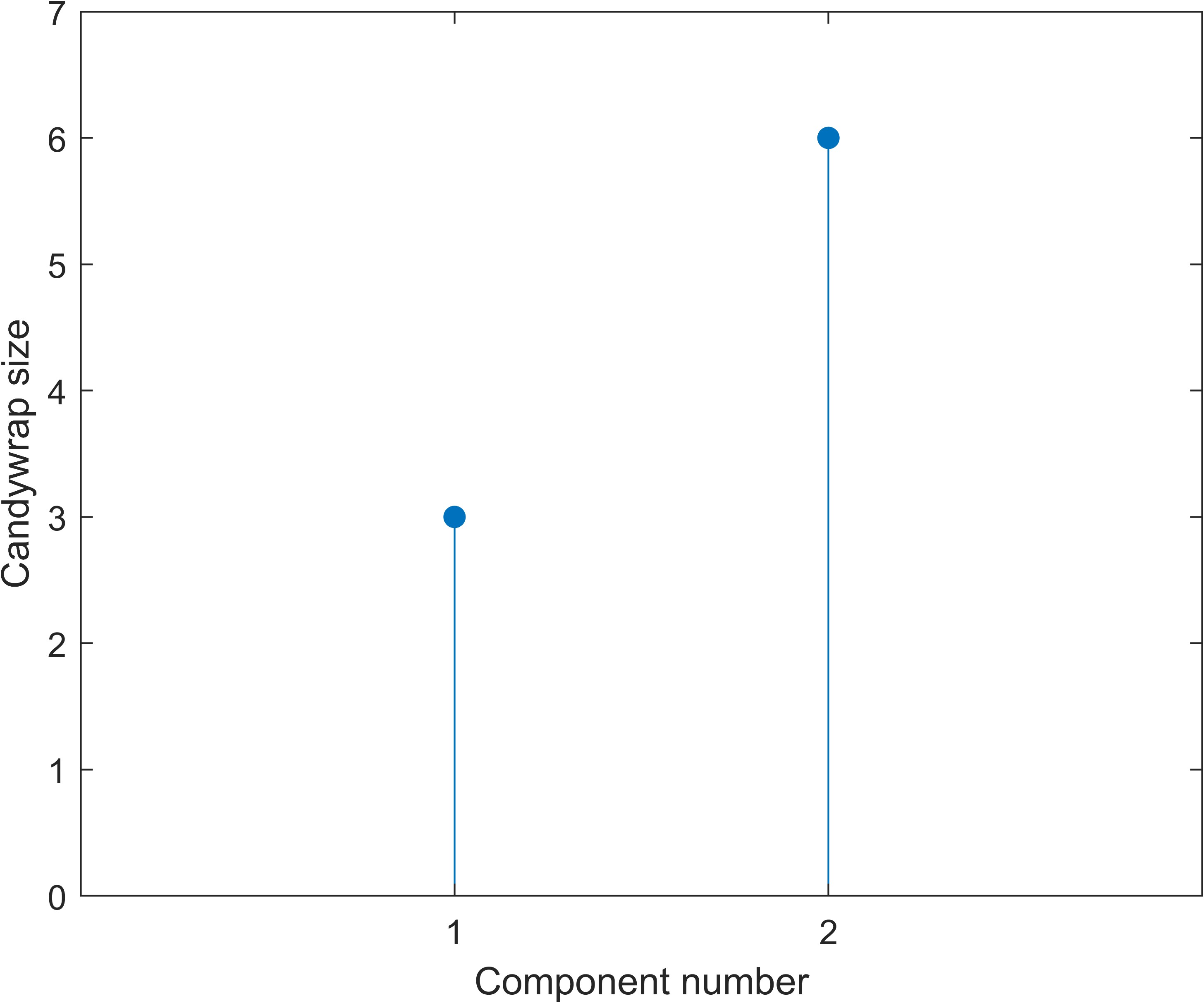}
    \caption{Number of found components and birth values of those, i.e., the size of the candywrap masks}
    \end{subfigure}
    
    \caption{The simplest target phantom containing only one inclusion which is not simply connected. Both boundary components are found with TILT.}
    \label{fig:annulus_results}
\end{figure}

\begin{figure}[!ht]
    \centering
    \begin{subfigure}[t]{0.3  \textwidth}
    \includegraphics[width=\textwidth]{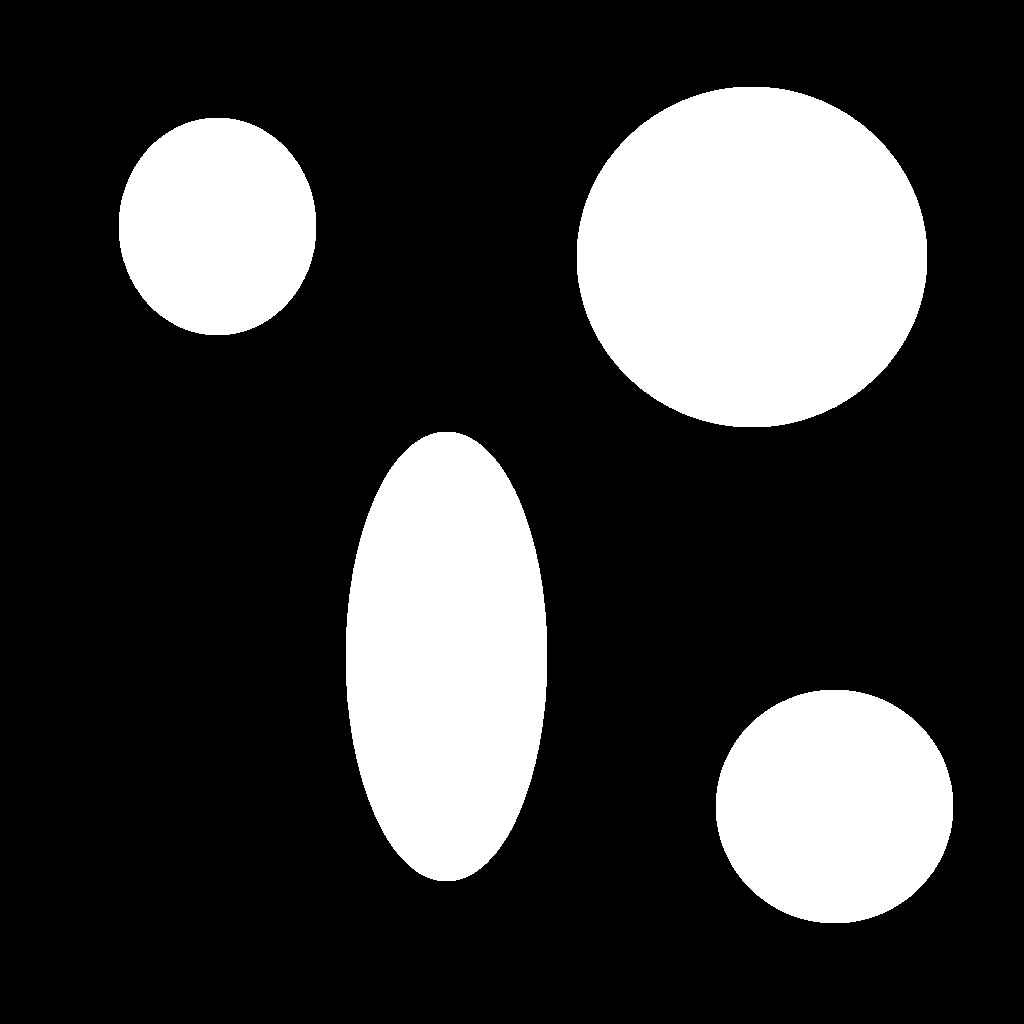} 
    \caption{A target phantom}
    \label{fig:target_multipleshapes}
    \end{subfigure}
        \hspace{2mm}
    \centering
    \begin{subfigure}[t]{0.3  \textwidth}
    \includegraphics[width=\textwidth]{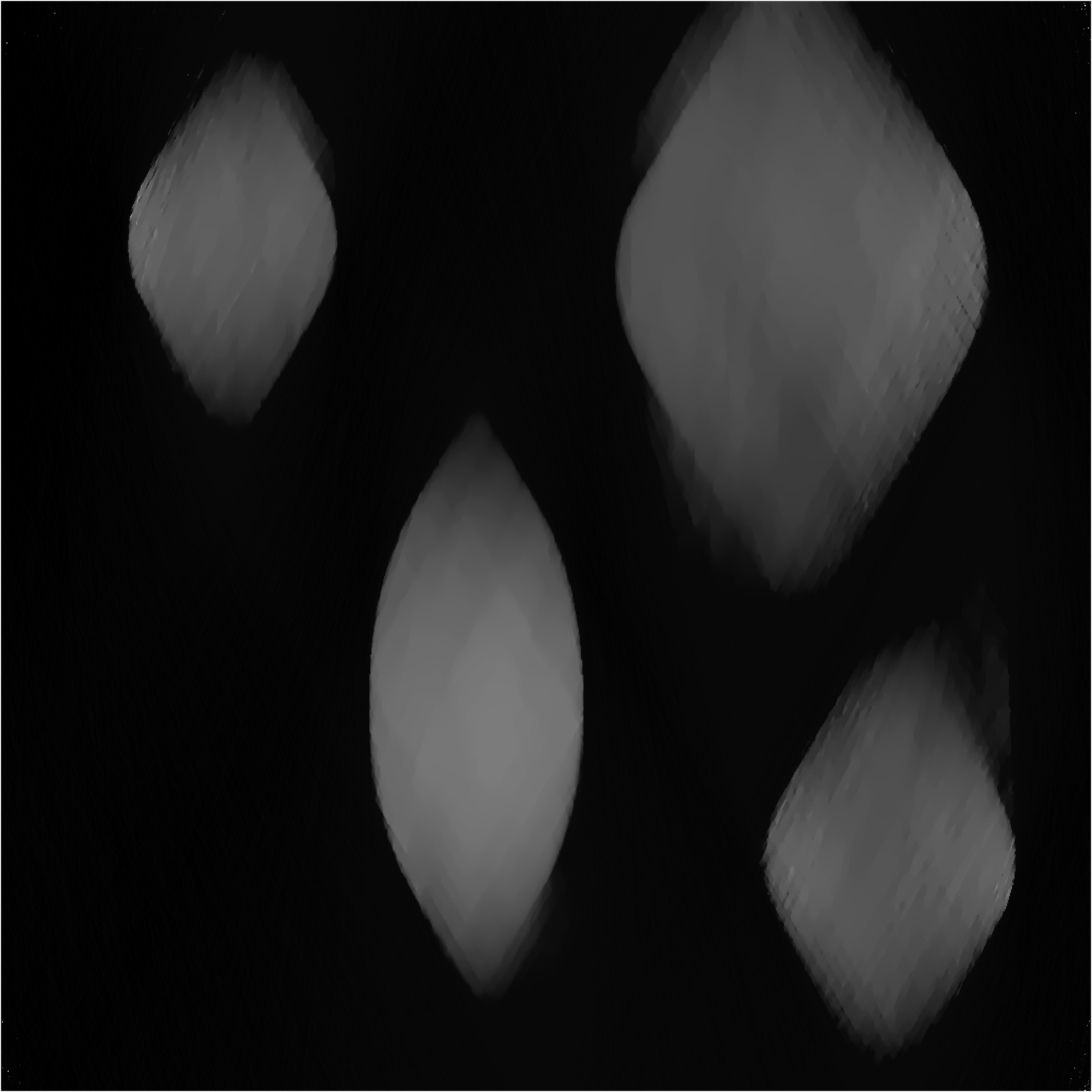} 
    \caption{An initial reconstruction of size $1024 \times 1024$ computed using total variation regularization with regularization parameter $\alpha=1$ and 2000 iterations.}
    \end{subfigure}
            \hspace{2mm}
    \begin{subfigure}[t]{0.3  \textwidth}
    \includegraphics[width=\textwidth]{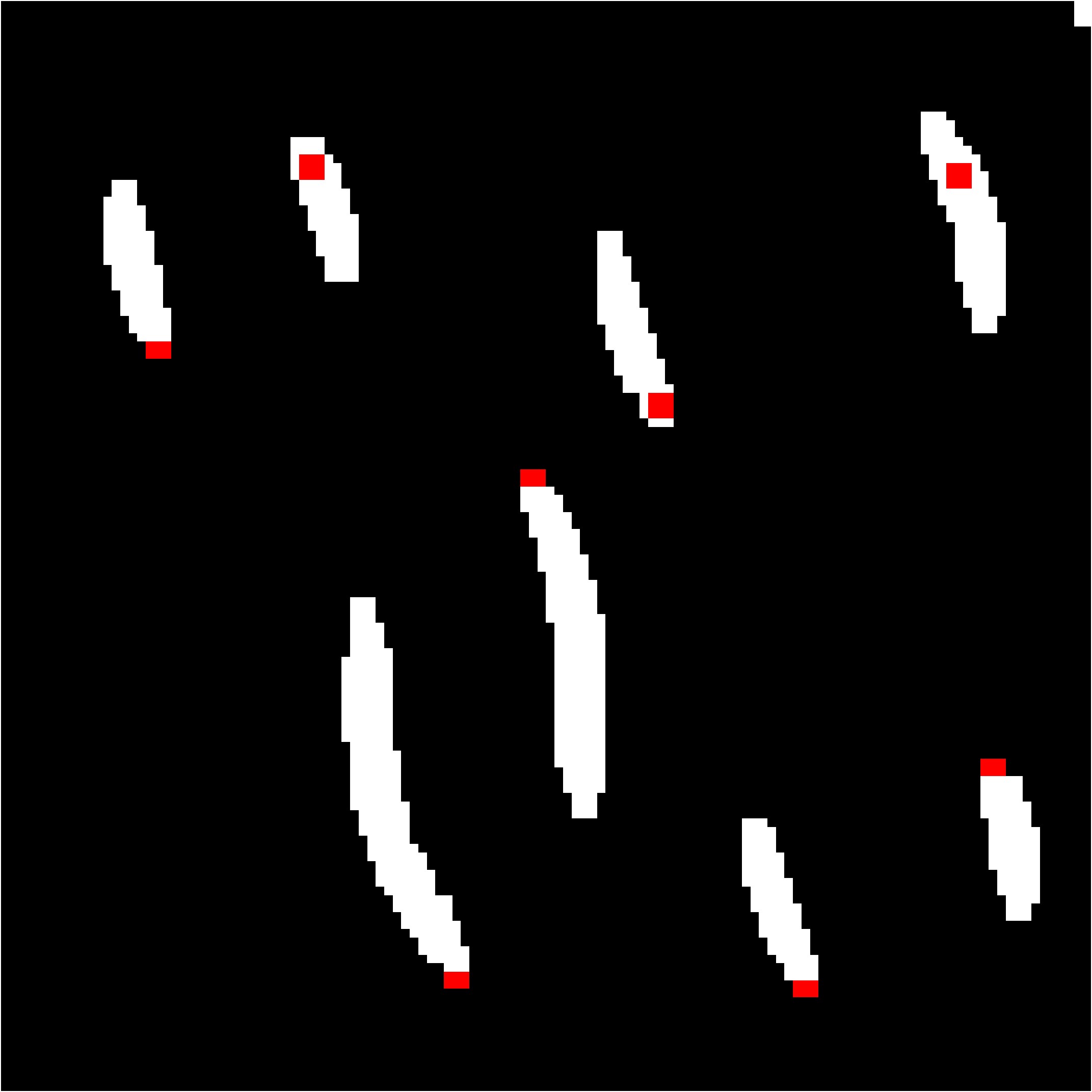} 
    \caption{Subbands $HL$ (white) of wavelet level $7$ computed using threshold value $t=0.1$ and line length $l=9$. Endpoints marked as red over the subbands.}
    \end{subfigure}
    \begin{subfigure}[t]{0.3  \textwidth}
    \includegraphics[width=\textwidth]{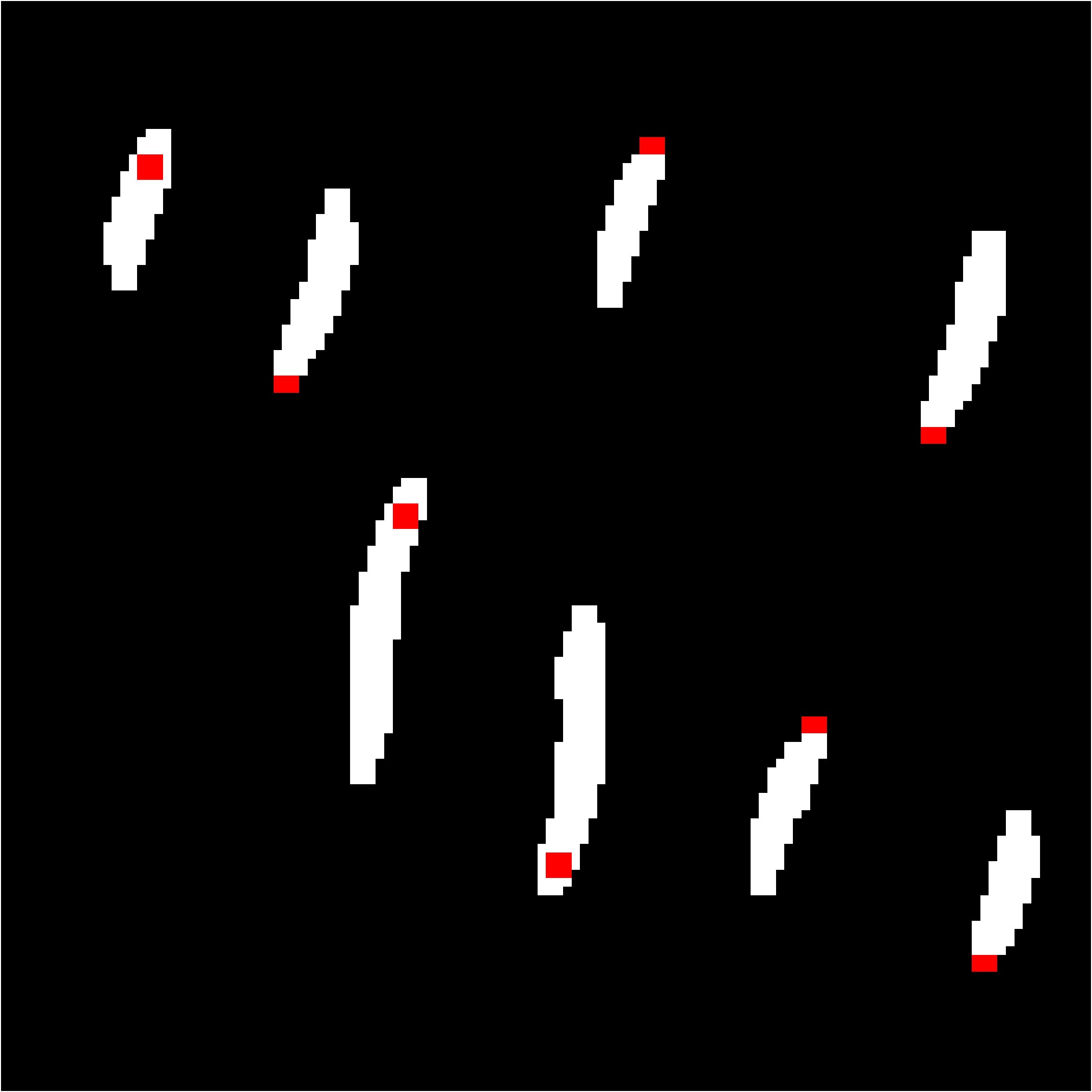} 
    \caption{Subbands $H\oL$ (white) of wavelet level $7$ computed using threshold value $t=0.1$ and line length $l=9$. Endpoints marked as red over the subbands.}
    \end{subfigure}
    \hspace{2mm}
    \begin{subfigure}[t]{0.3  \textwidth}
    \includegraphics[width=\textwidth]{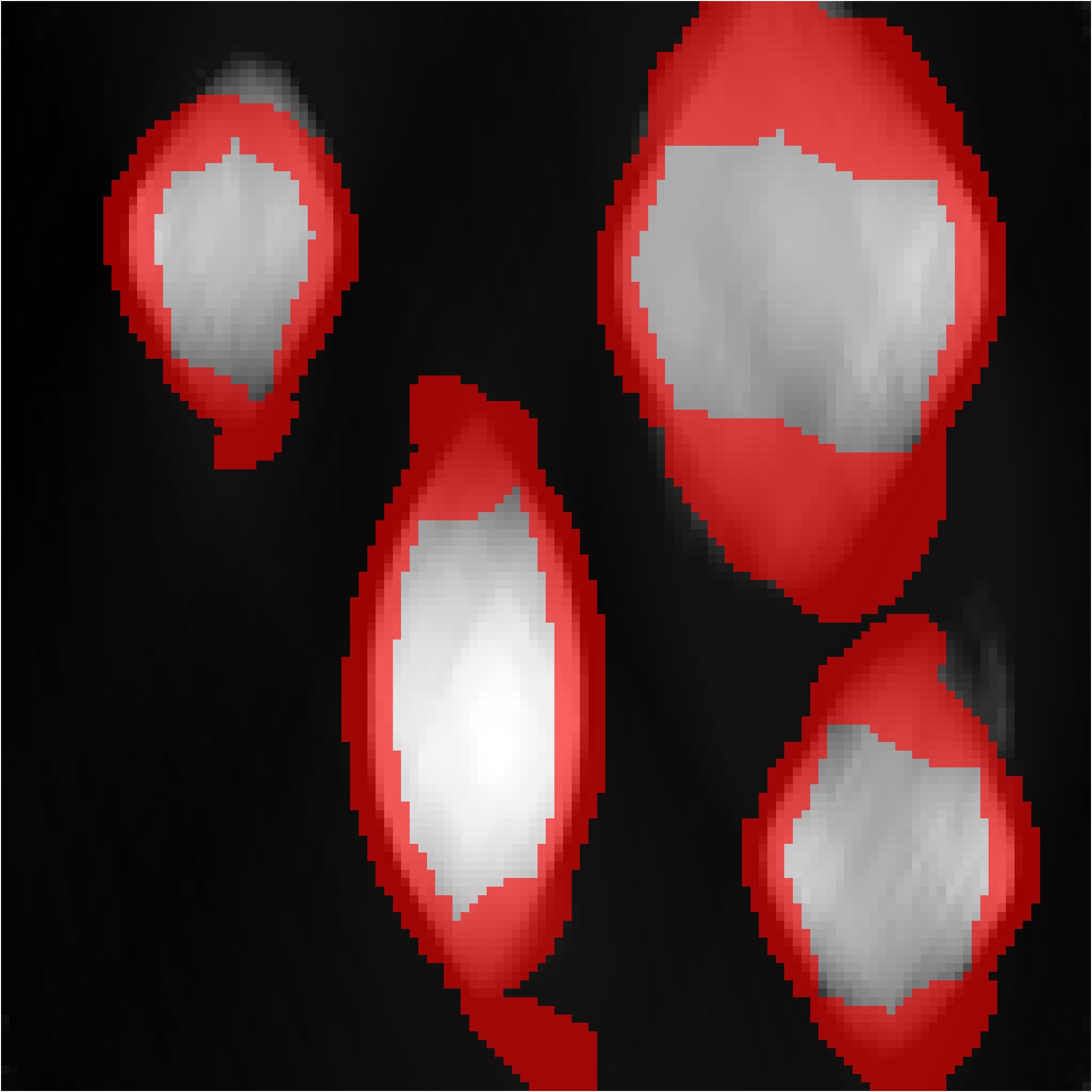} 
    \caption{Boundary neighborhood estimations over bicubic interpolated initial reconstruction.}
    \end{subfigure}
    \hspace{2mm}
\begin{subfigure}[t]{0.3  \textwidth}
    \includegraphics[width=\textwidth]{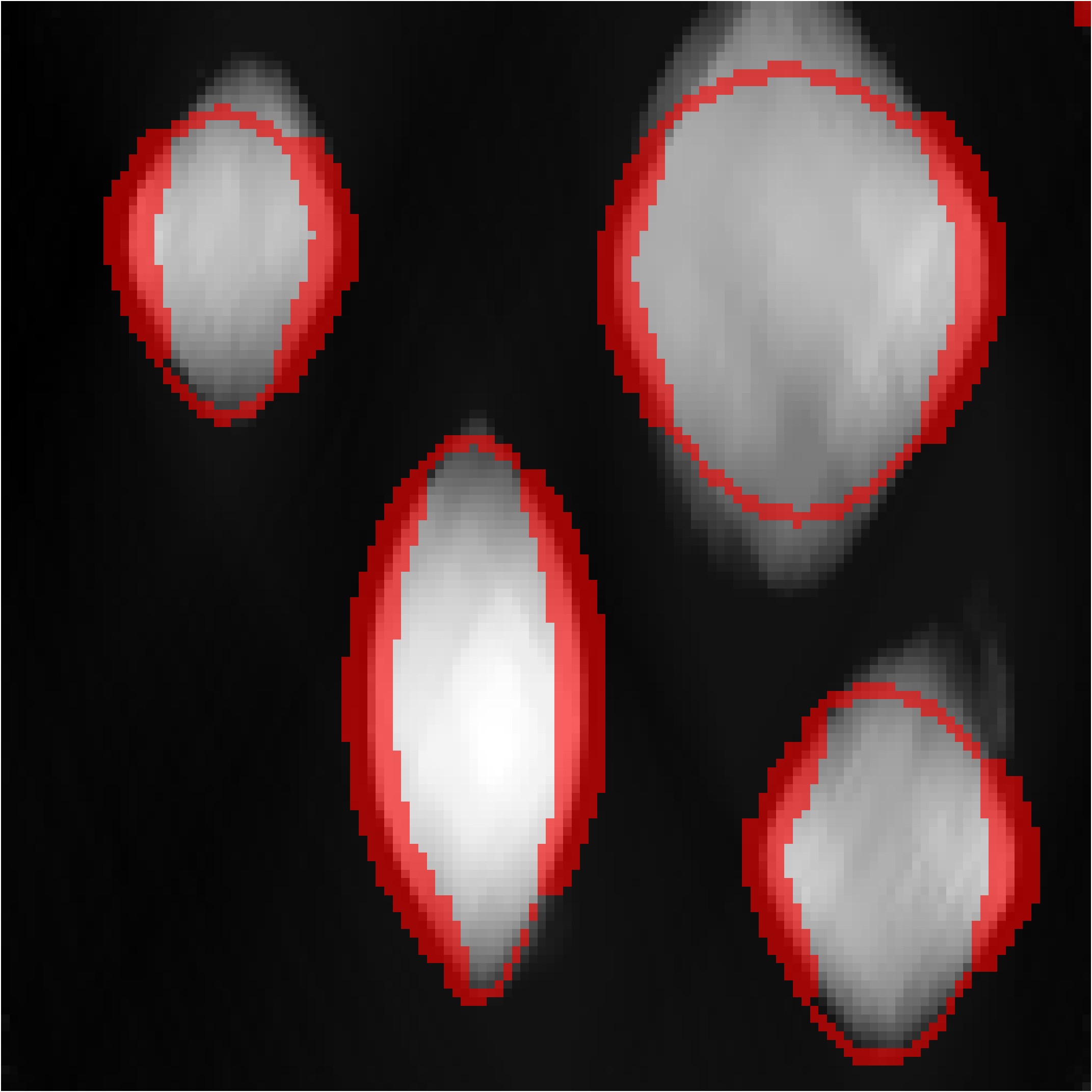} 
    \caption{Spline fill over bicubic interpolated initial reconstruction}
    \label{fig:shapes_spline}
    \end{subfigure}
    \begin{subfigure}[t]{0.3  \textwidth}
    \includegraphics[width=\textwidth]{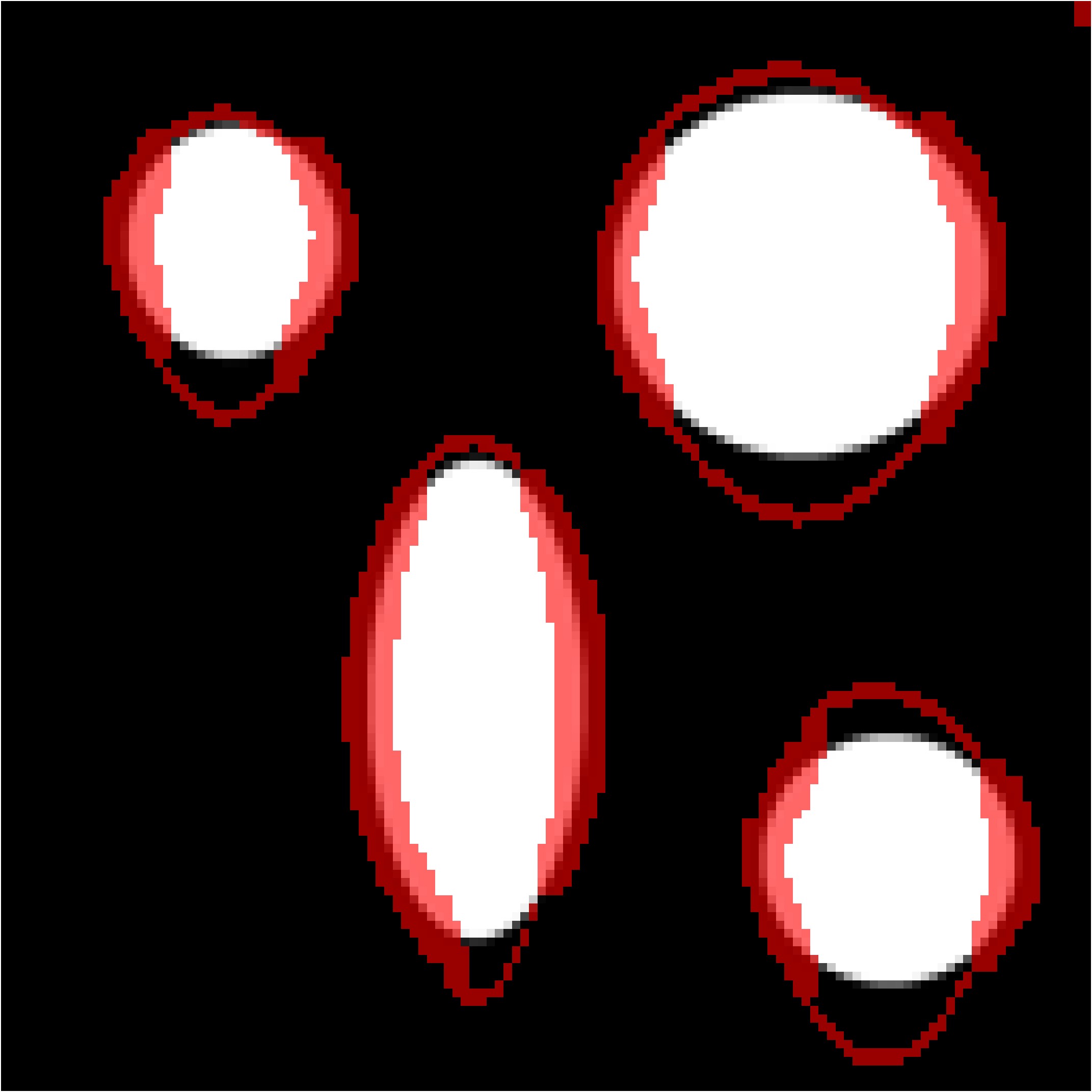} 
    \caption{Spline fill over bicubic interpolated target object.}
    \end{subfigure}
           \hspace{2mm}
\begin{subfigure}[t]{0.6  \textwidth}
    \includegraphics[width=\textwidth]{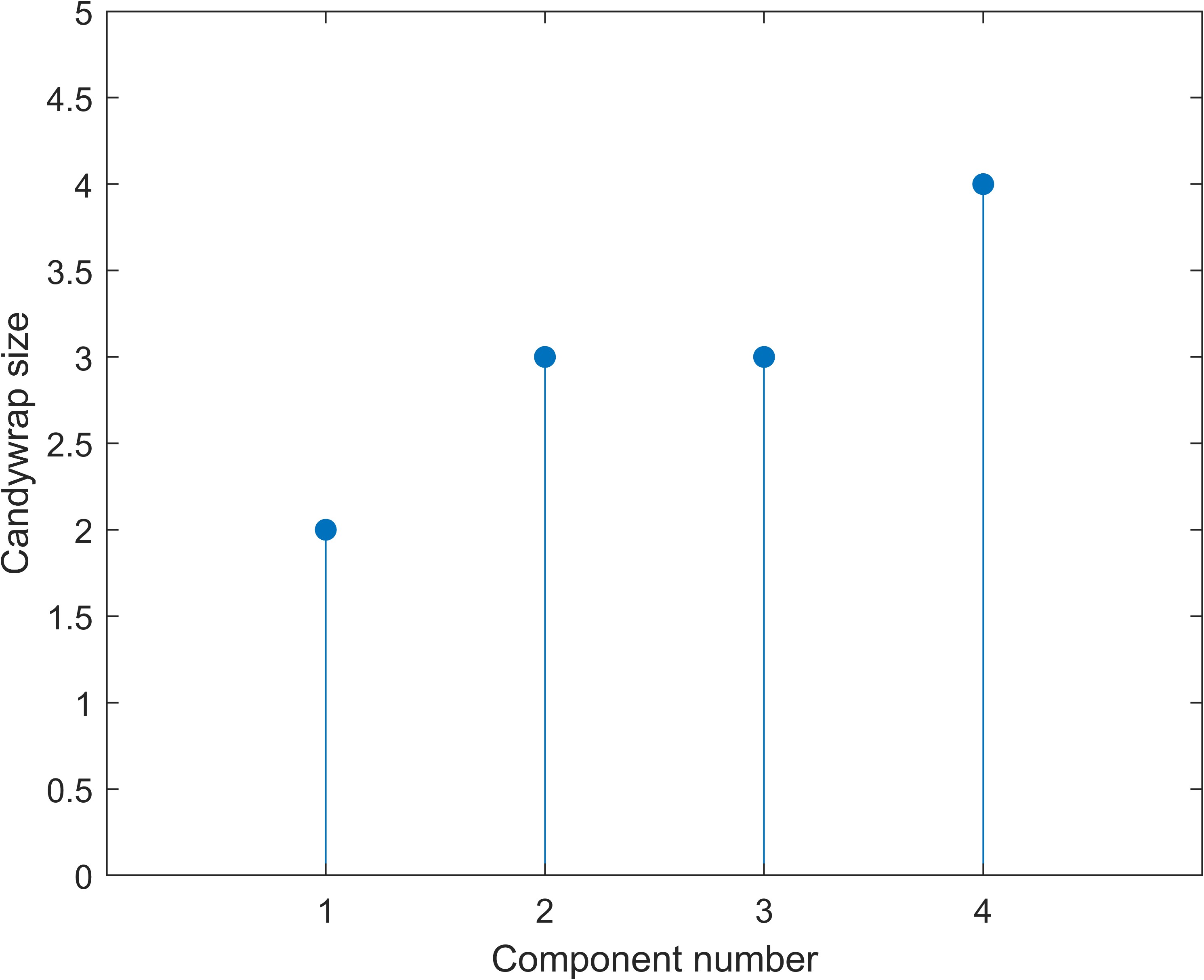}
    \caption{Number of components found and their birth values, i.e., the size of the candywrap masks}
    \end{subfigure}
    
    \caption{A target phantom containing multiple well-separated simple connected inclusions. TILT found all boundary components.}
    \label{fig:ellipses_results}
\end{figure}

\begin{figure}[!ht]
    \centering
    \begin{subfigure}[t]{0.3  \textwidth}
    \includegraphics[width=\textwidth]{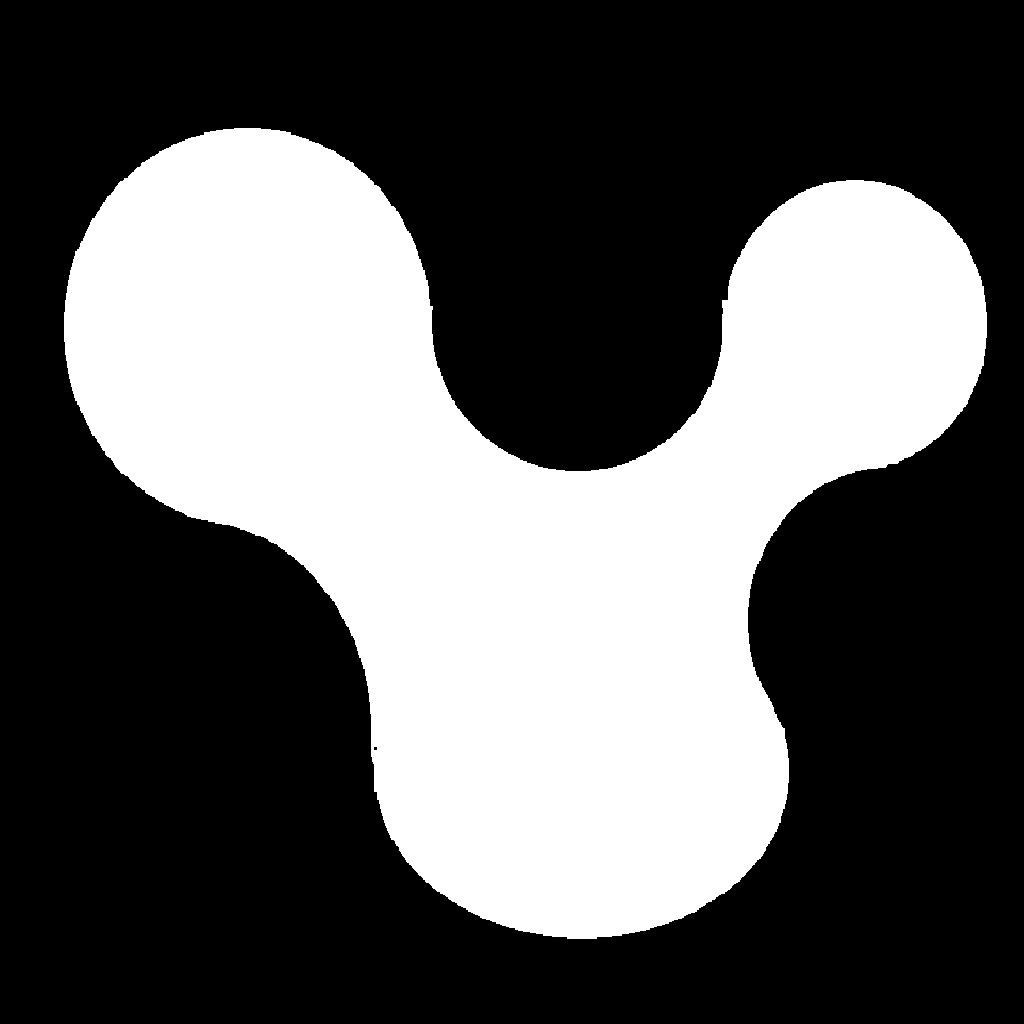} 
    \caption{A target phantom}
    \label{fig:target_blob}
    \end{subfigure}
    \centering
            \hspace{2mm}
    \begin{subfigure}[t]{0.3  \textwidth}
    \includegraphics[width=\textwidth]{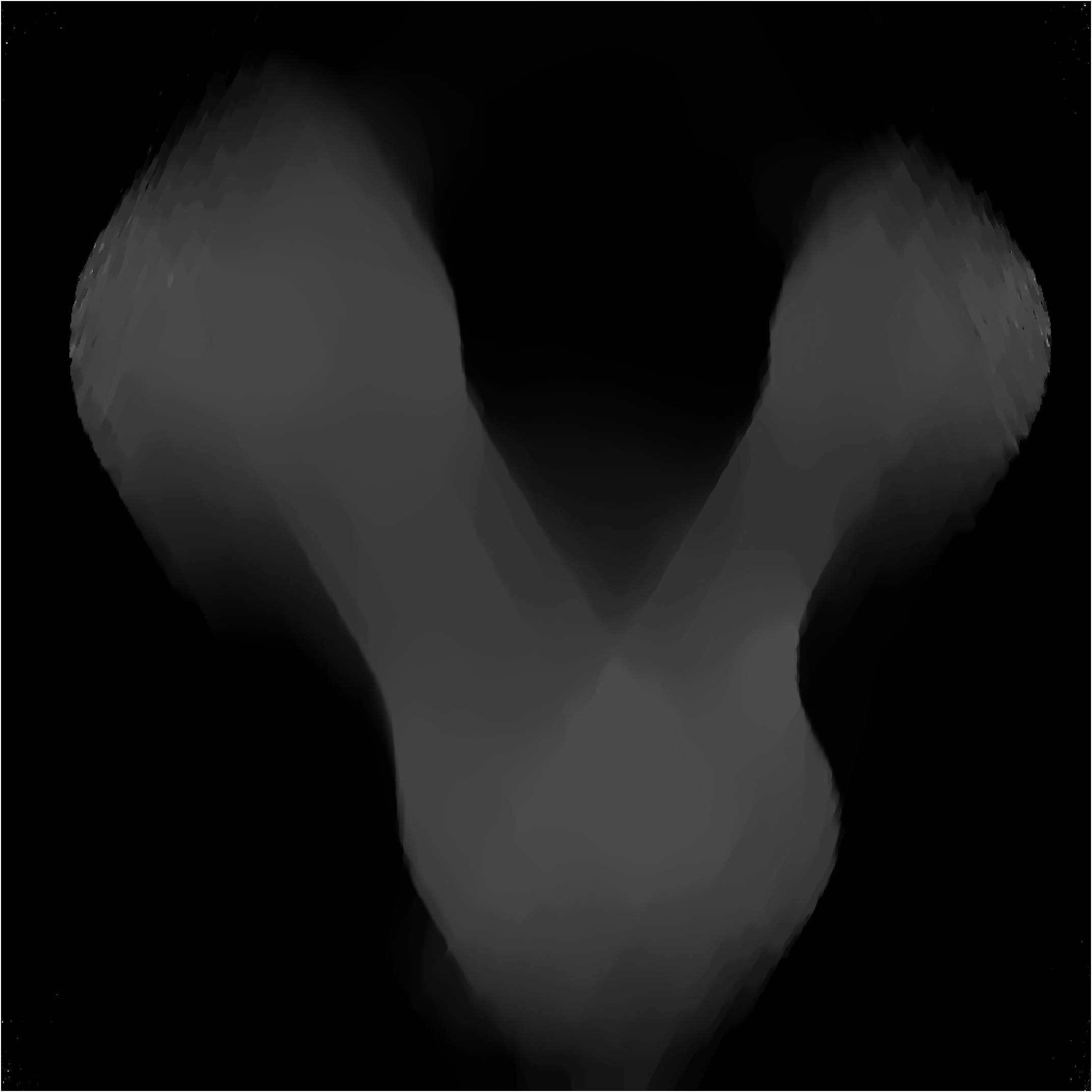} 
    \caption{An initial reconstruction of size $1024 \times 1024$ computed using total variation regularization with regularization parameter $\alpha=1$ and 2000 iterations.}
    \end{subfigure}
            \hspace{2mm}
    \begin{subfigure}[t]{0.3  \textwidth}
    \includegraphics[width=\textwidth]{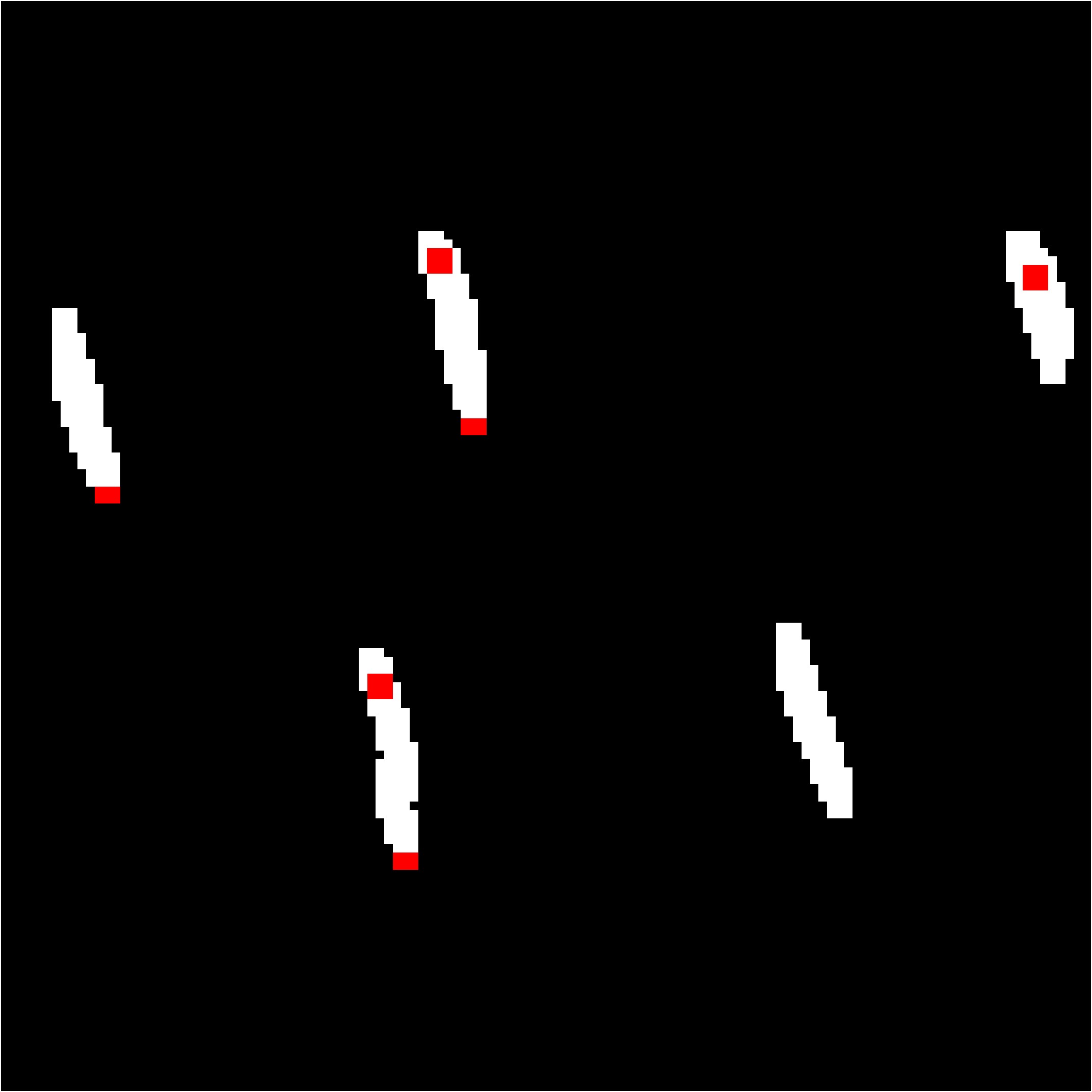} 
    \caption{Subbands $HL$ (white) of wavelet level $7$ computed using threshold value $t=0.1$ and line length $l=9$. Endpoints marked as red over the subbands.}
    \end{subfigure}
    \begin{subfigure}[t]{0.3  \textwidth}
    \includegraphics[width=\textwidth]{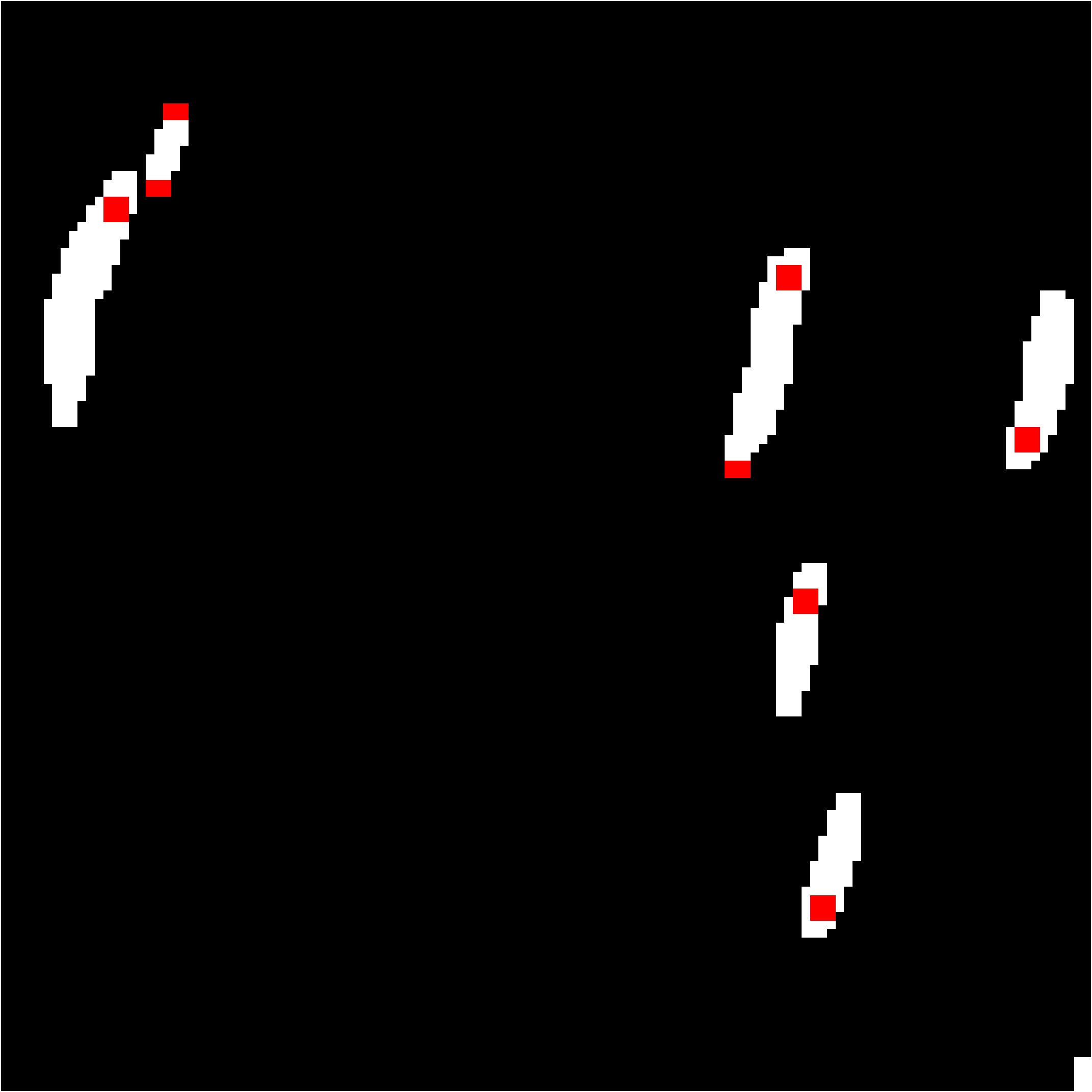} 
    \caption{Subbands $H\oL$ (white) of wavelet level $7$ computed using threshold value $t=0.1$ and line length $l=9$. Endpoints marked as red over the subbands.}
    \end{subfigure}
                \hspace{2mm}
    \begin{subfigure}[t]{0.3  \textwidth}
    \includegraphics[width=\textwidth]{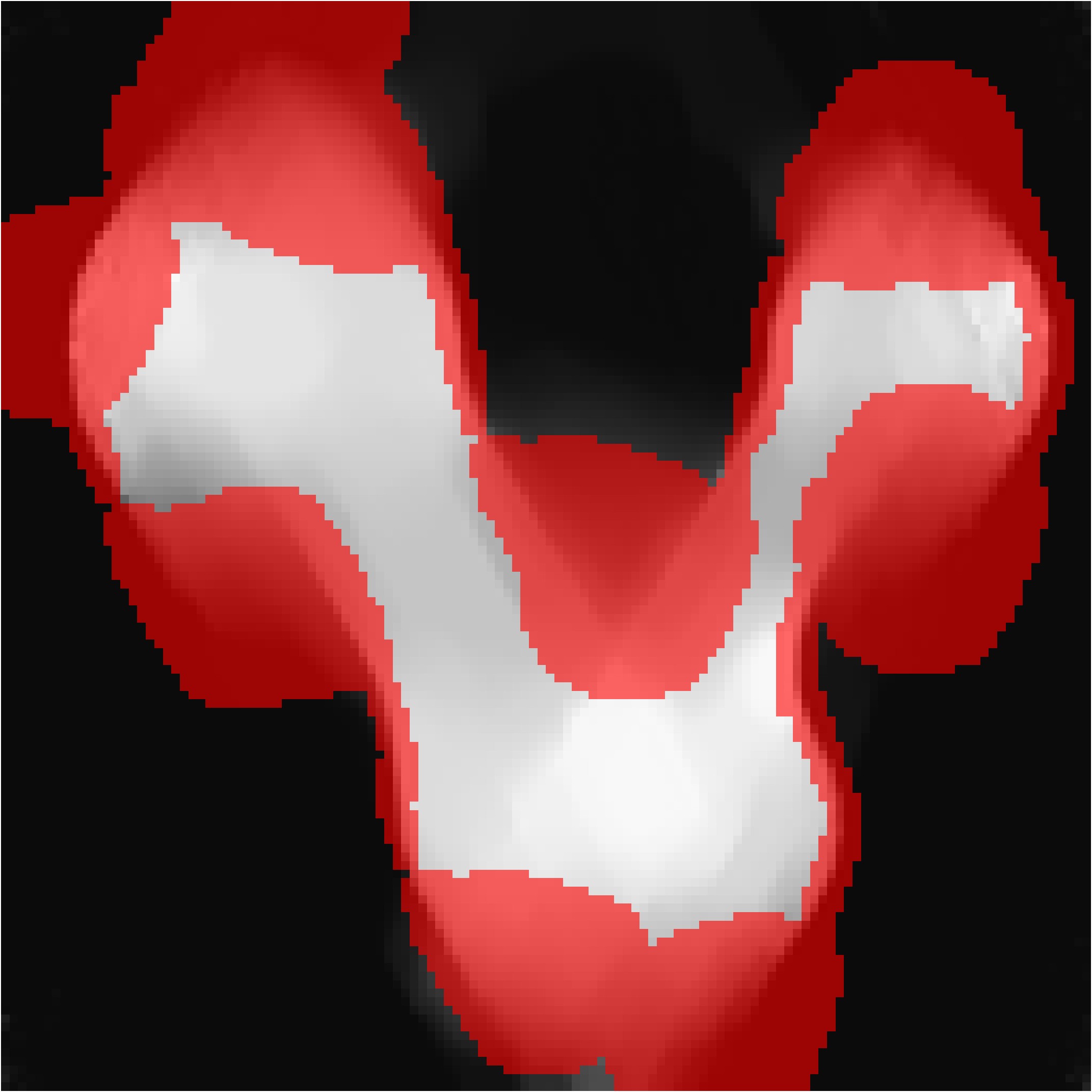} 
    \caption{Boundary neighborhood estimations over bicubic interpolated initial reconstruction.}
    \end{subfigure}
            \hspace{2mm}
\begin{subfigure}[t]{0.3  \textwidth}
    \includegraphics[width=\textwidth]{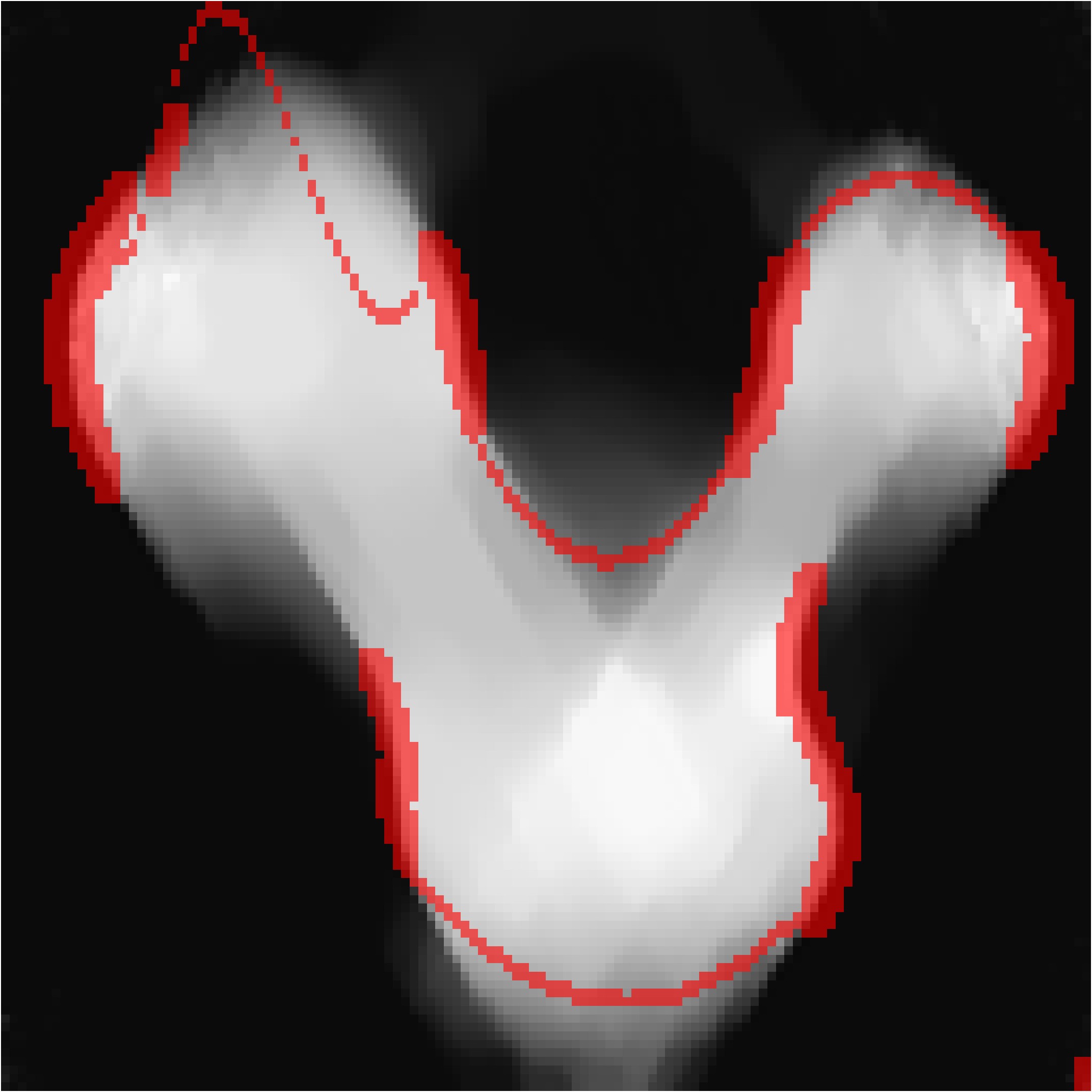} 
    \caption{Spline fill over bicubic interpolated initial reconstruction}
    \label{fig:blob_spline}
    \end{subfigure}
    \begin{subfigure}[t]{0.3  \textwidth}
    \includegraphics[width=\textwidth]{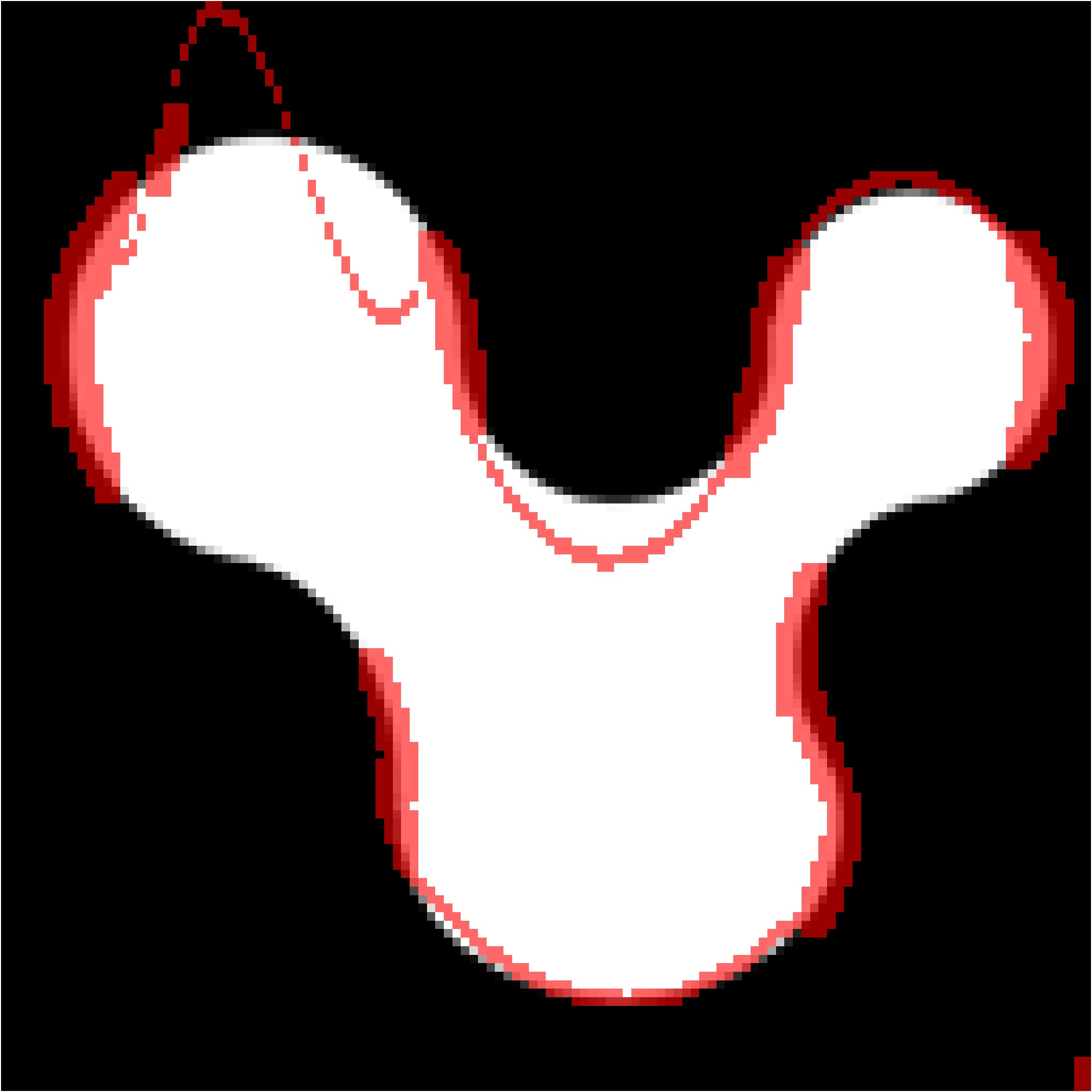} 
    \caption{Spline fill over bicubic interpolated target object.}
    \end{subfigure}
   \hspace{2mm}
\begin{subfigure}[t]{0.6  \textwidth}
    \includegraphics[width=\textwidth]{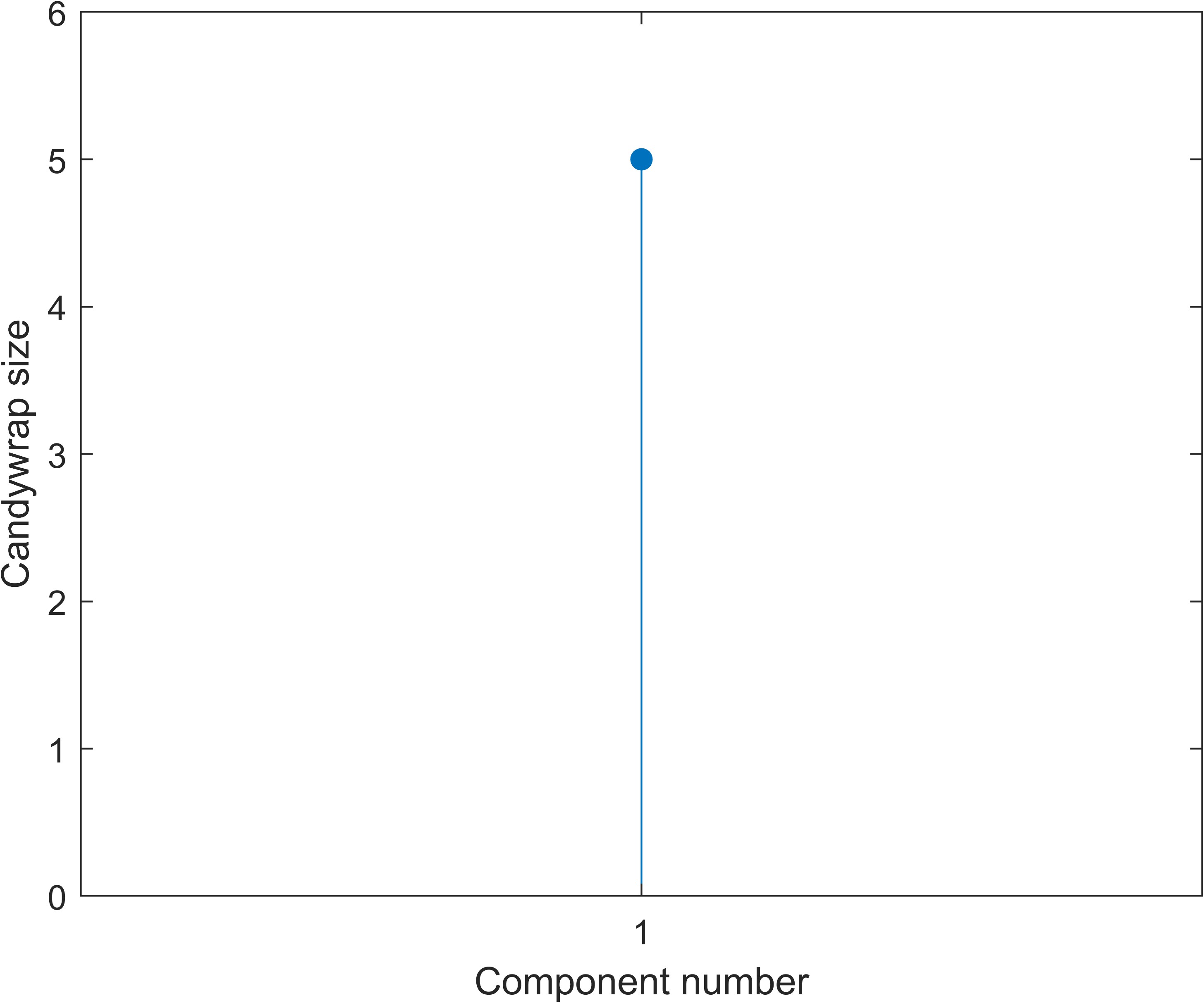}
    \caption{Number of found components and birth values of those, i.e., the size of the candywrap masks}
    \end{subfigure}

    \caption{A non-convex target phantom containing one boundary component. The proposed method found one boundary component. Imperfections in subbands, see Subfigure D, leads neighborhoods to partly wrong directions. We still found one boundary component when using TILT. }
    \label{fig:blob_results}
\end{figure}

\begin{figure}[!ht]
    \centering
    \begin{subfigure}[t]{0.3  \textwidth}
    \includegraphics[width=\textwidth]{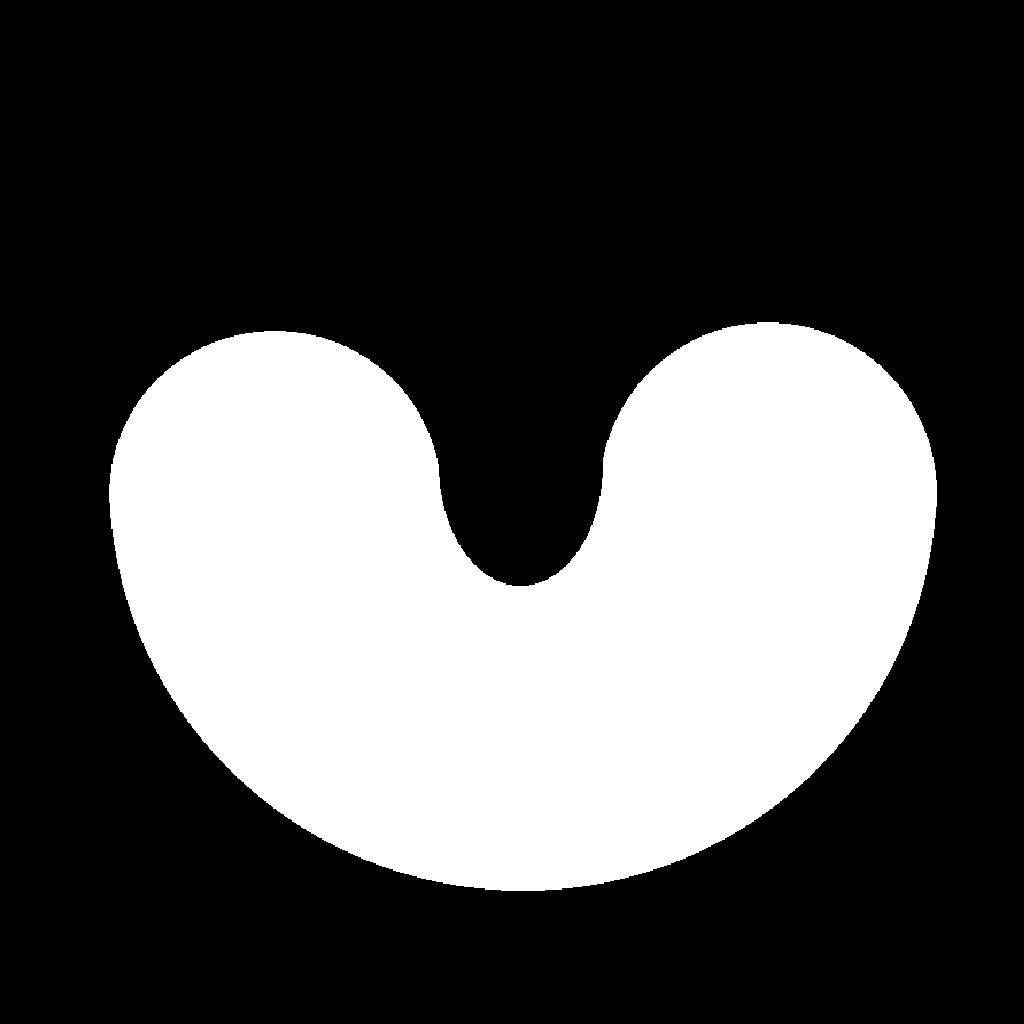} 
    \caption{A target phantom}
    \label{fig:target_makkara}
    \end{subfigure}
    \centering
                \hspace{2mm}
    \begin{subfigure}[t]{0.3  \textwidth}
    \includegraphics[width=\textwidth]{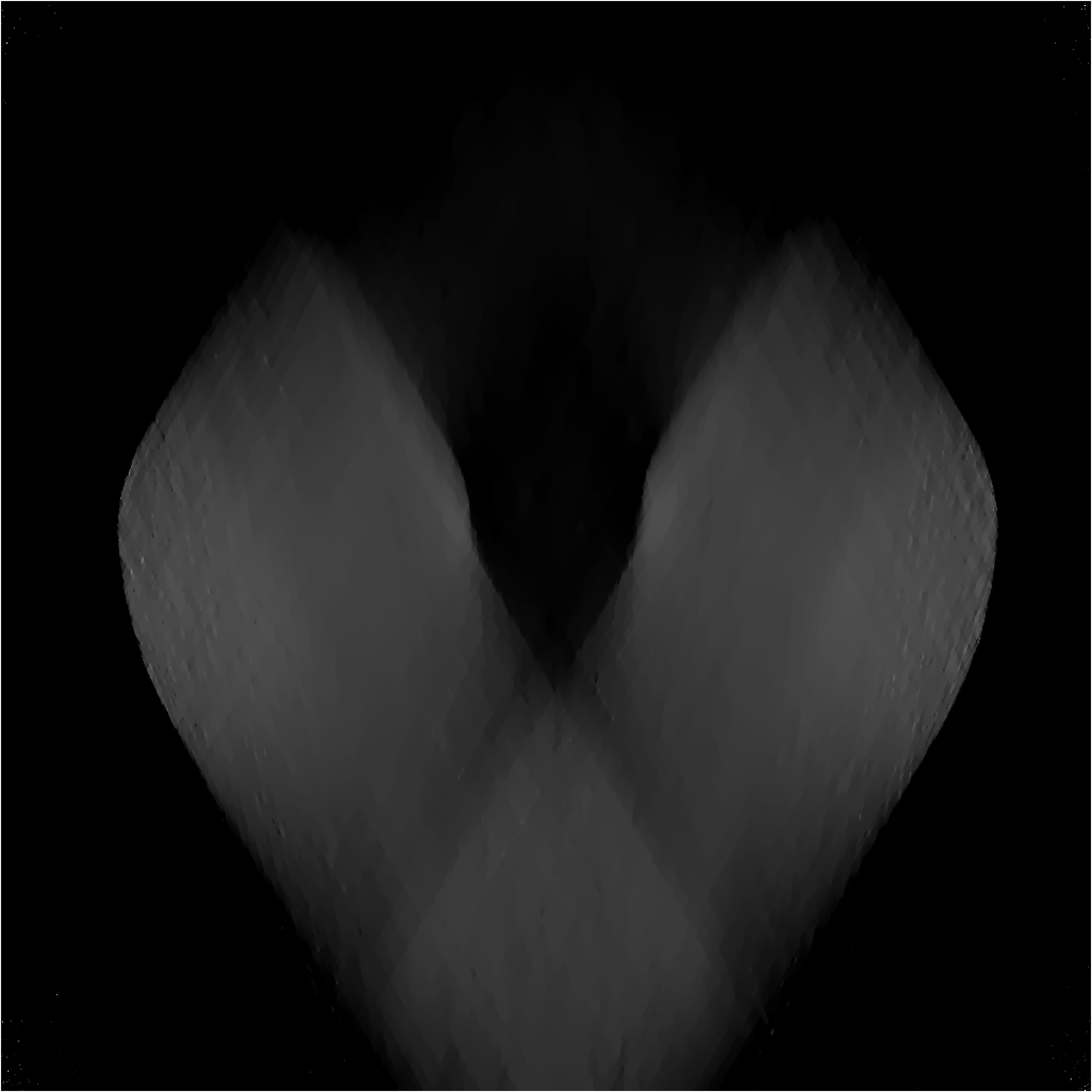} 
    \caption{An initial reconstruction of size $1024 \times 1024$ computed using total variation regularization with regularization parameter $\alpha=1$ and 2000 iterations.}
    \end{subfigure}
                \hspace{2mm}
    \begin{subfigure}[t]{0.3  \textwidth}
    \includegraphics[width=\textwidth]{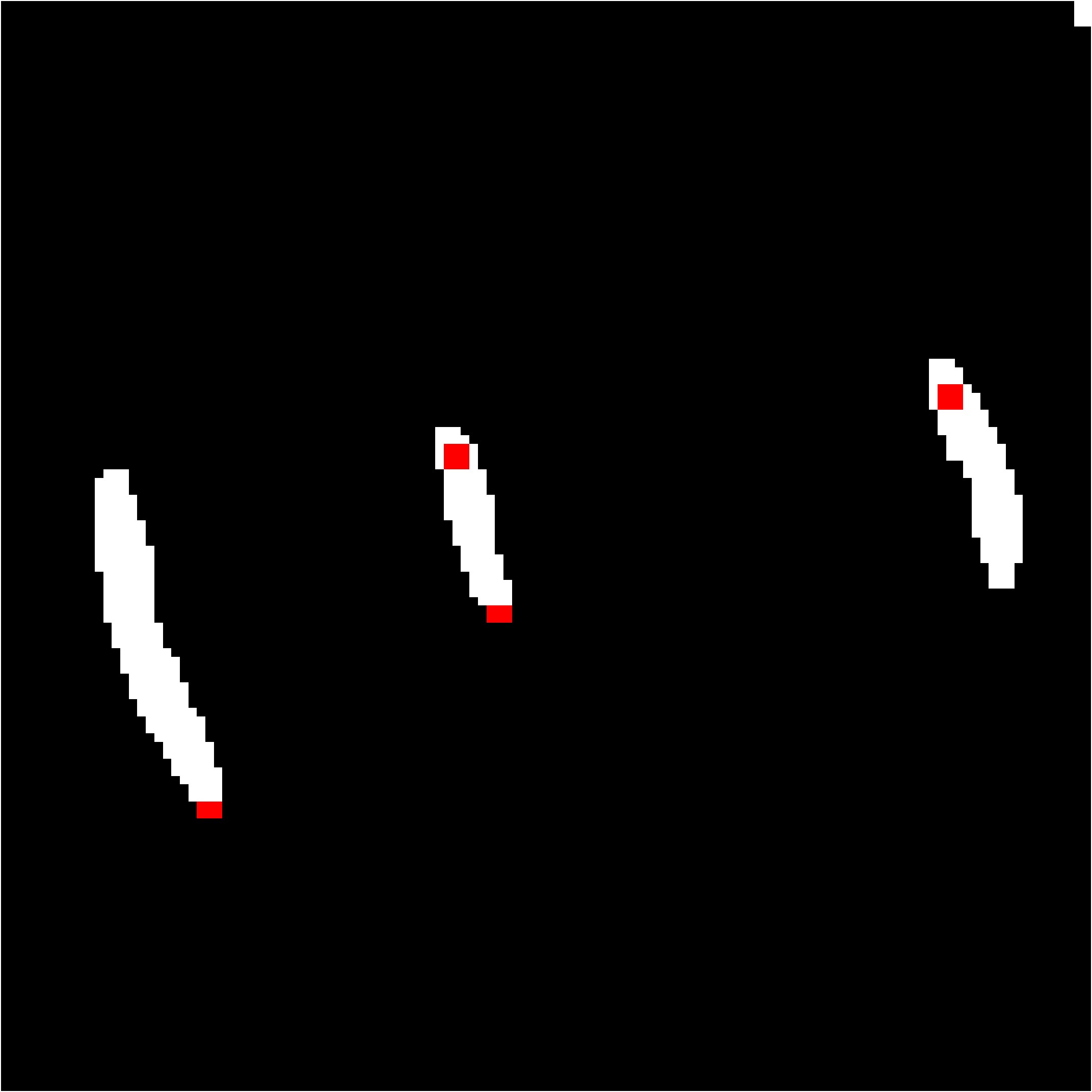} 
    \caption{Subbands $HL$ (white) of wavelet level $7$ computed using threshold value $t=0.1$ and line length $l=9$. Endpoints marked as red over the subbands.}
    \end{subfigure}
    \begin{subfigure}[t]{0.3  \textwidth}
    \includegraphics[width=\textwidth]{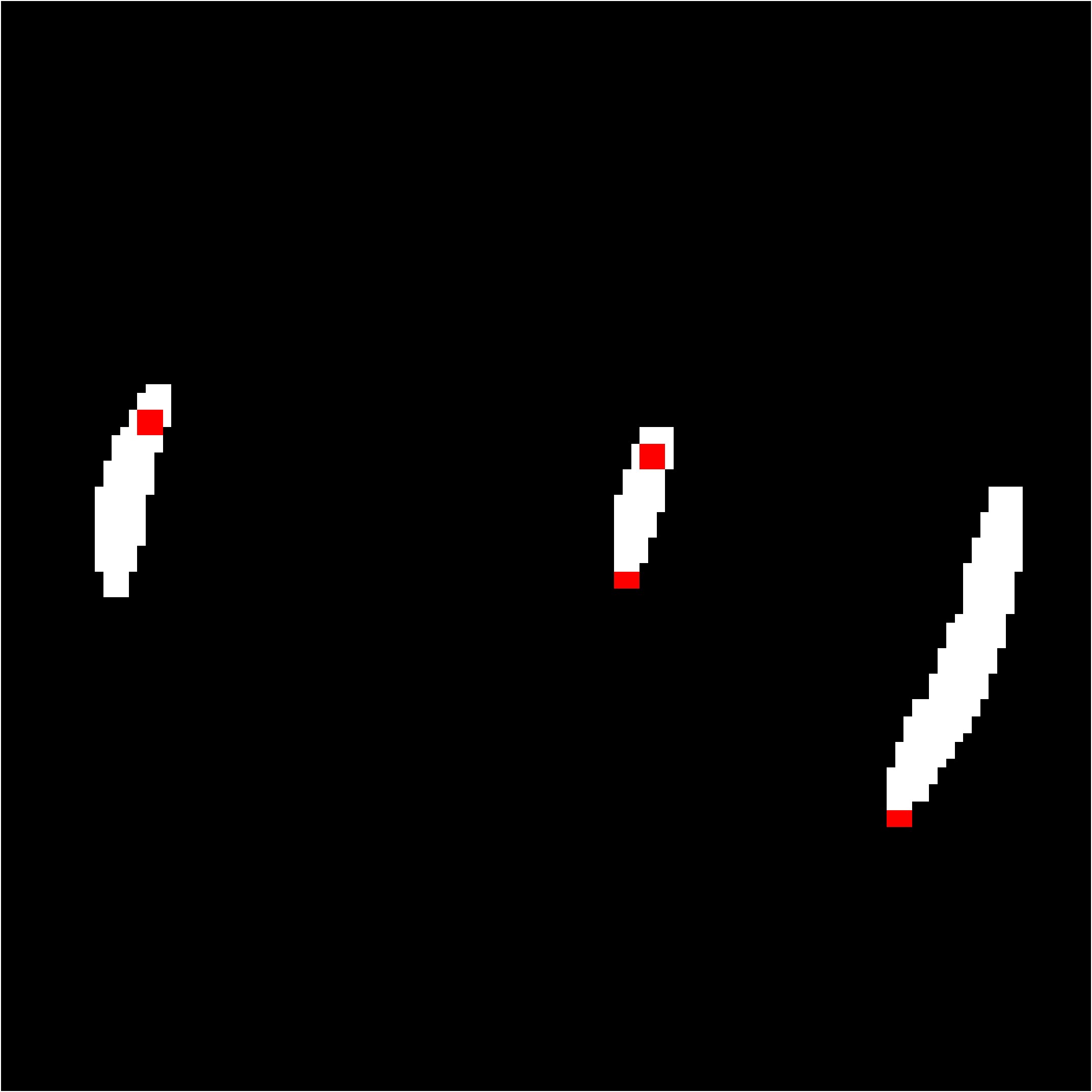} 
    \caption{Subbands $H\oL$ (white) of wavelet level $7$ computed using threshold value $t=0.1$ and line length $l=9$. Endpoints marked as red over the subbands.}
    \end{subfigure}
                \hspace{2mm}
    \begin{subfigure}[t]{0.3  \textwidth}
    \includegraphics[width=\textwidth]{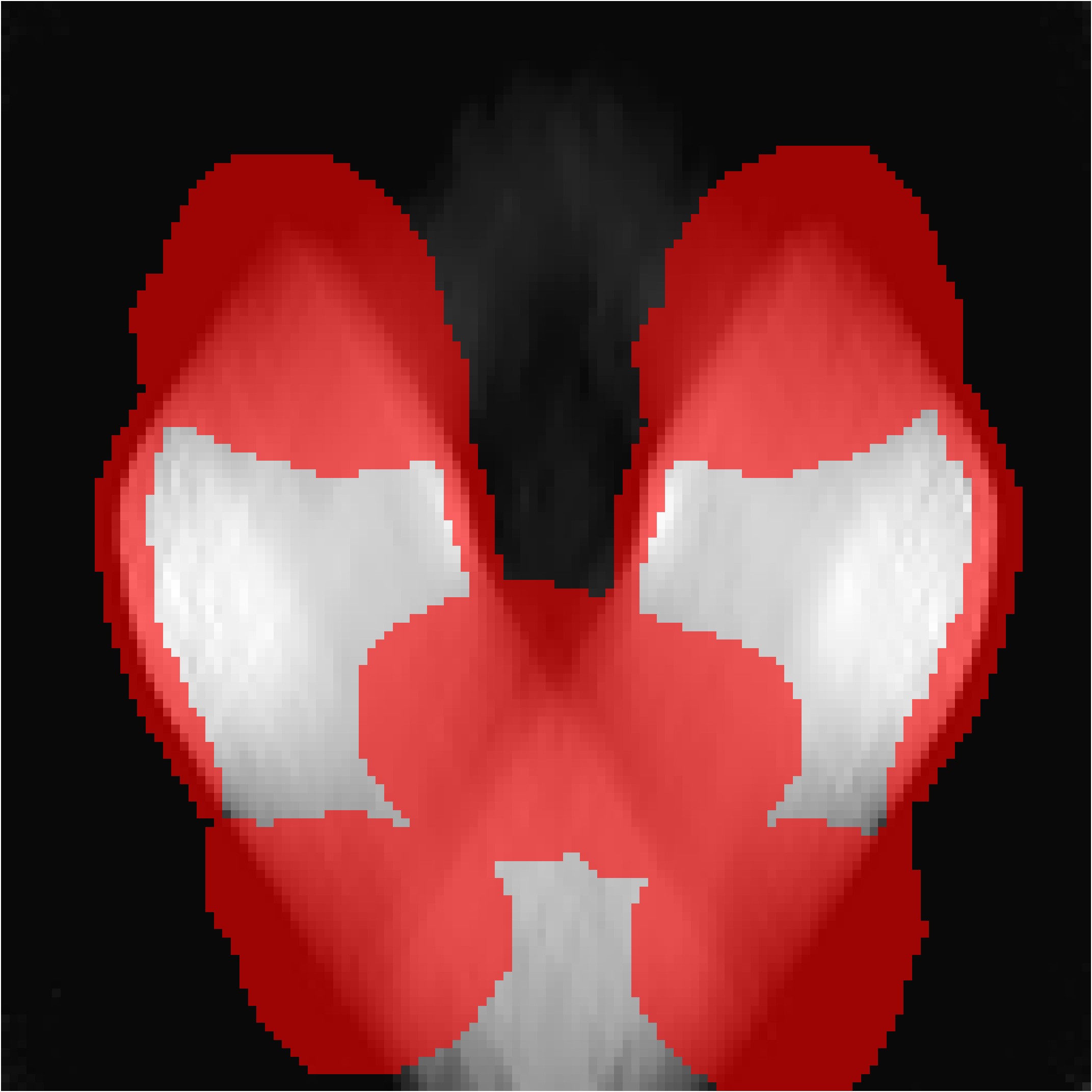} 
    \caption{Boundary neighborhood estimations over bicubic interpolated initial reconstruction.}
    \end{subfigure}
            \hspace{2mm}
\begin{subfigure}[t]{0.6  \textwidth}
    \includegraphics[width=\textwidth]{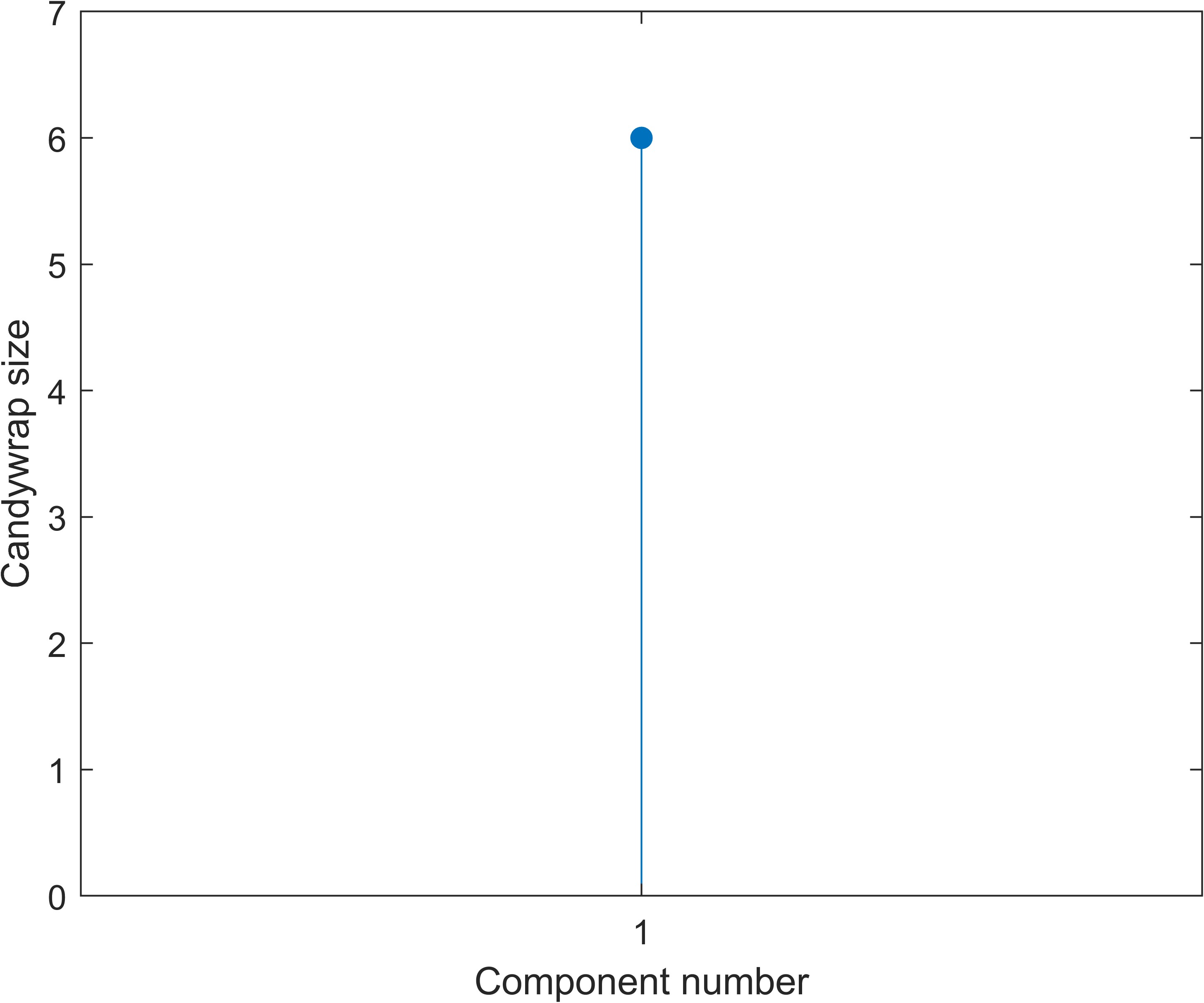}
    \caption{Number of found components and birth values of those, i.e., the size of the candywrap masks}
    \end{subfigure}
    \caption{A target phantom with non-convex inclusion with too much curvature for the proposed method. The bottom part of the missing boundary is not fully covered by neighborhood because the middle part neighborhoods and the bottom neighborhoods merge together when using TILT. }
    \label{fig:high_curvature_results}
\end{figure}


\section{Discussion} \label{sec:discussion}

In this paper, we propose a novelty way to use persistent homology with non-euclidean distance on the highly ill-posed inverse problem of limited-angle tomography. This study is inspired by the need to separate boundaries of structures, which is a demanding task in limited-angle cases. That is because only part of the wavefront set can be recovered stably. In our method, we estimate neighborhoods of invisible singularities, which give uncertain areas where the boundaries lie. We emphasize that each step in our method is theoretically justified. Also, every step's results are traceable, and hence the final results can be explained rigorously, which is crucial in certain applications.   	

If only a few angles are missing, and we are dealing with one simply connected inclusion with smooth boundary of moderate curvature, our theoretical background guarantees that we can estimate the locations of invisible singularities reliably. In our proposed computational method and experiments, we extended these assumptions by experimenting with the limits of the method. We consider situations when only a 60-degree angle is known and there is an unknown number of inclusion that may not be simply connected. As results showed, our proposed method could find a correct number of boundary components in each mentioned example. Phantom (iv) was too difficult for our method in the sense that the predicted neighborhoods overlapped in an unexpected way; see Figure \ref{fig:high_curvature_results}(E).

The main strength of our method is the concrete computational representation for the recovered boundary. It is especially useful in applications where the number of defects or inclusions is of interest. We are not aware of other methods offering this benefit. 

We stress that the TILT-based recovery of the unstable part may not be perfectly overlapping with the true boundary. This applies to any limited-angle tomography reconstruction method. The unstable part is recovered based on {\it a priori} knowledge imposed by the algorithm; in the case of TILT we assume that the boundary is smooth and not very curved. 

We originally approached the problem by computing the first homology groups, coming from the fact that each boundary component is diffeomorphic to $\bS ^1$. This approach leads to an enormous 3D matrix. The observation that exactly the same results can be achieved using the zeroth homology groups helped us speed up computations remarkably. The old approach took around one hundred times longer than the new way.

Since this work is just an introduction and initial feasibility study of a completely new method, there is plenty of room for future experiments and improvements. In this study, we do not investigate different initial reconstructions' impact on final estimations, namely whether complex wavelets perform better in extracting stable singularities than other reconstruction algorithms. Also, we noticed that complex wavelets are sensitive to noise, meaning that more noise leads to incorrectly classified singularities. Moreover, any kind of imperfection in the subbands inevitably leads to inaccurate final neighborhood estimations. Thus it might be worth exploring curvelets or shearlets to extract singularities from the reconstruction.       
\clearpage

\section*{Acknowledgments}
Elli Karvonen was funded by the Emil Aaltonen Foundation and Academy of Finland grant, decision number 332671.
We thank Kristian Bredies for sharing his TV regularization codes. Matti Lassas and Samuli Siltanen are supported by the Finnish Centre of Excellence in Inverse Modelling and Imaging, 2018-2025, decision numbers 312339 and 336797, Academy of Finland grants 284715 and  312110. Pekka Pankka is supported by the Academy of Finland grant, decision number 332671.

\bibliographystyle{abbrv}
\bibliography{main}

\end{document}